\documentclass{IEEEtaes}
\usepackage{amsmath,amsfonts}
\usepackage{algorithmic}
\usepackage{algorithm}
\usepackage{lineno,hyperref}
\usepackage{multirow}
\usepackage{epstopdf}
\usepackage{graphicx}
\usepackage{float}
\usepackage{caption}
\usepackage{stfloats}
\usepackage{changes}
\usepackage{color,array,amsthm}
\usepackage{amssymb}
\usepackage{graphicx}
\usepackage{eucal}
\usepackage{mathrsfs}
\usepackage{hyperref}      
\jvol{XX}
\jnum{XX}
\jmonth{XXXXX}
\paper{1234567}
\pubyear{2023}
\doiinfo{TAES.2023.Doi Number}

\newtheorem{theorem}{Theorem}

\newtheorem{definition}{Definition}
\newtheorem{remark}{Remark}
\newtheorem{assumption}{Assumption}
\usepackage[inkscapelatex=false]{svg}
\usepackage[numbers,sort&compress]{natbib}
\usepackage{soul} 
\usepackage{color, xcolor} 

\soulregister{\cite}7 
\soulregister{\citep}7 
\soulregister{\citet}7 
\soulregister{\ref}7 
\soulregister{\pageref}7 
\modulolinenumbers[5]

\begin{document}

\title{Adaptive Reduced-Attitude Control for Spacecraft Boresight Alignment with Safety Constraints and Accuracy Requirements}

\author{Jiakun Lei}
\affil{Zhejiang University, China} 

\author{Tao Meng}
\affil{Zhejiang University, China\\ Hainan Research Institute of Zhejiang University, China} 

\author{Kun Wang}
\affil{Zhejiang University, China}

\author{Weijia Wang}
\affil{Zhejiang University, China}

\author{Shujian Sun}
\affil{Zhejiang University, China}

\author{Lei Wang}
\affil{Zhejiang University, China}



\authoraddress{Jiakun Lei, Kun Wang, Weijia Wang, Shujian Sun are with the School of Aeronautics and Astronautics, Zhejiang University, Hangzhou 310027, China (Email: leijiakun@zju.edu.cn; wang\_kun@zju.edu.cn; weijiawang@zju.edu.cn; sunshujian@zju.edu.cn). Tao Meng(Corresponding Author) is with the School of Aeronautics and Astronautics, Zhejiang University, Hangzhou 310027, China, and Hainan Reserach Institute of Zhejiang University, Sanya, 572025, China (Email: mengtao@zju.edu.cn). Lei Wang is with the College of Control Science and Engineering, Zhejiang University, Hangzhou, 310027, China (Email: lei.wangzju,zju.edu.cn).}

\markboth{Lei ET AL.}{Preprint Submitted To IEEE Transactions On Aerospace and Electronic Systems}
\maketitle

	\begin{abstract}                
	This paper investigates the boresight alignment control problem under safety constraints and performance requirements, involving pointing-forbidden constraint, attitude angular velocity limitation, and pointing accuracy requirement. Meanwhile, the parameter uncertainty issue is taken into account simultaneously.
To address this problem, we propose a modified composite framework integrating the Artificial Potential Field (APF) methodology and the Prescribed Performance Control (PPC) scheme. The APF scheme ensures safety, while the PPC scheme is employed to realize an accuracy-guaranteed control. A Switched Prescribed Performance Function (SPPF) is proposed to facilitate the integration, which monitors various constraints and further establishes compatibility between safety and performance concerns by leveraging a special PPC freezing mechanism.
To further address the parameter uncertainty, we introduce the Immersion-and-Invariance (I\&I) adaptive control technique to derive an adaptive APF-PPC composite controller, further guaranteeing the closed-loop system's asymptotic convergence.
Finally, numerical simulations are carried out to validate the effectiveness of the proposed scheme.
\end{abstract}

%

\section{INTRODUCTION}\label{INTRO} 

On-orbit imaging and observation play a significant role in contemporary space missions, in which the key process is to reorient the payload sensor's boresight vector to provide a stable staring, facilitating the observation task. This requisition significantly raises the research interest in the Boresight Alignment Control problem, where the spacecraft is required to adjust its attitude and align the body-fixed sensor's boresight vector to a desired direction \cite{wie2002rapid,tanygin2017fast}. 
This task is challenging due to bright celestial bodies, such as the sun, which can interfere with sophisticated sensors \cite{hablani1999attitude}.
Meanwhile, the attitude transition process must not be too aggressive, or it may cause the startracker and the low-frequency gyro to malfunction \cite{kjellberg2013discretized}. Additionally, contemporary spacecraft missions often require high control accuracy for advanced tasks, which demand guaranteed pointing accuracy \cite{zhou2014high}.
Motivated by this topic, this paper investigates the boresight alignment control problem of rigid-body spacecraft with multiple constraints, including pointing-forbidden constraint, attitude angular velocity limitation, and pointing accuracy requirement. In order to facilitate the potential practical utilization of the proposed control scheme, this paper considers the issue of parameter uncertainty simultaneously.

The fundamental control objective of the boresight alignment control can be realized by providing a desired full-attitude representation, usually represented by the Direction Cosine Matrix (DCM) \cite{zlotnik2014rotation,forbes2015direction} or other parameterized attitude representations, such as the unit quaternion and the Modified Rodrigues Parameters (MRPs) \cite{yang2012spacecraft,bandyopadhyay2015attitude}. However, note that the control objective of the  boresight alignment control requires aligning the pointing vector to a specific direction. Since the rotation on the boresight vector is irrelevant, using a full-attitude representation necessitates additional constraints to construct a desired attitude, which leads to redundancy and conservatism problem. Also, typical unit quaternion-based representation suffers from the so-called unwinding problem due to the double-covered topological construction of $SU_{2}$ \cite{chaturvedi2011rigid}.
Due to these reasons, the reduced-attitude representation has come to the center of this issue, which circumvents these drawbacks. The reduced-attitude representation only concentrates on the pointing direction of the concerned sensor and naturally excludes the possible redundancy problem and the unwinding problem, as discussed in \cite{chaturvedi2011rigid}. The reduced-attitude representation was first discovered and utilized for attitude control problem in \cite{bullo1995control}, and has been applied to many spacecraft attitude control scenarios, as investigated in \cite{shao2022fault,chi_reduced,dongare2021attitude}.

Meanwhile, as investigated by many pioneering works, the constrained attitude control problem can be mainly addressed through planning-based and non-planning-based approaches. Since mostly planning-based methodologies require an obstacle-free path calculating beforehand \cite{feron2012randomized,xu2017rapid,sun2015spacecraft}, it may be computational-costly for spacecraft with limited computing ability. Meanwhile, non-planning-based methodologies, such as the Artificial Potential Field (APF), provide less calculating complexity. As a result of its efficiency, the APF methodology has been combined with the reduced-attitude representation to develop many efficient controller frameworks for the boresight alignment control problem with safety constraints, as investigated in  \cite{chi_reduced,hu2019reduced,coates2020reduced}. 
Essentially, these APF-based reduced-attitude controllers utilize the gradient of the attraction field to guide the system toward the equilibrium region, indicating that a higher pointing accuracy may require a higher gain of the attraction field. However, a trade-off should be made to balance the attraction and the repulsion field to guarantee safety, indicating that the gain of the attraction field cannot be arbitrarily large. This will significantly limit its capability to maintain a high control accuracy.
With the aforementioned necessary trade-off, it is hard to obtain an appropriate parameter set to obtain a safe and high-accuracy alignment control for all possible control scenarios under external disturbance. On the other hand, typical APF controllers are not designed with performance considerations in mind, indicating that their control accuracy cannot be guaranteed beforehand.

Fortunately, the Prescribed Performance Control (PPC) scheme has been proposed to address the performance issue in advance, as stated in \cite{bechlioulis2011robust}. PPC control schemes utilize parameterized functions, such as the exponential function \cite{bechlioulis2011robust}, the polynomial function \cite{liu2019appointed}, or a function designed with fixed-time stability \cite{gao2021finite}, to characterize the performance criteria. Then, using a homeomorphic error transformation or the Barrier Lyapunov Function technique, it translates the original constrained-error state into an equivalently unconstrained one. As a result, by guaranteeing the convergence (boundedness) of the translated error, or the translated-error-defined Barrier Lyapunov Function (BLF), the state trajectory can be restricted in the performance envelope, and the original error state can be made to converge while satisfying the specified performance requirement. Due to its ability, the PPC scheme has been applied to attitude control scenarios in many existing works, as investigated in \cite{bu2023prescribed,wei2018learning,liu2019appointed,lei2023singularity}.

For the parameter uncertainty issue, the Immersion-and-Invariance (I\&I) adaptive has raised much attention as they provide advanced performance for closed-loop systems \cite{shao2021data,gao2018immersion}. As highlighted in \cite{shao2021data}, typical certain-equivalency-based (CE-based) adaptive controllers can only guarantee the boundedness of their resulting estimation error, and no guarantee is established for the convergence of the parameter estimation error or estimation functional error. Hence, the system failed to recover to the original ideal case, leading to performance degradation. The I\&I methodology, however, tactfully constructs a regulation function to drive the whole system to a manifold in which it will be immersed, providing extraordinary robustness \cite{seo2009non,astolfi2003immersion}. As a result, the I\&I adaptive control has shown its capability in many nonlinear control problems \cite{wang2019adaptive,wang2016immersion}, especially spacecraft control problems \cite{xia2022anti,shao2021immersion,shao2021data}, and may be suitable to address the parameter uncertainty issue for our discussed problem.

As stated before, the to-be-discussed problem involves multiple constraints, which can be mainly classified into two types: safety-oriented constraints, including the pointing-forbidden constraint and the attitude angular velocity limitation, and the performance requirement.
As summarized, previous researches have endeavored to address these two types of constraints independently: the APF control efficiently handles the safety-oriented constraints \cite{shen2018rigid,shao2021immersion,chi_reduced,dongare2021attitude}, while the PPC scheme effectively realizes the performance-guaranteed control \cite{bechlioulis2011robust,hu2018adaptive,wei2018learning,wei2021overview}. However, existing literature related to the APF method does not consider the performance issue, while no existing literature about PPC control makes an effort on the issue of "PPC control with safety constraints."  Therefore, from a macro perspective, the discussed problem remains open as no solution has been provided. This point of view is also bolstered by existing literature, as highlighted in \cite{wei2021overview,HUZHANWANG}. On the other hand, a natural idea is to integrate the APF methodology with the PPC scheme. However, we will later highlight in Subsection \ref{formulation}-\ref{PROBLEM}.\ref{motive} that such an integration is challenging and nontrivial, necessitating significant modification to derive a feasible solution. 
Due to these reasons, to the best of the author's knowledge, the boresight alignment control problem involving safety constraints and control accuracy requirements still remains open. Correspondingly, main contributions of this paper are stated as follows:

\begin{itemize}
\item[1)] A composite APF-PPC control framework is first proposed for addressing the boresight alignment control problem under pointing-forbidden constraint, attitude angular velocity limitation and pointing-accuracy requirement.

\item[2)] With the awareness of the so-called \textit{contradicted constraint} issue, this paper presents the concept of the Switched Prescribed Performance Function (SPPF). The SPPF delivers the idea of monitoring the potential factor that may lead to the contradiction between multiple constraints through the design of the switching indicator. Then, it will further switch between different dynamics to prior the safety concern and make a temporarily concession on the performance issue, thereby establishing compatibility.

\item[3)] Specifically, we present a freezing mechanism to establish the compatibility, which freezes the value of the PPC system and hence temporarily excludes the impact of the PPC control scheme on the closed-loop system when contradiction exists. 
\end{itemize}

\textbf{Notation.} The following notations are used in this paper. The Euclidean norm of any given vector or the induced-norm of any given matrix is denoted by the symbol $\|\cdot\|$. Given any vector $\boldsymbol{a}\in\mathbb{R}^{3}$, the symbol $\boldsymbol{a}^{\times}$ represents the $\mathbb{R}^{3\times 3}$ skew-symmetric matrix for the cross-product manipulation. $\text{diag}(a_{i})$ represents the diagonal matrix, of which the diagonal elements are $a_{1},...a_{N}$, $\max(\cdot)$ denotes the maximum function that outputs the maxima of input signals.
 $\boldsymbol{I}_{3}\in\mathbb{R}^{3\times 3}$ stands for the identity $\mathbb{R}^{3\times 3}$ matrix, and $\boldsymbol{0}_{3}\in\mathbb{R}^{3}$ denotes a column vector whose elements are $0$.
For the coordinate definition, the symbol $\mathfrak{R}_{b}$ denotes the spacecraft body-fixed frame, while $\mathfrak{R}_{i}$ denotes the earth-central inertial frame.

\section{PROBLEM FORMULATION}\label{formulation}
\subsection{Introduction of the Reduced-Attitude System}
As widely discussed in existing literature \cite{chaturvedi2011rigid}, the full-attitude representation can be specified by a rotation matrix that transforms between different coordinates, defined on the special orthogonal group: $SO(3): \left\{\boldsymbol{A}\in\mathbb{R}^{3\times 3}| \boldsymbol{A}^{\text{T}}\boldsymbol{A} = \boldsymbol{I}_{3}, \text{det}\left(\boldsymbol{A}\right) = 1\right\}$.
Based on this concept, we define $\boldsymbol{A}_{bi}\in SO(3)$ as the rotation matrix that transforms from the inertia frame $\mathfrak{R}_{i}$ to the body-fixed frame $\mathfrak{R}_{b}$, and let $\boldsymbol{\omega}_{s}\in\mathbb{R}^{3}$ be the attitude angular velocity of the spacecraft's body-fixed frame $\mathfrak{R}_{b}$ with respect to the inertia frame $\mathfrak{R}_{i}$, expressed in $\mathfrak{R}_{b}$. The kinematics equation of the full-attitude representation can be then given as \cite{chaturvedi2011rigid}: 
\begin{equation}
	\dot{\boldsymbol{A}}_{bi} = -\boldsymbol{\omega}^{\times}_{s}\boldsymbol{A}_{bi}
\end{equation}

Different from the full-attitude representation, the construction space of the reduced-attitude representation shrinks to a 3-dimensional unit sphere: $\mathbb{S}_{2} \triangleq \left\{\boldsymbol{x}\in\mathbb{R}^{3}| \|\boldsymbol{x}\|^{2} = 1\right\}$. Let $\boldsymbol{B}_{b}\in \mathbb{S}_{2}$ be the boresight vector of a given sensor, resolved in $\mathfrak{R}_{b}$, and denote the boresight vector resolved in $\mathfrak{R}_{i}$ as $\boldsymbol{B}_{i}\in\mathbb{S}_{2}$, then the kinematics equation of the reduced-attitude model can be derived by considering the time-derivative of $\boldsymbol{B}_{i}$, expressed as follows \cite{bullo1995control}:
\begin{equation}
	\begin{aligned}
			\dot{\boldsymbol{B}}_{i} &= -\left(\boldsymbol{B}^{\times}_{i}\boldsymbol{A}^{\text{T}}_{bi}\right)\boldsymbol{\omega}_{s}\\
	\end{aligned}
\end{equation}

As highlighted in Section \ref{INTRO}, the control objective of the boresight alignment control requires $\boldsymbol{B}_{i}$ reorient to another target direction. Correspondingly, we denote $\boldsymbol{r}_{i}\in\mathbb{S}_{2}$ as the target pointing-direction vector resolved in $\mathfrak{R}_{i}$, and further denote its expression given in the body-fixed frame $\mathfrak{R}_{b}$ as $\boldsymbol{r}_{b}\in\mathbb{S}_{2}$. We then denote the angle between $\boldsymbol{B}_{b}$ and $\boldsymbol{r}_{b}$ as $\Theta(t)\in[0,\pi]$ such that $\cos\Theta = \boldsymbol{B}^{\text{T}}_{b}\boldsymbol{r}_{b}$ holds. Accordingly, define the following variable $x_{e}(t)$:
\begin{equation}
	x_{e}(t) \triangleq 1-\cos\Theta =  1 - \boldsymbol{B}^{\text{T}}_{b}\boldsymbol{r}_{b}
\end{equation}
it can be observed that $x_{e}(t)\in\left[0,2\right]$ holds, and $x_{e} = 0$ is a stable equilibrium point \cite{chaturvedi2011rigid}, equivalent to the circumstance that $\Theta=0$. Therefore, $x_{e}(t)$ can be regarded as a pointing error state.
Notably, this paper discusses on the circumstance that the target pointing direction is static, i.e., $\dot{\boldsymbol{r}}_{i} = \boldsymbol{0}_{3}$, thus the reduced-attitude error model can be formulated as follows \cite{chi_reduced}:
\begin{equation}\label{errorsystem}
	\begin{aligned}
			\dot{x}_{e} &= \left(\boldsymbol{r}^{\times}_{b}\boldsymbol{B}_{b}\right)^{\text{T}}\boldsymbol{\omega}_{s}\\ \boldsymbol{J}\dot{\boldsymbol{\omega}}_{s} &= -\boldsymbol{\omega}^{\times}_{s}\boldsymbol{J}\boldsymbol{\omega}_{s}+\boldsymbol{u}
	\end{aligned}
\end{equation}
where $\boldsymbol{J}\in\mathbb{R}^{3\times 3}$ stands for the inertia matrix, $\boldsymbol{u}\in\mathbb{R}^{3}$ stands for the control input.
Before going further, we first present the following assumption about the inertial matrix $\boldsymbol{J}$.
\begin{assumption} (Unknown Diagonal Inertia Matrix)\label{AssumpJ}
	The inertia matrix $\boldsymbol{J}\in\mathbb{R}^{3\times 3}$ is assumed to be an unknown \textbf{diagonal} positive-definite matrix such that $J_{ij}=0$ holds for $\forall i\neq j$.
\end{assumption}

Subsequently, we introduce the linear operator $\mathcal{L}(\cdot):\mathbb{R}^{3}\to\mathbb{R}^{3\times3}$ stated in \cite{shao2021immersion}, which ensures that $\boldsymbol{J}\boldsymbol{x} = \mathcal{L}(\boldsymbol{x})\boldsymbol{\theta}$ holds for arbitrary $\boldsymbol{x}\in\mathbb{R}^{3}$, where $\boldsymbol{\theta} \triangleq \left[J_{11},J_{22},J_{33}\right]^{\text{T}}\in\mathbb{R}^{3}$ stands for the estimation vector of unknown inertial parameters.
By applying the linear operator $\mathcal{L}(\cdot)$, we further define a regression matrix as $\boldsymbol{\phi}_{J}\triangleq -\boldsymbol{\omega}^{\times}_{s}\mathcal{L}(\boldsymbol{\omega}_{s})$. Accordingly, the system's dynamics equation (\ref{errorsystem}) can be rearranged into the following compact linear-plant form as:
\begin{equation}\label{sysdyna}
	\dot{\boldsymbol{\omega}}_{s} = \boldsymbol{J}^{-1}\left[\boldsymbol{\phi}_{J}\boldsymbol{\theta}+\boldsymbol{u}\right]
\end{equation} 

\subsection{Pointing-Forbidden Constraint}
As stated in existing literature \cite{chi_reduced,shao2022fault}, bright celestial bodies always exist that the sensor's boresight vector should circumvent.
Suppose that there exist $m$ forbidden directions, 
let $\boldsymbol{f}^{N}_{i}\in\mathbb{S}_{2}(N = 1,2,3...m)$ be the $N$-th forbidden direction, resolved in $\mathfrak{R}_{i}$, then its resulting expression in $\mathfrak{R}_{b}$ can be denoted and given as $\boldsymbol{f}^{N}_{b} \triangleq \boldsymbol{A}_{bi}\boldsymbol{f}^{N}_{i}$. Let $\Theta^{N}_{f}$ be the permitted minimum angle that corresponds to the $N$-th constraint, then each pointing-forbidden constraint can be formulated as \cite{chi_reduced}:
\begin{equation}\label{pointcons}
 \boldsymbol{B}^{\text{T}}_{i}\boldsymbol{f}^{N}_{i} = \boldsymbol{B}^{\text{T}}_{b}\boldsymbol{f}^{N}_{b} < \cos\Theta^{N}_{f}
\end{equation}
\begin{figure}[hbt!]
	\centering 
	\includegraphics[width=0.5\textwidth]{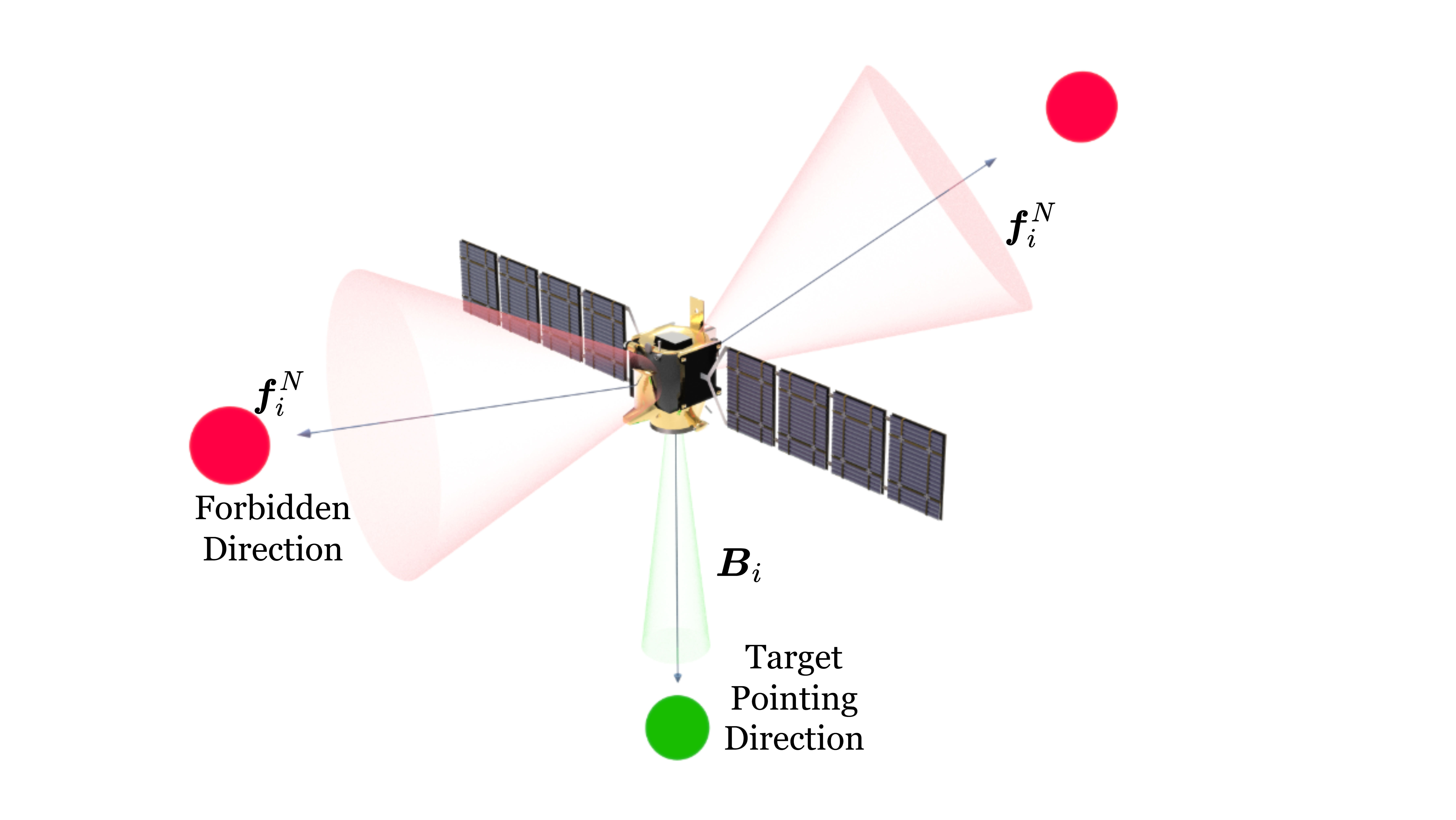}
	\caption{Sketch-map of the pointing-forbidden constraint. The green object represents the target direction, while red objects represent forbidden directions.}      
	\label{Sat}   
\end{figure}

Based on this definition, we define a pointing-forbidden-constraint-satisfying set $\mathcal{S}_{p}$, expressed as follows:
\begin{equation}
	\mathcal{S}_{p} \triangleq \left\{\boldsymbol{B}_{i}\in\mathbb{S}_{2}|\boldsymbol{B}^{\text{T}}_{i}\boldsymbol{f}^{N}_{i} < \cos\Theta^{N}_{f}, \text{for}\quad \forall N = 1,2,3...,m\right\}
\end{equation}

	Notably, we assume that the initial condition of the pointing direction, denoted as $\boldsymbol{B}_{i}(t_{0})$, and the desired pointing direction $\boldsymbol{r}_{i}$, are outside any forbidden zone, i.e., $\boldsymbol{B}_{i}(t_{0}), \boldsymbol{r}_{i}\in\left\{\boldsymbol{r}^{\text{T}}_{i}\boldsymbol{f}^{N}_{i}<\cos\Theta^{N}_{f},\boldsymbol{B}^{\text{T}}_{i}(t_{0})\boldsymbol{f}^{N}_{i}<\cos\Theta^{N}_{f},\quad\forall N = 1,2,3..m\right\}$.
Such statement ensures that the desired pointing direction is admissible, which can be found in many existing literature, such as \cite{shao2022fault}.

\subsection{Attitude Angular Velocity Constraint}
In order to ensure that the attitude's information can be obtained normally, the proper functioning of the startracker should be guaranteed. This possesses additional requirement that the magnitude of each axis's angular velocity should be remained under a given limitation. Let $M_{\omega}>0$ be the corresponding upper bound, this constraint can be formulated as follows \cite{chi_reduced}:  
\begin{equation}\label{angcons}
	|\omega_{si}(t)| < M_{\omega},i=1,2,3
\end{equation} 
where $\omega_{si}(t)$ stands for the angular velocity of the $i$-th axis of the spacecraft. Accordingly, the angular velocity-constraint-satisfying set can be defined as: $\mathcal{S}_{\omega} \triangleq \left\{\boldsymbol{\omega}_{s}\in\mathbb{R}^{3}||\omega_{si}(t)|<M_{\omega}, \forall i=1,2,3\right\}$

\subsection{Performance Function Envelope (PFE) Constraint}
This paper employs the main structure of the PPC control scheme to achieve an accuracy-guaranteed boresight alignment control. As we introduced in Section \ref{INTRO}, the PPC control scheme using parameterized functions \cite{wei2021overview} to characterize the performance requirement. Then, the desired performance criteria can be achieved by keeping the state trajectory restricted in the Performance Function Envelope (PFE) consistently. 
Based on this philosophy, let $\rho_{q}(t)>0$ be the performance function designed for $x_{e}(t)$, then the PFE constraint can be formulated as follows \cite{bechlioulis2011robust}:
\begin{equation}\label{pfe}
	x_{e}(t)\in\left[0,\rho_{q}(t)\right)
\end{equation}
The specific design of the performance function $\rho_{q}(t)$ will be provided in Subsection \ref{APFPPC}-\ref{SPPF} later. Correspondingly, one can define a PFE constraint-satisfying set as: $\mathcal{S}_{e} \triangleq \left\{x_{e}(t)|x_{e}(t)\in\left[0,\rho_{q}(t)\right)\right\}$.

\subsection{Control Problem Statement and Major Challenge}\label{PROBLEM}
\subsubsection{Problem Statement} This paper aims to develop a control law $\boldsymbol{u}$, such that for the reduced-attitude error system specified in equation (\ref{errorsystem}), the pointing direction of the boresight axis $\boldsymbol{B}_{i}$ (expressed in $\mathfrak{R}_{i}$) will be reoriented to the target direction $\boldsymbol{r}_{i}$, even facing parameter uncertainties. Meanwhile, the pointing-forbidden constraint (\ref{pointcons}), the attitude angular velocity limitation (\ref{angcons}) and the PFE constraint (\ref{pfe}) will be respected during the whole control process, and the system will be finally guided to achieve the desired performance criteria.

\subsubsection{Major Challenge: Potentially Contradicted Constraints}\label{motive}
As we briefly mentioned in Section \ref{INTRO}, 
due to a potential contradiction between safety-oriented constraints and the PFE constraint, directly integrating APF method with existing PPC frameworks is not a feasible solution for the to-be discussed multiple-constraint problem.

To be specific, when the boresight vector $\boldsymbol{B}_{i}$ needs to avoid a pointing-forbidden direction $\boldsymbol{f}^{N}_{i}$, the convergence rate of $x_{e}$ will be slow down, as $\boldsymbol{B}_{i}$ needs to bypass a circular region on $\mathbb{S}_{2}$. Meanwhile, if there exists angular velocity limitation, the convergence rate of $x_{e}$ will also be limited by the allowed maximum velocity bound. However, typical PPC schemes utilize rapidly-converged functions to form the performance envelope and to bound the system's state trajectory, as done in existing works \cite{bu2023prescribed,wei2018learning,liu2019appointed,wei2021overview,hu2018adaptive}. Therefore, the $x_{e}$-trajectory will be forced to converge rapidly to keep it within the performance envelope. This, in turn, contradicts the requisition yielded by the pointing-forbidden constraint and the angular velocity limitation, where the slow convergence of $x_{e}$-trajectory is necessitated.

\begin{figure}[hbt!]
	\centering 
	\includegraphics[width=0.45\textwidth]{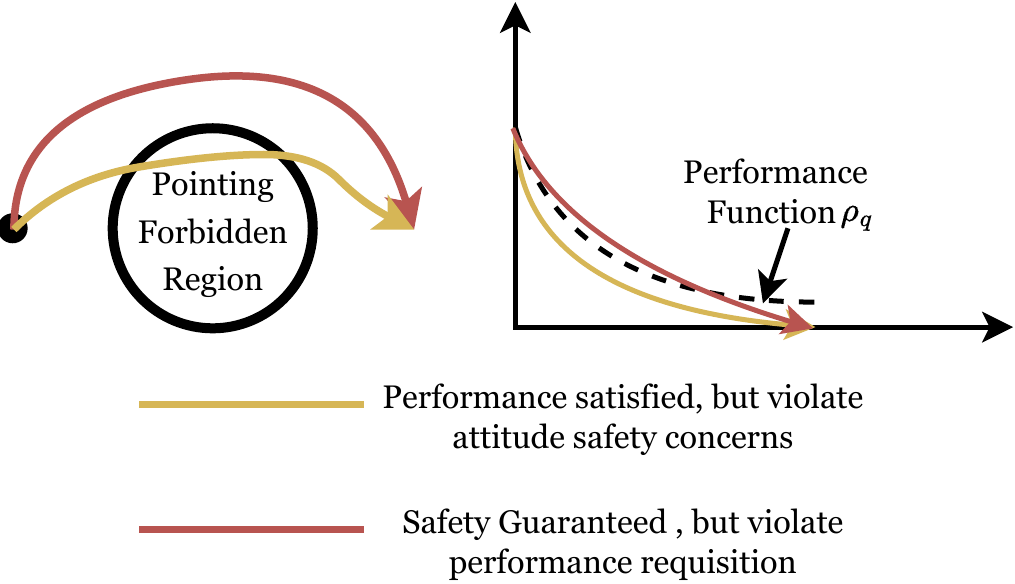}
	\caption{Sketch map of the intrinsic conflict between safety constraints and PFE constraint}      
	\label{conflict}   
\end{figure}

As shown in Figure \ref{conflict}, if the convergence rate of $x_{e}$-trajectory is relatively slow, there arises a potential risk of violating the PFE constraint, which may further result in a so-called singularity problem of the PPC control scheme, as mentioned by \cite{yong2020flexible}. In contrast, if the $x_{e}$-trajectory is forced to remain within the performance envelope, it will lead to an over-speed convergence behavior and further violations of other safety constraints.

The above analysis highlights that these potentially conflicting constraints introduce intricate challenges that cannot be readily addressed using existing frameworks, making the whole problem nontrivial, therefore, necessitating modifications to derive a feasible solution.

\section{APF-PPC CONTROL SCHEME}\label{APFPPC}

\subsection{Core Logic of the APF-PPC Control Scheme}
This paper presents a control scheme integrating the Artificial Potential Field (APF) method with the Prescribed Performance Control (PPC) scheme. The basis of such an integration is that both potential functions utilized for ensuring general safety constraints and the Barrier Lyapunov Function (BLF) designed for performance consideration are constructed based on a "barrier condition, as emphasized in \cite{ames2019control,tee2009barrier}, which transforms the satisfaction of given constraints to be equivalent with the boundedness of their resulting Lyapunov candidates. Therefore, this perspective establishes a vital connection between these barrier-based functions.

To further address the potential contradiction between safety constraints and the PFE constraint as discussed in Subsection \ref{formulation}-\ref{PROBLEM}.\ref{motive}, we present a perspective that safety concerns should be prioritized, while performance requirement can only be fulfilled with no further contradiction. Therefore, this motivates us that a concession on the performance requirement should be temporarily made if such a contradiction exists, and the consideration of performance issue should be recovered if the contradiction no longer exists.

Based on this core logic, we propose the Switched Prescribed Performance Function (SPPF), detailed in Subsection \ref{APFPPC}-\ref{SPPF}. By monitoring and detecting the contradiction between multiple constraints, the SPPF switches between different dynamics to realize the aforementioned "temporarily concession on performance" and the "recover of performance issue", thereby effectively mitigating the contradiction through a so-called "PPC Freezing" mechanism. An overall structure of the proposed APF-PPC control scheme is depicted in Figure \ref{overallstructure}.
\begin{figure}[hbt!]
	\centering 
	\includegraphics[width=0.45\textwidth]{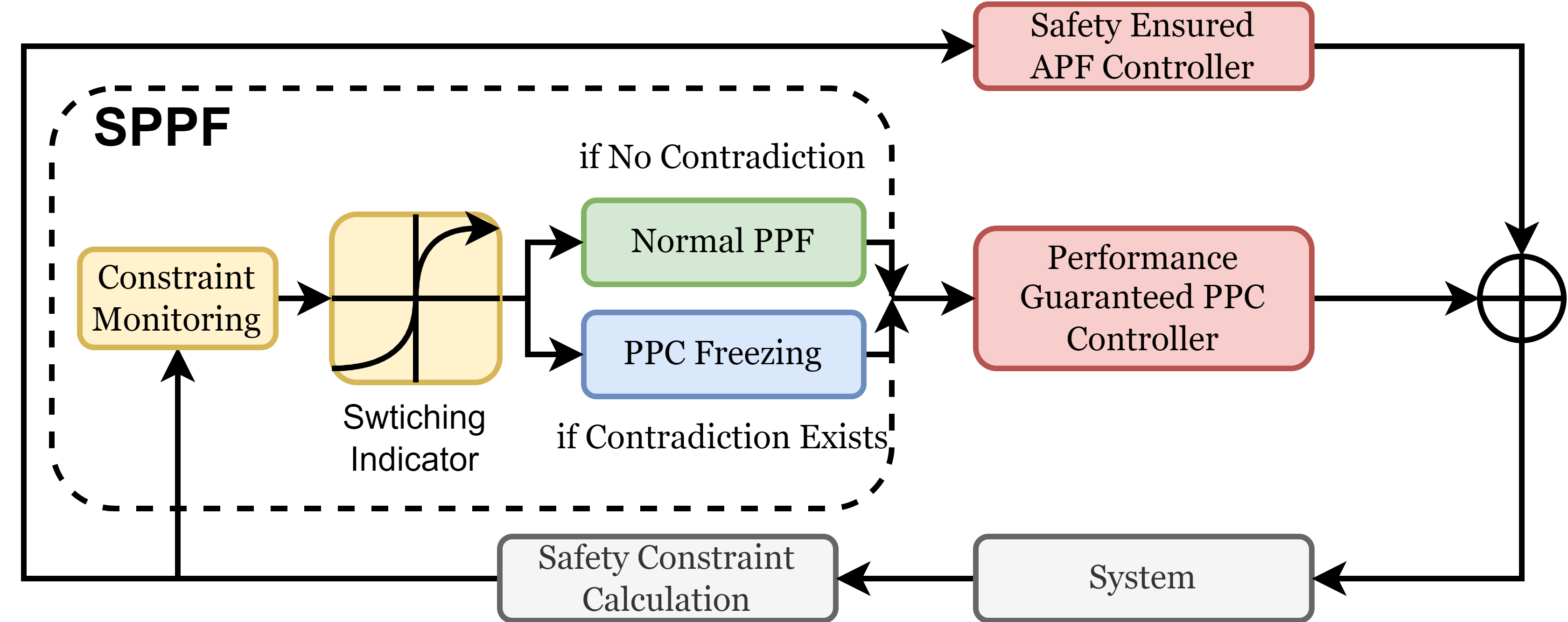}
	\caption{Overall Structure of the Proposed APF-PPC Control Scheme}       
	\label{overallstructure}   
\end{figure}	

The following of this section is organized as follows: Subsection \ref{APFPPC}-\ref{APFDESIGN} presents the fundamental design of potential functions that are designed for handling safety constraints, while Subsection \ref{APFPPC}-\ref{BLFDESIGN} delivers the BLF that is constructed following the PPC control scheme. Subsequently, the proposed SPPF is detailed in Subsection \ref{APFPPC}-\ref{SPPF}, along with a thorough theoretic analysis of its effect and the freezing mechanism, presented in Subsection \ref{APFPPC}-\ref{SPPFMECHANISM}.

\subsection{Potential Function Design for APF}\label{APFDESIGN}
\subsubsection{Potential Function For Pointing Constraint}
We first introduce an error function associated with $x_{e}(t)$ to ensure the convergence of the pointing error angle $\Theta(t)$. The attraction function $U_{a}(t)$ is designed as follows \cite{dongare2021attitude}:
\begin{equation}
	U_{a}(t) \triangleq k_{a}x_{e}(t)
\end{equation}
with $k_{a}>0$ the constant gain of the attraction field $U_{a}(t)$.

Subsequently, in order to ensure the satisfaction of the pointing-forbidden constraint, a repulsion field is designed for each $N$-th forbidden zone, denoted as $U^{N}_{r}(t)(N=1,2,3..m)$. Motivated by \cite{dongare2021attitude}, we design a repulsion field with an explicit acting region, expressed as follows:
\begin{equation}
\begin{aligned}
	&U^{N}_{r}(t) \triangleq \\
	&\begin{cases}
	  1,\quad\text{if} \gamma_{N}(t) \in\left[-1,P^{N}_{0}\right)\\
	  k_{r}\sec\left(a_{N}\gamma_{N}(t)+b_{N}\right)+1-k_{r},\quad \text{if}\gamma_{N}(t)\in\left[P^{N}_{0},P^{N}_{1}\right)
	\end{cases}
\end{aligned}
\end{equation}
where $\gamma_{N}(t)$ is defined as $\gamma_{N}(t) \triangleq \boldsymbol{B}^{\text{T}}_{i}\boldsymbol{f}^{N}_{i} = \boldsymbol{B}^{\text{T}}_{b}\boldsymbol{f}^{N}_{b}$, $a_{N}$ and $b_{N}$ are design parameters given as $a_{N} = \frac{\pi}{2(P^{N}_{1}-P^{N}_{0})}$ and $b_{N} = -a_{N}P^{N}_{0}$, $P^{N}_{0}\in\left(-1,P^{N}_{1}\right)$ stands for the segment point where the repulsion field start works, while $P^{N}_{1} = \cos\Theta^{N}_{f}$ stands for the boundary of $U^{N}_{r}(t)$, corresponding to each $N$-th constraint. $k_{r} > 0$ stands for the constant gain parameter of the repulsion field.

 Note that when $\gamma_{N}(t)\in\left[-1,P^{N}_{0}\right)$ holds, $U^{N}_{r}(\gamma_{N})$ is zero-gradient with respect to $\gamma_{N}$, indicating that no repulsive effect is exerted on the closed-loop system when $\boldsymbol{B}_{i}$ is outside the acting region. 
 Meanwhile, $\lim_{\gamma_{N}\to P^{N}_{1}}U^{N}_{r}(t) = +\infty$ holds and hence provides the barrier characteristic, which guarantees the invariant of the pointing-constraint-satisfying set $\mathcal{S}_{p}$ if $U^{N}_{r}(t)\in\mathcal{L}_{\infty}$.
 
Based on each single $N$-th repulsion field, we then further define the combined repulsion field $U_{r}(t)$ of $m$ forbidden zones as $U_{r}(t) = \sum_{N = 1}^{m}U^{N}_{r}(t)$.
In order to ensure that $x_{e} = 0$ stands for an equilibrium point of the potential field, the overall potential function $U(t)$ is constructed as follows:
\begin{equation}
	U(t) \triangleq U_{a}(t)U_{r}(t) = U_{a}(t)\sum_{N=1}^{m}U^{N}_{r}(t)
\end{equation}
Taking the time-derivative of $U_{a}(t)$ and $U^{N}_{r}(t)$ yields the following results:
\begin{equation}
	\begin{aligned}
		\dot{U}_{a}(t) &= k_{a}\dot{x}_{e} = k_{a}\left(\boldsymbol{r}^{\times}_{b}\boldsymbol{B}_{b}\right)^{\text{T}}\boldsymbol{\omega}_{s}\\
\dot{U}^{N}_{r}(t) 
&= 
\begin{cases}
	\begin{aligned}
			0&,\quad\text{if}\quad\gamma_{N}(t)\in\left[-1,P^{N}_{0}\right)\\
		-&k_{r}a_{N}\sec(P_{N}) \tan(P_{N})(\left(\boldsymbol{f}^{N}_{b}\right)^{\times}\boldsymbol{B}_{b})^{\text{T}}\boldsymbol{\omega}_{s}\\
		&,\quad\text{if}\quad\gamma_{N}(t)\in\left[P^{N}_{0},P^{N}_{1}\right)\\
	\end{aligned}
\end{cases}
	\end{aligned}
\end{equation}
where $P_{N}\triangleq a_{N}\gamma_{N}+b_{N}$ is defined for brevity.
 Note that $\dot{U}^{N}_{r}(t)$ can be uniformly expressed as $\dot{U}^{N}_{r}(t) = k_{r}\boldsymbol{\nabla}^{\text{T}}_{rN}\boldsymbol{\omega}_{s}(N=1,2,3...m)$, where $\boldsymbol{\nabla}_{rN} = 0$ holds for $\gamma_{N}(t)\in\left[
 -1,P^{N}_{0}\right)$ and $\boldsymbol{\nabla}_{rN} = -a_{N}\sec\left(P_{N}\right)\tan\left(P_{N}\right)(\boldsymbol{f}^{N}_{b})^{\times}\boldsymbol{B}_{b}$ holds for $\gamma_{N}(t)\in\left[P^{N}_{0},P^{N}_{1}\right)$.
Accordingly, the time-derivative of $U(t)$ can be also rearranged as $\dot{U}(t) = \boldsymbol{\nabla}^{\text{T}}_{U}\boldsymbol{\omega}_{s}$, where $\boldsymbol{\nabla}_{U}\in\mathbb{R}^{3}$ represents the combined gradient of the potential field, expressed as:
\begin{equation}\label{nablaU}
	\boldsymbol{\nabla}_{U} = k_{a}U_{r}(t)\boldsymbol{r}^{\times}_{b}\boldsymbol{B}_{b} - k_{r}U_{a}(t)\sum_{N=1}^{m}\boldsymbol{\nabla}_{rN}
\end{equation}

Meanwhile, for the notable "Local Minima Issue", it has been elaborated in existing literature that it is not a major concern for the reduced-attitude representation, such as \cite{chi_reduced}. Practically speaking, a radial control torque will help the system escapes from the resulting critical point even under the worst circumstance, as highlighted in \cite{doria2013algorithm}. Therefore, in order to focus on the major issue, we do not explicitly consider this issue in this paper.

\subsubsection{Potential Function for Angular Velocity Constraint}
In order to constrain the maximum magnitude of each component of the attitude angular velocity and keep it beneath $M_{\omega}$ as required by equation (\ref{angcons}), we introduce the following barrier function $V_{\omega}(t)$ \cite{shen2018rigid}:
\begin{equation}
	V_{\omega}(t) = \frac{k_{\omega}}{2}\sum_{i=1}^{3}\ln\left(\frac{M^{2}_{\omega}}{M^{2}_{\omega} - \omega^{2}_{si}(t)}\right)
\end{equation}
where $k_{\omega}>0$ is a constant gain parameter, $M_{\omega}>0$ stands for the given upper bound. It can be observed that $V_{\omega} = 0$ holds for $\boldsymbol{\omega}_{s} = \boldsymbol{0}_{3}$.
Taking the time-derivative of $V_{\omega}(t)$, one has:
\begin{equation}\label{Vomega}
	\begin{aligned}
		\dot{V}_{\omega}(t) &= k_{\omega}\sum_{i=1}^{3}\frac{1}{(M^{2}_{\omega}-\omega^{2}_{si}(t))}\omega_{si}(t)\dot{\omega}_{si}(t) = \boldsymbol{\omega}^{\text{T}}_{s}\boldsymbol{R}_{\omega}\dot{\boldsymbol{\omega}}_{s}
	\end{aligned}
\end{equation}
where $\boldsymbol{R}_{\omega}\in\mathbb{R}^{3\times 3}$ denotes a diagonal matrix, defined as:
\begin{equation}\label{Romega}
	\boldsymbol{R}_{\omega} \triangleq \text{diag
	}\left(k_{\omega}/(M^{2}_{\omega} - \omega^{2}_{si}(t))\right)(i=1,2,3)
\end{equation}
Accordingly, it can be observed that if $V_{\omega}$ remains bounded, then the angular velocity constraint depicted in equation (\ref{angcons}) remains invariant.

\subsection{Barrier Lyapunov Function Design for PPC}\label{BLFDESIGN}
Following a common error transformation process in the typical PPC control philosophy, this section presents a BLF to ensure the satisfaction of the PFE constraint.

Defining a unified error variable as $\varepsilon_{q}(t) \triangleq x_{e}(t)/\rho_{q}(t)$, then the PFE constraint given in equation (\ref{pfe}) can be equivalently rearranged as: $\varepsilon_{q}(t) < 1$. 
 Accordingly, the BLF $V_{B}(t)$ is provided as \cite{liu2019appointed}:
\begin{equation}
	V_{B}(t) = k_{B}\ln\frac{1}{1-\varepsilon_{q}(t)}
\end{equation}
where $k_{B}>0$ stands for a constant design parameter. Note that $\lim_{\varepsilon_{q}(t)\to1}V_{B}(t)\to+\infty$ holds, hence the PFE-constraint-satisfying region $\mathcal{S}_{e}$ remains invariant if $V_{B}(t)$ remains bounded during the whole control process.
To process further, consider the time-derivative of $\varepsilon_{q}(t)$, it can be obtained that:
\begin{equation}\label{doteps}
	\dot{\varepsilon}_{q}(t) = \frac{1}{\rho_{q}(t)}\left(\left(\boldsymbol{r}^{\times}_{b}\boldsymbol{B}_{b}\right)^{\text{T}}\boldsymbol{\omega}_{s} - \frac{\dot{\rho}_{q}(t)}{\rho_{q}(t)}x_{e}(t)\right)
\end{equation}
Therefore, the time-derivative of $V_{B}(t)$ can be derived as:
\begin{equation}\label{dotVB}
	\dot{V}_{B}(t) = R_{\rho}(t)\left[
	\frac{1}{\rho_{q}(t)}\left(\boldsymbol{r}^{\times}_{b}\boldsymbol{B}_{b}\right)^{\text{T}}\boldsymbol{\omega}_{s}\right]-R_{\rho}(t)\frac{\dot{\rho}_{q}(t)}{\rho_{q}(t)}\varepsilon_{q}(t)
\end{equation}
where $R_{\rho}(t)$ is defined as :
\begin{equation}\label{Rrho}
R_{\rho}(t) \triangleq k_{B}/(1-\varepsilon_{q}(t))
\end{equation}

\subsection{Switched Prescribed Performance Function (SPPF) Design}\label{SPPF}
In this subsection, we highlight the design of the proposed Switched Prescribed Performance Function (SPPF). 
The proposed SPPF monitors various factors that may cause the aforementioned contradiction between multiple constraints, and then prioritizes the most critical one, further switches the dynamic of $\rho_{q}(t)$ to establish compatibility.

\subsubsection{Basic Structure of SPPF}
The basic structure of the proposed SPPF $\rho_{q}(t)$ is established as follows:
\begin{equation}\label{dRho}
	\begin{aligned}		
		\dot{\rho}_{q}(t) &= -k_{\rho}\left(\rho_{q}(t)-\rho_{\infty}\right)\left(1-\Omega_{Q}(t)\right)\\
		&\quad +\left(\frac{\dot{x}_{e}(t)}{x_{e}(t)}\right)\rho_{q}(t)\Omega_{Q}(t)
	\end{aligned}
\end{equation}
where $\rho_{\infty}>0$ denotes the constant terminal value of $\rho_{q}(t)$, which is usually designed regarding to the given performance criteria, $k_{\rho} > 0$ controls the exponential decaying rate of $\rho_{q}(t)$. $\Omega_{Q}$ is defined as the switching indicator that governs the $\rho_{q}$-dynamics, of which the range is $\Omega_{Q}\in\left[0,1\right]$. Correspondingly, note that $\dot{\rho}_{q}(t) = -k_{\rho}(\rho_{q}(t)-\rho_{\infty})$ holds for $\Omega_{Q}(t)=0$, and $\dot{\rho}_{q}(t)=\frac{\dot{x}_{e}(t)}{x_{e}(t)}\rho_{q}(t)$ holds for $\Omega_{Q}(t)=1$, and hence we have $\lim_{t\to+\infty}\rho_{q}(t) = \rho_{\infty}$ for $\Omega_{Q}(t) = 0$. The specific expression of $\Omega_{Q}$ will be provided in the following Subsection \ref{APFPPC}-\ref{SPPF}.\ref{SwitchingIndicator}.

\subsubsection{Design of the Switching Indicator $\Omega_{Q}$}\label{SwitchingIndicator}
To construct the overall switching indicator, we first defining a piece-wise smooth mollified switching function.
\begin{definition} The mollified smooth switching function $\Omega(x,p,S_{0},S_{1})$ is defined as follows:
\begin{equation}
	\label{Omegadef}
	\Omega(x) = \begin{cases}
		0 &  x\in\left(-\infty,S_{0}\right)\\
		\frac{1}{2}\left[\tanh\frac{p\left(S_{1}-S_{0}\right)\cdot\left(x-S_{m}\right)}{\sqrt{(x-S_{0})(S_{1}-x)}} + 1\right] & x\in\left[S_{0},S_{1}\right)\\
		1 &  x\in\left[S_{1},+\infty\right)\\
	\end{cases}
\end{equation} 
\begin{figure}[hbt!]
	\centering 
	\includegraphics[width=0.45\textwidth]{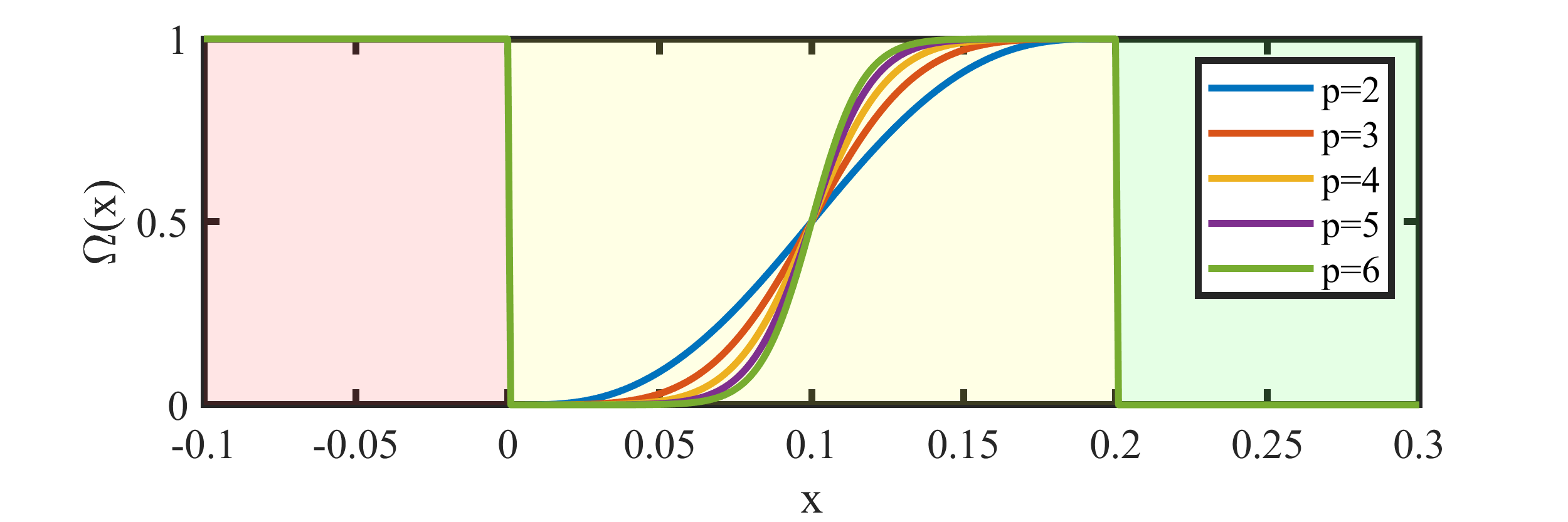}
	\caption{Examples of the Switching Function $\Omega(x)$ with $S_{0} = 0$, $S_{1} = 0.2$, $p=2,3,4,5,6$}       
	\label{Omega}
\end{figure}
where $x$ is the argument of $\Omega(x)$ and $p$, $S_{0}$, and $S_{1}$ are all design parameters. Specifically, $S_{0}$ and $S_{1}$ stand for segment points of $\Omega(x)$, $S_{m} \triangleq \frac{1}{2}\left(S_{0}+S_{1}\right)$ is defined for brevity and $p>\frac{1}{S_{1}-S_{0}}$ controls the increasing rate of $\Omega(x)$ as $x$ increases. It can be observed that $\Omega(x)$ compresses arbitrary input $x$ to a range of $\Omega(x)\in\left[0,1\right]$. Notably, when the definition domain of the input $x$ is not $(-\infty,+\infty)$, then the left and right bound of the definition domain shown in equation (\ref{Omegadef}) should be modified accordingly.
As an additional explanation, several examples of the switching function $\Omega(\cdot)$ are presented in Figure \ref{Omega}.

\end{definition}

Based on the defined switching function, we further define the following $3$ types of switching indicators denoted as $\Omega^{N}_{f}(t)(N=1,..,m)$, $\Omega^{i}_{\omega}(t)(i=1,2,3)$ and $\Omega_{\text{PPC}}(t)$, corresponding to $m$ pointing-forbidden constraints, $3$ attitude angular velocity constraints and the PFE constraint. These switching indicators are specified as follows:
\begin{equation}
	\begin{aligned}
		\Omega^{N}_{f}(t) &\triangleq \Omega(\gamma_{N}(t),M^{N}_{f},S^{N}_{f0},S^{N}_{f1}),(N=1,2,3...m)\\
			\Omega^{i}_{\omega}(t) &\triangleq \Omega(\omega^{2}_{si}(t),N_{\omega},S_{\omega0},S_{\omega1}),(i=1,2,3)\\
						\Omega_{\text{PPC}}(t) &\triangleq \Omega(\varepsilon_{q}(t),Q_{\varepsilon},S_{\varepsilon0},S_{\varepsilon1})
	\end{aligned}
\end{equation}
	where $\gamma_{N}(t)$ is previously defined as $\gamma_{N}(t) = \boldsymbol{B}^{\text{T}}_{i}\boldsymbol{f}^{N}_{i}$, $S^{N}_{f1}$, $S^{N}_{f0}$, $S_{\omega0}$, $S_{\omega1}$, $S_{\varepsilon0}$ and $S_{\varepsilon1}$ are constant segment points that satisfy the following requirements: $S^{N}_{f1} < \cos\Theta^{N}_{f}$, $-1<S^{N}_{f0} < S^{N}_{f1}$, $0<S_{\omega0}<S_{\omega1}$, $S_{\omega1} < M^{2}_{\omega}$, $0<S_{\varepsilon0}<S_{\varepsilon1}$ and $S_{\varepsilon1}<1$. Notably, $S_{\varepsilon0}$ and $S_{\varepsilon1}$ should be set close to $1$. Meanwhile, $M^{N}_{f}$ is a design parameter satisfies $M^{N }_{f} > \frac{1}{S^{N}_{f1} - S^{N}_{f0}}$, corresponding to each $N$-th constraint, $N_{\omega}>0$ satisfies $N_{\omega} > \frac{1}{S_{\omega1}-S_{\omega_{0}}}$ and $Q_{\varepsilon}>\frac{1}{S_{\varepsilon1}-S_{\varepsilon0}}$ holds.

Next, in order to ensure that $\rho_{q}(t)>\rho_{\infty}$ holds under any circumstance, an additional switching indicator $\Omega_{\rho}(t)$ is carried out for robust consideration, expressed as follows:
\begin{equation}
	\Omega_{\rho}(t) = \Omega(\rho_{q}(t),W_{\rho},S_{\rho0},S_{\rho1})
\end{equation}
where $W_{\rho}$ stands for an adjusting parameter, $S_{\rho0}$ and $S_{\rho1}$ are two design parameters, satisfy $\rho_{\infty}<S_{\rho0}<S_{\rho1}$ and $S_{\rho1}$. Meanwhile, both $S_{\rho0}$ and $S_{\rho1}$ should be set close enough to $\rho_{\infty}$.

Subsequently, based on these defined smooth switching indicator functions $\Omega^{N}_{f}(t)(N=1,2,3...m)$, $\Omega^{i}_{\omega}(t)(i=1,2,3)$, $\Omega_{\text{PPC}}(t)$ and $\Omega_{\rho}(t)$, the overall switching indicator $\Omega_{Q}(t)$ is designed in a composite form, expressed as:
\begin{equation}
	\Omega_{Q}(t) \triangleq \Omega_{\rho}(t)\cdot\max(\Omega^{N}_{f}(t),\Omega^{i}_{\omega},\Omega_{\text{PPC}}(t))
\end{equation}

In order to ensure that $\Omega_{Q}(t)$ is smooth and differentiable, we employ the well-defined RealSoftMax function (RSM) that presented in \cite{zhang2021dive} to construct a smooth differentiable maximum approximation, replacing the original non-smooth maximum function $\max(\cdot)$. Correspondingly, $\Omega_{Q}(t)$ can be expressed as follows:
\begin{equation}
	\Omega_{Q}(t) = \frac{1}{K_{s}}\ln\left[\sum_{N=1}^{m}e^{K_{s}\Omega^{N}_{f}}+\sum_{i=1}^{3}e^{K_{s}\Omega^{i}_{\omega}}+e^{K_{s}\Omega_{\text{PPC}}}\right]\cdot\Omega_{\rho}(t)
\end{equation}
where $K_{s}>0$ is a design parameter.

\begin{remark}
Notably, a common characteristic of the RealSoftMax function is that for a set of inputs $x_{j}(j=1,...,k)$, one has $\max(x_{j})<RSM(x_{j})\le\max(x_{j})+\frac{\ln(k)}{K_{s}}$ \cite{zhang2021dive}, with $k$ the number of inputs. Therefore, a big parameter $K_{s}$ should better be chosen to let the tailor-term $\frac{\ln(m+4)}{K_{s}}$ be small enough, thus ensuring a tight enough approximation. Meanwhile, to avoid the circumstance that $RSM(\Omega^{N}_{f},\Omega^{i}_{\omega},\Omega_{\text{PPC}})$ exceeds its original upper bound $1$, an operation can be added for practical implementation:
	\begin{equation}
		\begin{aligned}
		RSM(\Omega^{N}_{f},\Omega^{i}_{\omega},\Omega_{\text{PPC}}) &= 1,\quad \text{if} \quad	RSM(\Omega^{N}_{f},\Omega^{i}_{\omega},\Omega_{\text{PPC}})>1\\
		\end{aligned}
	\end{equation}

\end{remark}

\begin{remark}\label{Omegaselect}
	Notably, the selection of each $S^{N}_{f0}$ should satisfy $S^{N}_{f0}>\boldsymbol{r}^{\text{T}}_{i}\boldsymbol{f}^{N}_{i}$. This is necessary to ensure that each switching indicator $\Omega^{N}_{f}(t)$ will switch back to $0$ when the boresight vector $\boldsymbol{B}_{i}$ arrives to be sufficiently close to the desired pointing direction $\boldsymbol{r}_{i}$.
	
\end{remark}

\subsection{Working Mechanism of the SPPF and the Freezing Mechanism}\label{SPPFMECHANISM}
In this subsection, we elaborate on the working mechanism of the SPPF, who mitigates the contradiction following a two-step logic: contradiction monitoring and PPC freezing.

\subsubsection{Constraint Monitoring}
Those switching indicators $\Omega^{N}_{f}(t)$, $\Omega^{i}_{\omega}(t)$ and $\Omega_{\text{PPC}}(t)$ are defined to provide consistent monitoring on each kind of constraint. Note that each switching indicator function will switch to $1$ when it approaches the boundary of its corresponding constraint-satisfying set. We present a sketch map to facilitate the explanation, shown in Figure \ref{exp}. For instance, if $\boldsymbol{B}_{i}$ approaches the forbidden region $\boldsymbol{f}^{N}_{i}$ such that $\gamma_{N}(t)\ge S^{N}_{f0}$ holds (Figure \ref{exp}-Subfigure $1$ and $4$), then $\Omega^{N}_{f}$ will increase and finally reaches $\Omega^{N}_{f}=1$ if $\gamma_{N}(t) \ge S^{N}_{f1}$.
A similar process will happen for other constraints, like when $|\omega_{si}(t)|$ approaches its allowed upper bound $M_{\omega}$, or the $x_{e}$-trajectory is extremely close to the performance envelope's boundary. Overall, these circumstances will lead $\Omega_{Q}(t)$ switches to $1$. (Figure \ref{exp}-Subfigure $2$).
\begin{figure}[hbt!]
	\centering 
	\includegraphics[width=0.45\textwidth]{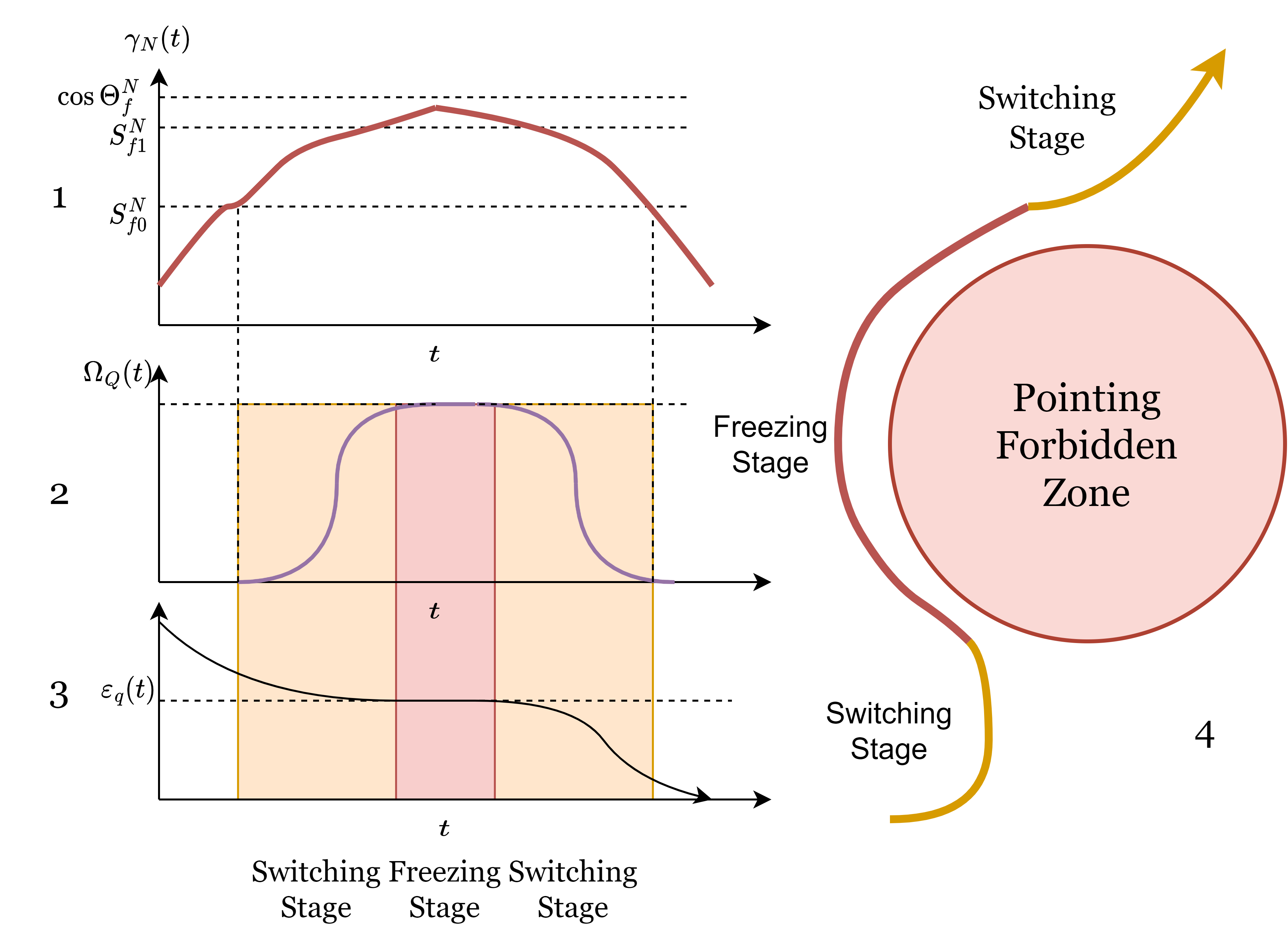}
	\caption{Explanation Sketch map of SPPF's Working Mechanism}       
	\label{exp}   
\end{figure}	

\subsubsection{PPC Freezing}
According to the design shown in Subsection \ref{APFPPC}-\ref{SPPF}. equation (\ref{dRho}), for $\Omega_{Q}(t)=1$, note we have $\dot{\rho}_{q}(t) = \dot{x}_{e}(t)\rho_{q}(t)/x_{e}(t)$. Recalling the definition of $\varepsilon_{q}(t) = x_{e}(t)/\rho_{q}(t)$ and combined with the expression of $\dot{\rho}_{q}(t)$, it can be further yielded that: $\dot{\varepsilon}_{q}(t) = \frac{\dot{x}_{e}(t)\rho_{q}(t)-\dot{\rho}_{q}(t)x_{e}(t)}{\rho^{2}_{q}(t)} = 0$.
This indicates that $\varepsilon_{q}(t)$ will remain unchanged no matter how the pointing error variable $x_{e}(t)$ varies, which stands for the meaning of the PPC freezing (Figure \ref{exp}-Subfigure 3, the Freezing Stage).

According to the design of $\Omega_{\text{PPC}}(t)$, the value of $\varepsilon_{q}(t)$ will be frozen to be smaller than $S_{\varepsilon1}$, indicating that $\varepsilon_{q}(t)<1$ holds under any circumstance. Therefore, this directly excludes the possibility of the PPC-singularity as the performance envelope always encloses the state trajectory. Meanwhile, Since the BLF $V_{B}$ is directly built based on $\varepsilon_{q}(t)$, such a mechanism naturally vanishes the impact of the PPC system on the total closed-loop system as $\dot{V}_{B}(t) = 0$ holds. 
Therefore, the system will be temporarily guided only by the APF part, providing a safely obstacle-circumvent process. 

When the boresight vector approaches the desired pointing direction, the design of $\Omega_{\rho}(t)$ guarantees that $\Omega_{Q}(t)$ will switch back to $0$, and hence $\rho_{q}(t)$ will converge exponentially again as $\dot{\rho}_{q}(t) = -k_{\rho}(\rho_{q}(t)-\rho_{\infty})$ holds, indicating the PPC system ($\varepsilon_{q}$-system) is reactivated again. Therefore, $\rho_{q}(t)$ will finally converge to approach $\rho_{q}(t)\to\rho_{\infty}$, and $\varepsilon_{q}(t)$ will converge simultaneously, which further resulting the $x_{e}$-trajectory to converge with the satisfaction of the pointing accuracy requirement. 

Additionally, the mechanism of how the combined PPC control scheme enhances the system's performance will be later discussed from a Lyapunov sense in Section \ref{Discussion}.

%

\section{ADAPTIVE I\&I-based APF-PPC COMPOSITE CONTROLLER DESIGN}\label{ADAPTIVE}
In this subsection, we present an adaptive APF-PPC controller to achieve the control objective stated in Subsection \ref{formulation}-\ref{PROBLEM}, of which is constructed by combining the proposed integrated APF-PPC control framework with the I\&I adaptive methodology. By regarding the control law derived from the APF-PPC control scheme as a static nominal control part, following the basic idea stated in \cite{shao2021immersion}, it allows us to incorporate this nominal control law into the parameter adaptive strategy and hence constructs the structure of the I\&I adaptive control. A schematic diagram of the proposed controller is illustrated in Figure \ref{Controller}.

\begin{figure}[hbt!]
	\centering 
	\includegraphics[width=0.5\textwidth]{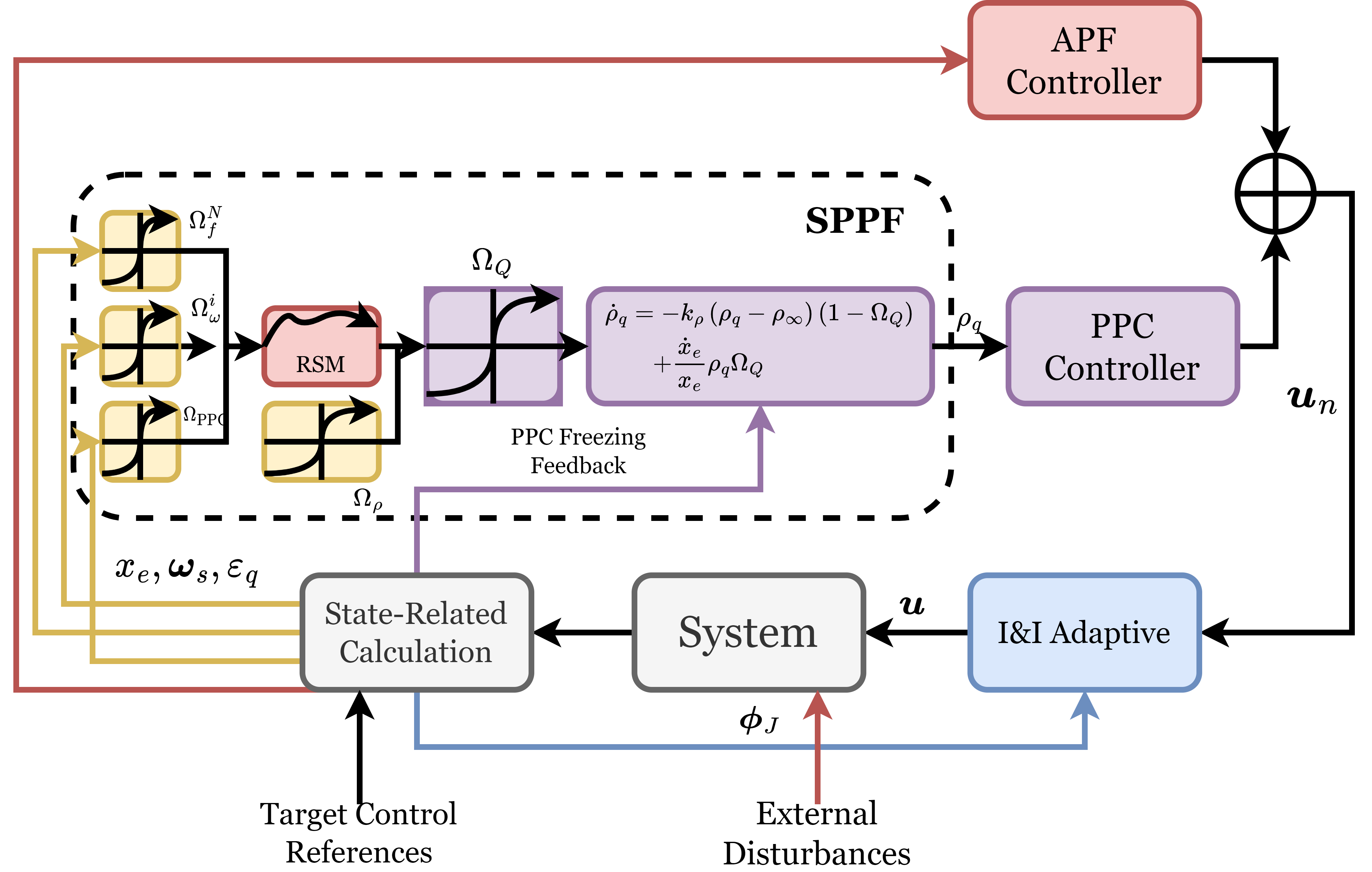}
	\caption{Sketch map of the Proposed Adaptive APF-PPC Controller}       
	\label{Controller}
\end{figure}

 \subsection{Controller Derivation}
 The adaptive control law $\boldsymbol{u}\in\mathbb{R}^{3}$ is designed as follows:
 \begin{equation}\label{controllaw}
 	\begin{aligned}
 		\boldsymbol{u} =& -\boldsymbol{\Psi}\left(\hat{\boldsymbol{\theta
 		}}+\boldsymbol{\beta}\right)
 	\end{aligned}
 \end{equation}
 where $\boldsymbol{\Psi}\in\mathbb{R}^{3\times 3}$ denotes an expanded regression matrix, defined as $\boldsymbol{\Psi}\triangleq \boldsymbol{\phi}_{J}-\mathcal{L}(\boldsymbol{u}_{n})$, $\boldsymbol{u}_{n}\in\mathbb{R}^{3}$ represents an additional "static" control term, specified as follows:
 \begin{equation}\label{Un}
 	\begin{aligned}
 		\boldsymbol{u}_{n} &\triangleq -\boldsymbol{R}^{-1}_{\omega}\left(\boldsymbol{\nabla}_{U}+\frac{R_{\rho}\boldsymbol{r}^{\times}_{b}\boldsymbol{B}_{b}}{\rho_{q}}\right)-K_{\omega}\boldsymbol{R}_{\omega}\boldsymbol{\omega}_{s}\\
 		&\quad +  		\boldsymbol{R}^{-1}_{\omega}\left(\frac{R_{\rho}\dot{\rho}_{q}\varepsilon_{q}}{\rho_{q}}\right)\frac{\boldsymbol{\omega}_{s}}{\|\boldsymbol{\omega}_{s}\|^{2}}
 	\end{aligned}
 \end{equation}
 where $K_{\omega}> 0$ is a constant controller gain parameter, $\boldsymbol{R}_{\omega}$ (cf. equation (\ref{Romega})) and $R_{\rho}$ (cf. equation (\ref{Rrho})) are previously defined matrix and scalar function, $\boldsymbol{\nabla}_{U}\in\mathbb{R}^{3}$ is the defined overall gradient of potential functions (cf. equation (\ref{nablaU})).
 
 In order to facilitate the following analysis, we first decompose the expanded regression matrix $\boldsymbol{\Psi}$ into two parts as $\boldsymbol{\Psi}_{1} \triangleq \boldsymbol{\phi}_{J}\in\mathbb{R}^{3\times 3}$ and $\boldsymbol{\Psi}_{2} \triangleq \mathcal{L}(-\boldsymbol{u}_{n})\in\mathbb{R}^{3\times3}$, hence it naturally comes $\boldsymbol{\Psi} = \boldsymbol{\Psi}_{1} + \boldsymbol{\Psi}_{2}$.
  Regarding necessary steps in I\&I adaptive control \cite{astolfi2003immersion}, a regulation function $\boldsymbol{\beta}\in\mathbb{R}^{3}$ needs to be specified such that $\partial \boldsymbol{\beta}/\partial\boldsymbol{\omega}_{s} = \boldsymbol{\Psi}^{\text{T}}$ holds. Considering the decomposed regression matrix $\boldsymbol{\Psi}_{2}$, it can be observed that a feasible solution $\boldsymbol{\beta}_{2}\in\mathbb{R}^{3}$ of the Partial Differentiation Equation (PDE): $\partial \boldsymbol{\beta}_{2}/\partial\boldsymbol{\omega}_{s} = \boldsymbol{\Psi}^{\text{T}}_{2}$ can be directly obtained following an element-wisely operation. The $i$-th component of $\boldsymbol{\beta}_{2}$ can be given as follows:
 \begin{equation}\label{betatwo}
 	\begin{aligned}
 		 	\beta_{2i} &= \frac{1}{k_{\omega}}\left(M^{2}_{\omega}\omega_{si}-\frac{1}{3}\omega^{3}_{si}\right)\nabla_{Ci}-\frac{K_{\omega}k_{\omega}}{2}\ln\left(\frac{1}{R_{\omega i}}\right)\\
 		 	&\quad +\frac{\delta_{n}(t)}{2k_{\omega}}\left[\ln(\|\boldsymbol{\omega}_{s}\|^{2})\left(M^{2}_{\omega}+\|\boldsymbol{\omega}_{s}\|^{2}-\omega^{2}_{si}\right)-\omega^{2}_{si}\right]
 	\end{aligned}
 \end{equation}
 where $\delta_{n}(t) \triangleq R_{\rho}(t)\dot{\rho}_{q}(t)\varepsilon_{q}(t)/\rho_{q}(t)$ is defined for brevity, $\nabla_{Ci}$ stands for the $i$-th component of a combined vector $\boldsymbol{\nabla}_{C}\in\mathbb{R}^{3}$, defined as $\boldsymbol{\nabla}_{C}\triangleq \boldsymbol{\nabla}_{U} + \frac{R_{\rho}\boldsymbol{r}^{\times}_{b}\boldsymbol{B}_{b}}{\rho_{q}}$, $R_{\omega i}$ denotes the $i$-th diagonal element of $\boldsymbol{R}_{\omega}$.
 
On the other hand, as for $\boldsymbol{\Psi}_{1}$, due to the well-known integrability obstacle that brought by the anti-symmetric cross product manipulation matrix, as mentioned by \cite{shao2021immersion,xia2022anti}, a closed-form solution of the 
 PDE: $\partial \boldsymbol{\beta}_{1}/\partial\boldsymbol{\omega}_{s}=\boldsymbol{\Psi}^{\text{T}}_{1}$ does not exist, hence we introduce a filtering technique to remove this barrier. $\boldsymbol{\beta}_{1}\in\mathbb{R}^{3}$ is designed as follows:
 \begin{equation}\label{betaone}
 	\begin{aligned}
 		\boldsymbol{\beta}_{1} = \hat{\boldsymbol{\Psi}}^{\text{T}}_{1}\boldsymbol{\omega}_{s}
 	\end{aligned}
 \end{equation}
 where $\hat{\boldsymbol{\Psi}}_{1}\in\mathbb{R}^{3\times 3}$ denotes a filtered regression matrix of the unsolvable part, constructed by replacing every $\boldsymbol{\omega}_{s}$ with the filtered $\hat{\boldsymbol{\omega}}_{s}$, expressed as:
 \begin{equation}\label{psihat}
 	\hat{\boldsymbol{\Psi}}_{1}\triangleq -\hat{\boldsymbol{\omega}}^{\times}_{s}\mathcal{L}({\boldsymbol{\omega}}_{s})
 \end{equation}
 where $\hat{\boldsymbol{\omega}}_{s}\in\mathbb{R}^{3}$ denotes the filtered variable corresponds to $\boldsymbol{\omega}_{s}$, generated by an identification filter that is introduced in \cite{xia2022anti}, expressed as follows:
 \begin{equation}\label{filter}
 	\begin{aligned}
 		\dot{\hat{\boldsymbol{\omega}}}_{s} &=- \left(K_{f}+g_{a}(t)\right)\left(\hat{\boldsymbol{\omega}}_{s} - \boldsymbol{\omega}_{s}\right)  + \boldsymbol{u}_{n},\hat{\boldsymbol{\omega}}_{s}(t_{0}) = \boldsymbol{\omega}_{s}(t_{0})\\
 	\end{aligned}
 \end{equation}
 where $K_{f}>0$ represents a constant filtering gain parameter, while $g_{a}(t)$ is a time-varying filtering gain parameter, updated by the following dynamics:
 \begin{equation}\label{dotka}
 	\dot{g}_{a}(t) = K_{a}\|\hat{\boldsymbol{\omega}}_{s}-\boldsymbol{\omega}_{s}\|^{2},g_{a}(t_{0})=0
 \end{equation}
 where $K_{a}>0$ is a constant regulation parameter.
 
Finally, the parameter adaptive law $\dot{\hat{\boldsymbol{\theta}}}$ and the regulation function $\boldsymbol{\beta}$ are specified as:
 \begin{equation}\label{thetahat}
 	\begin{aligned}
 	  		\boldsymbol{\beta} &= C_{\beta}\left[\hat{\boldsymbol{\Psi}}_{1}^{\text{T}}\boldsymbol{\omega}_{s}+\boldsymbol{\beta}_{2}\right]\\
 	  \dot{\hat{\boldsymbol{\theta}}} &= -C_{\beta}\left[\dot{\hat{\boldsymbol{\Psi}}}_{1}^{\text{T}}\boldsymbol{\omega}_{s}+\left(\hat{\boldsymbol{\Psi}}_{1}+\boldsymbol{\Psi}_{2}\right)^{\text{T}}\left(\boldsymbol{u}_{n}\right)\right] -C_{\beta}\dot{\bar{\boldsymbol{\beta}}}_{2}
 	\end{aligned}
 \end{equation} 
 where $C_{\beta}>0$ denotes the gain parameter, $\dot{\bar{\boldsymbol{\beta}}}_{2}\in\mathbb{R}^{3}$ is calculated by considering the time-derivative of $\boldsymbol{\beta}_{2}$ with respect to all variables except the angular velocity, i.e., $\dot{\bar{\boldsymbol{\beta}}}_{2}\triangleq \dot{\boldsymbol{\beta}}_{2} - \frac{\partial\boldsymbol{\beta}_{2}}{\partial\boldsymbol{\omega}_{s}}\dot{\boldsymbol{\omega}}_{s}$.
 
\subsection{Discussion on Implementation of the Adaptive APF-PPC Controller}
 	Additionally, some certain statements on the implementation of the proposed controller are presented below.

Note that $\dot{\bar{\boldsymbol{\beta}}}_{2}$ is necessitated for parameter adaptation, which can be expressed as:
 	\begin{equation}
 		\dot{\bar{\boldsymbol{\beta}}}_{2}=\frac{\partial\boldsymbol{\beta}_{2}}{\partial \delta_{n}}\dot{\delta}_{n}+\frac{\partial\boldsymbol{\beta}_{2}}{\partial\boldsymbol{\nabla}_{C}}\dot{\boldsymbol{\nabla}}_{C}
 	\end{equation}
 	The second part can be directly obtained since its time-derivative consists of $\dot{\boldsymbol{f}}^{N}_{b}$, $\dot{\boldsymbol{r}}_{b}$, $\varepsilon_{q}$, $\dot{\varepsilon}_{q}$, $\dot{\rho}_{q}$, and each signal is calculable using its analytical expression.
 	 However, for its first part, owing to the fact that $\dot{\delta}_{n}(t)$ includes a term like $\dot{\rho}_{q}/\rho_{q}$, its time-derivative cannot be obtained in a closed form as $\dot{\boldsymbol{\omega}}_{s}$ is unavailable. For practically implementation, one could use a Dynamic Surface Control technique that was originally announced in \cite{swaroop2000dynamic}, which utilizes a first-order low-pass filter to obtain a filtered signal of the original input. On the other hand, high-order differentiators or Tracking Differentiators that mentioned in \cite{shao2021immersion} are also applicable for this scenario. 
 	 
 	 Meanwhile, upon observing the expression of $\dot{\delta}_{n}$, note that it contains a radially unbounded term $\frac{1}{1-\varepsilon_{q}}$. However, owing to the design of the freezing mechanism, we could freeze the value of $\varepsilon_{q}$ and keep it remain beneath $S_{\varepsilon1}$. Therefore, this allows us to reduce the resulting exogenous perturbation caused by the time-derivative of input signal to an acceptable level such that it only has tiny influence on the filtering system.

\section{STABILITY ANALYSIS}

\begin{theorem}\label{T1}
	For the reduced-attitude system given by equation (\ref{errorsystem}), under the satisfaction of Assumption \ref{AssumpJ}, with the given adaptive APF-PPC controller given by equations (\ref{controllaw})(\ref{Un})(\ref{betatwo})(\ref{betaone})(\ref{psihat})(\ref{filter})(\ref{dotka})(\ref{thetahat}), the closed-loop system is asymptotically stabled, with the pointing-forbidden constraint (\ref{pointcons}), attitude angular velocity constraint (\ref{angcons}) and the PPC PFE constraint (\ref{pfe}) respected.
\end{theorem}

\begin{proof}
Firstly, note that the control law stated in equation (\ref{controllaw}) can be equivalently decomposed as follows:
\begin{equation}
	\boldsymbol{u} = -\boldsymbol{\phi}_{J}(\hat{\boldsymbol{\theta}}+\boldsymbol{\beta})+\mathcal{L}(\boldsymbol{u}_{n})(\hat{\boldsymbol{\theta}}+\boldsymbol{\beta})
\end{equation}
Therefore, by plugging the control law into the system's dynamics equation (\ref{sysdyna}), it can be obtained that:
\begin{equation}\label{dotomega1}
	\begin{aligned}
			\dot{\boldsymbol{\omega}}_{s} &= \boldsymbol{J}^{-1}\left[\boldsymbol{\phi}_{J}\boldsymbol{\theta}-\boldsymbol{\phi}_{J}\left(\hat{\boldsymbol{\theta}}+\boldsymbol{\beta}\right)\right] -\boldsymbol{J}^{-1}\mathcal{L}(-\boldsymbol{u}_{n})\left(\hat{\boldsymbol{\theta}}+\boldsymbol{\beta}\right)
	\end{aligned}
\end{equation}

Defining an estimation error variable $\tilde{\boldsymbol{\theta}}\in\mathbb{R}^{3}$ as $\tilde{\boldsymbol{\theta}} \triangleq \hat{\boldsymbol{\theta}}+\boldsymbol{\beta}-\boldsymbol{\theta}$, one has $\tilde{\boldsymbol{\theta}}+\boldsymbol{\theta}=\hat{\boldsymbol{\theta}}+\boldsymbol{\beta}$. Accordingly, this rearranges the closed-loop dynamics of $\boldsymbol{\omega}_{s}$ that depicted in equation (\ref{dotomega1}) into the following form:
\begin{equation}\label{rearrangeddotomega}
	\begin{aligned}
			\dot{\boldsymbol{\omega}}_{s} &= \boldsymbol{J}^{-1}\left[
		-\boldsymbol{\phi}_{J}\tilde{\boldsymbol{\theta}} - \mathcal{L}(-\boldsymbol{u}_{n})\left(\boldsymbol{\theta}+\tilde{\boldsymbol{\theta}}\right)
		\right]\\
		&= -\boldsymbol{J}^{-1}\left[\boldsymbol{\phi}_{J}\tilde{\boldsymbol{\theta}}+\mathcal{L}(-\boldsymbol{u}_{n})\tilde{\boldsymbol{\theta}}\right]+\boldsymbol{u}_{n}\\
	\end{aligned}
\end{equation}
By employing the previously defined expanded regression matrix $\boldsymbol{\Psi}\in\mathbb{R}^{3\times3}$, the closed-loop dynamics of $\boldsymbol{\omega}_{s}$ given in equation (\ref{rearrangeddotomega}) can be further rewritten in a compact form as:
\begin{equation}\label{dotomegas}
	\dot{\boldsymbol{\omega}}_{s} = -\boldsymbol{J}^{-1}\boldsymbol{\Psi}\tilde{\boldsymbol{\theta}}+\boldsymbol{u}_{n}
\end{equation}

 Subsequently, we consider the dynamics of the estimator design. Taking the time-derivative of $\tilde{\boldsymbol{\theta}}$, one has $\dot{\tilde{\boldsymbol{\theta}}}=\dot{\hat{\boldsymbol{\theta}}}+\dot{\boldsymbol{\beta}}$.
By utilizing the design of $\hat{\boldsymbol{\theta}}$ and $\boldsymbol{\beta}$ that depicted in the equation (\ref{thetahat}), this gives a rise to:
\begin{equation}
	\begin{aligned}
			&\dot{\tilde{\boldsymbol{\theta}}} = \dot{\hat{\boldsymbol{\theta}}}+\dot{\boldsymbol{\beta}}= -C_{\beta}\left(\hat{\boldsymbol{\Psi}}_{1}+\boldsymbol{\Psi}_{2}\right)^{\text{T}}\boldsymbol{J}^{-1}\boldsymbol{\Psi}\tilde{\boldsymbol{\theta}}
	\end{aligned}
\end{equation}
Let the discrepancy between the combined filtered regression matrix $\hat{\boldsymbol{\Psi}}_{1}+\boldsymbol{\Psi}_{2}$ and its true value $\boldsymbol{\Psi}$ as: $\boldsymbol{\Delta}_{\Psi}\triangleq \hat{\boldsymbol{\Psi}}_{1}+\boldsymbol{\Psi}_{2}-\boldsymbol{\Psi}\in\mathbb{R}^{3\times3}$. Notably, the only difference between $\hat{\boldsymbol{\Psi}}_{1}$ and $\boldsymbol{\Psi}_{1}$ is essentially generated by the discrepancy between $\hat{\boldsymbol{\omega}}_{s}$ and $\boldsymbol{\omega}_{s}$. Hence, according to the Mean Value Theorem (MVT), this derives the following result:
\begin{equation}
	\|\boldsymbol{\Delta}_{\Psi}\|^{2}\le S_{\omega}\|\hat{\boldsymbol{\omega}}_{s}-\boldsymbol{\omega}_{s}\|^{2}
\end{equation}
where $S_{\omega}>0$ stands for an unknown Lipschitz constant.
Applying the notation of $\boldsymbol{\Delta}_{\Psi}$, the time-derivative of $\tilde{\boldsymbol{\theta}}$ can be rearranged as:
\begin{equation}
    \begin{aligned}
    	 \dot{\tilde{\boldsymbol{\theta}}}= -C_{\beta}\boldsymbol{\Psi}^{\text{T}}\boldsymbol{J}^{-1}\boldsymbol{\Psi}\tilde{\boldsymbol{\theta}}-C_{\beta}\boldsymbol{\Delta}^{\text{T}}_{\Psi}\boldsymbol{J}^{-1}\boldsymbol{\Psi}\tilde{\boldsymbol{\theta}}
    \end{aligned}
\end{equation}

To compensate for the impact resulted from the filtering-error $\boldsymbol{\Delta}_{\Psi}$, motivated by \cite{xia2022anti}, this paper designs a bounded dynamic scaling factor to facilitate the stability analysis. The scaled error variable $\boldsymbol{z}$ is defined as:
\begin{equation}
	\boldsymbol{z}\triangleq \tilde{\boldsymbol{\theta}}/\chi(t)
\end{equation} where $\chi(t)$ denotes the designed scaling factor, expressed as:
\begin{equation}\label{chit}
	\chi(t) = e^{f(r)}\sqrt{J_{\min}}/e
\end{equation}
with $f(r)$ a function of the auxiliary variable $r(t)$, defined as:
\begin{equation}\label{f_r}
	f(r) = \frac{2}{1+e^{-k_{\chi}r(t)}}-1
\end{equation}
where $k_{\chi}>0$ stands for an adjusting parameter. $r(t)$ is a dynamically-generated input, specified as:
\begin{equation}
	\begin{aligned}
		\dot{r}(t) &= \frac{C_{r1}}{\partial f(r)/\partial r}\|\boldsymbol{\Delta}_{\Psi}\|^{2},\quad r(t_{0})>0\\
	\end{aligned}
\end{equation}
where $C_{r1}>0$ is a design parameter that will be specified later. 
From the expression of $f(r)$ given in equation (\ref{f_r}), note that we have the following properties: $f(0) = 0$, $\lim_{r\to+\infty}f(r) = 1$ and $\partial f(r)/\partial r > 0$. Therefore, owing to the fact that $\dot{r}(t)>0$, given the initial condition of $r(t)$ as $r(t_{0})>0$, one has $f(r(t))\in\left(0,1\right)$ for $\forall t\in\left[t_{0},+\infty\right)$, and this further yields $\chi(t)\in\left(\sqrt{J_{\min}}/e,\sqrt{J_{\min}}\right)$.

Meanwhile, if $J_{\min}\le e^{2}$ holds, the initial condition of $r(t_{0})$ can be chosen to satisfy $e^{f(r)}>\frac{e}{\sqrt{J_{\min}}}$ such that $\chi(t_{0})>1$ holds; if $J_{\min}>e^{2}$, then any $r(t_{0}) > 0$ ensures $\chi(t_{0})>1$. Accordingly, for any circumstance, a sufficiently large initial condition of $r(t_{0})$ exists, which ensures $\chi(t_{0})\ge1$.

Accordingly, taking the time-derivative of $\chi(t)$ and combined with the dynamics of the given $r(t)$ yields:
		$\dot{\chi}(t)/\chi(t) = C_{r1}\|\boldsymbol{\Delta}_{\Psi}\|^{2}$.
		Next, considering the time-derivative of $\boldsymbol{z}$, it can be derived that:
\begin{equation}\label{dotz}
	\begin{aligned}
		\dot{\boldsymbol{z}} &= -C_{\beta}\boldsymbol{\Psi}^{\text{T}}\boldsymbol{J}^{-1}\boldsymbol{\Psi}\boldsymbol{z}-C_{\beta}\boldsymbol{\Delta}^{\text{T}}_{\Psi}\boldsymbol{J}^{-1}\boldsymbol{\Psi}\boldsymbol{z} -C_{r1}\|\boldsymbol{\Delta}_{\Psi}\|^{2}\boldsymbol{z}
	\end{aligned}
\end{equation}

We further choose the candidate lyapunov function of the scaled estimator $\boldsymbol{z}$ as $V_{z} = \frac{1}{2}\boldsymbol{z}^{\text{T}}\boldsymbol{z}$. Taking the differentiation of $V_{z}$ with respect to $t$, by considering the result of equation (\ref{dotz}) and applying the Young Inequality and Peter-Paul Inequality \cite{trif2005note}, it can be yielded that:
\begin{equation}\label{dotVz}
	\begin{aligned}
		\dot{V}_{z} &\le -C_{\beta}J_{\min}\|\boldsymbol{J}^{-1}\boldsymbol{\Psi}\boldsymbol{z}\|^{2}-C_{r1}\|\boldsymbol{\Delta}_{\Psi}\|^{2}\|\boldsymbol{z}\|^{2}\\
		&\quad +\frac{C_{\beta}J_{\min}}{2}\|\boldsymbol{J}^{-1}\boldsymbol{\Psi}\boldsymbol{z}\|^{2}+ \frac{C_{\beta}}{2J_{\min}}\|\boldsymbol{\Delta}_{\Psi}\|^{2}\|\boldsymbol{z}\|^{2}\\
	\end{aligned}
\end{equation}
We then use $C_{r1}$ to establish an additional compensation to dominate the perturbation term expressed as $\frac{C_{\beta}}{2J_{\min}}\|\boldsymbol{\Delta}_{\Psi}\|^{2}\|\boldsymbol{z}\|^{2}$. Let $C_{r1} = \frac{C_{\beta}}{2J_{\min}}$, then it can be derived from equation (\ref{dotVz}) that:
\begin{equation}\label{dotVz2}
			\dot{V}_{z} 
		\le -\frac{C_{\beta}J_{\min}}{2}\|\boldsymbol{J}^{-1}\boldsymbol{\Psi}\boldsymbol{z}\|^{2}
\end{equation}

Further, considering the time-derivative of the potential function $V_{\omega}$ shown in equation (\ref{Vomega}), by combining it with the equation (\ref{dotomegas}) and the expression of $\boldsymbol{u}_{n}$ (\ref{Un}), one has:
\begin{equation}
	\begin{aligned}
		\dot{V}_{\omega} &=  -\chi\boldsymbol{\omega}^{\text{T}}_{s}\boldsymbol{R}_{\omega}\boldsymbol{J}^{-1}\boldsymbol{\Psi}\boldsymbol{z}+\boldsymbol{\omega}^{\text{T}}_{s}\boldsymbol{R}_{\omega}\boldsymbol{u}_{n}\\
		&\le \frac{1}{C_{\beta}}\|\boldsymbol{R}^{\text{T}}_{\omega}\boldsymbol{\omega}_{s}\|^{2}+\frac{\chi^{2}C_{\beta}}{4}\|\boldsymbol{J}^{-1}\boldsymbol{\Psi}\boldsymbol{z}\|^{2}+\boldsymbol{\omega}^{\text{T}}_{s}\boldsymbol{R}_{\omega}\boldsymbol{u}_{n}\\
		&\le -\left(K_{\omega}-\frac{1}{C_{\beta}}\right)\|\boldsymbol{R}^{\text{T}}_{\omega}\boldsymbol{\omega}_{s}\|^{2}+\frac{J_{\min}C_{\beta}}{4}\|\boldsymbol{J}^{-1}\boldsymbol{\Psi}\boldsymbol{z}\|^{2}\\
		&\quad -\boldsymbol{\omega}^{\text{T}}_{s}\left[\frac{R_{\rho}\boldsymbol{r}^{\times}_{b}\boldsymbol{B}_{b}}{\rho_{q}}-\frac{R_{\rho}\dot{\rho}_{q}\varepsilon_{q}}{\rho_{q}}\frac{\boldsymbol{\omega}_{s}}{\|\boldsymbol{\omega}_{s}\|^{2}}\right] -\boldsymbol{\omega}^{\text{T}}_{s}\boldsymbol{\nabla}_{U}
	\end{aligned}
\end{equation}

Next, considering the filtering system of $\boldsymbol{\omega}_{s}$, we define a filtering error as $\tilde{\boldsymbol{\omega}}_{s} \triangleq \hat{\boldsymbol{\omega}}_{s} - \boldsymbol{\omega}_{s}$ , and further choose a candidate Lyapunov function as $V_{f1} = \frac{1}{2}\|\tilde{\boldsymbol{\omega}}_{s}\|^{2}$. Taking the time-derivative of $V_{f1}$ and combined with $\dot{\hat{\boldsymbol{\omega}}}_{s}$ that given in equation (\ref{filter}), it can be obtained that:
\begin{equation}
	\begin{aligned}
		\dot{V}_{f1} &=\tilde{\boldsymbol{\omega}}^{\text{T}}_{s}\left[-(K_{f}+g_{a}(t))\tilde{\boldsymbol{\omega}}_{s}+ \boldsymbol{J}^{-1}\boldsymbol{\Psi}\chi\boldsymbol{z}\right] \\
		&\le  -(K_{f}+g_{a})\|\tilde{\boldsymbol{\omega}}_{s}\|^{2}+\frac{J_{\min}C_{\beta}}{4}\|\boldsymbol{J}^{-1}\boldsymbol{\Psi}\boldsymbol{z}\|^{2}+\frac{1}{C_{\beta}}\|\tilde{\boldsymbol{\omega}}_{s}\|^{2}\\
		&= -\left(K_{f}+g_{a}(t)-\frac{1}{C_{\beta}}\right)\|\tilde{\boldsymbol{\omega}}_{s}\|^{2}+\frac{J_{\min}C_{\beta}}{4}\|\boldsymbol{J}^{-1}\boldsymbol{\Psi}\boldsymbol{z}\|^{2}\\
	\end{aligned}
\end{equation}
Meanwhile, considering a candidate Lyapunov-like function for the scaling factor $\chi(t)$ as $V_{\chi} = \frac{1}{2J_{\min}C_{r1}}\chi^{2}(t)$. Note that $C_{r1}$ denotes the gain parameter of the dynamic scaling factor, which has been previously selected as $C_{r1} = \frac{C_{\beta}}{2J_{\min}}$. Accordingly, taking the time-derivative of $V_{\chi}$ yields:
\begin{equation}\label{dotVCHI}
	\begin{aligned}
		\dot{V}_{\chi}(t) &= \frac{1}{J_{\min}C_{r1}}\chi^{2}\frac{\dot{\chi}}{\chi}\le \|\boldsymbol{\Delta}_{\Psi}\|^{2}
	\end{aligned}
\end{equation}
By utilizing the property of the aforementioned Mean Value Theorem (MVT) as: $\|\boldsymbol{\Delta}_{\Psi}\|^{2}\le S_{\omega}\|\tilde{\boldsymbol{\omega}}_{s}\|^{2}$, it can be further derived from equation (\ref{dotVCHI}) that $\dot{V}_{\chi}(t) \le S_{\omega}\|\tilde{\boldsymbol{\omega}}_{s}\|^{2}$ holds.

We further consider a candidate Lyapunov-like function of the time-varying gain parameter $g_{a}(t)$ as $V_{k} = \frac{1}{2K_{a}}(g_{a}(t)-S_{\omega})^{2}$. Taking the time-derivative of $V_{k}$, one has:
	$\dot{V}_{k} = \frac{1}{2K_{a}}2(g_{a}-S_{\omega})K_{a}\|\tilde{\boldsymbol{\omega}}_{s}\|^{2}=(g_{a}-S_{\omega})\|\tilde{\boldsymbol{\omega}}_{s}\|^{2}$

Overall, considering a lumped Lyapunov certificate as $V_{1} = V_{\omega} + V_{z} + V_{f1} + V_{\chi} + V_{k}$, taking its time-derivative yields:
\begin{equation}
	\begin{aligned}
		\dot{V}_{1} &\le -\left(K_{\omega}-\frac{1}{C_{\beta}}\right)\|\boldsymbol{R}^{\text{T}}_{\omega}\boldsymbol{\omega}_{s}\|^{2}-\left(K_{f}+g_{a}(t)-\frac{1}{C_{\beta}}\right)\|\tilde{\boldsymbol{\omega}}_{s}\|^{2}\\
		&\quad-\boldsymbol{\omega}^{\text{T}}_{s}\left[\frac{R_{\rho}\boldsymbol{r}^{\times}_{b}\boldsymbol{B}_{b}}{\rho_{q}}-\left(\frac{R_{\rho}\dot{\rho}_{q}\varepsilon_{q}}{\rho_{q}}\right)\frac{\boldsymbol{\omega}_{s}}{\|\boldsymbol{\omega}_{s}\|^{2}}\right] -\boldsymbol{\omega}^{\text{T}}_{s}\boldsymbol{\nabla}_{U}\\
		&\quad +(g_{a}(t)-S_{\omega})\|\tilde{\boldsymbol{\omega}}_{s}\|^{2}+S_{\omega}\|\tilde{\boldsymbol{\omega}}_{s}\|^{2}\\
		&= -\left(K_{\omega}-\frac{1}{C_{\beta}}\right)\|\boldsymbol{R}^{\text{T}}_{\omega}\boldsymbol{\omega}_{s}\|^{2}-\left(K_{f}-\frac{1}{C_{\beta}}\right)\|\tilde{\boldsymbol{\omega}}_{s}\|^{2}\\
		&\quad-\boldsymbol{\omega}^{\text{T}}_{s}\left[\frac{R_{\rho}\boldsymbol{r}^{\times}_{b}\boldsymbol{B}_{b}}{\rho_{q}}-\left(\frac{R_{\rho}\dot{\rho}_{q}\varepsilon_{q}}{\rho_{q}}\right)\frac{\boldsymbol{\omega}_{s}}{\|\boldsymbol{\omega}_{s}\|^{2}}\right] -\boldsymbol{\omega}^{\text{T}}_{s}\boldsymbol{\nabla}_{U}\\
	\end{aligned}
\end{equation}
To facilitate the system's stability, let $K_{\omega}>\frac{1}{C_{\beta}}$, $K_{f}>\frac{1}{C_{\beta}}$ holds. By combining with the potential function $U(t)$ and the design of the BLF $V_{B}(t)$, an overall Lyapunov certificate is chosen as $V = V_{1} + U + V_{B}$. Correspondingly, by combining the result of $\dot{U}(t)$ with $\dot{V}_{B}(t)$ that provided in equation (\ref{nablaU}) and (\ref{dotVB}), it can be derived that:
\begin{equation}\label{Vall}
	\begin{aligned}
		\dot{V}&\le -\left(K_{\omega}-\frac{1}{C_{\beta}}\right)\|\boldsymbol{R}^{\text{T}}_{\omega}\boldsymbol{\omega}_{s}\|^{2}-\left(K_{f}-\frac{1}{C_{\beta}}\right)\|\tilde{\boldsymbol{\omega}}_{s}\|^{2}\\
	\end{aligned}
\end{equation}
Firstly, upon observing the result that shown in equation (\ref{dotVz2}), note that $\dot{V}_{z}\le 0$ holds, indicating that $\boldsymbol{z}\in\mathcal{L}_{\infty}$. Meanwhile, referring to the design of $\chi(t)$ shown in equation (\ref{chit}), one has $\chi(t)\le \sqrt{J_{\min}}$ and hence $\chi(t)\in\mathcal{L}_{\infty}$. By utilizing the Barbalat's Lemma, one has $\lim_{t\to+\infty}\dot{V}_{z}=0$.
On the other hand, based on the result of equation (\ref{Vall}), one has $\dot{V}(t)\le0$, and hence $V(t)\in\mathcal{L}_{\infty}$. This, in turn, shows that $x_{e}(t)$, $\varepsilon_{q}(t)$, $\boldsymbol{R}_{\omega}\boldsymbol{\omega}_{s}$, $g_{a}(t)$, $\tilde{\boldsymbol{\omega}}_{s}\in\mathcal{L}_{\infty}$. By invoking Barbalat's Lemma, one has $\lim_{t\to+\infty}V(t) = 0$. Therefore, $x_{e}$, $\boldsymbol{\omega}_{s}$, $\tilde{\boldsymbol{\omega}}\in\mathcal{L}_{2}\cap\mathcal{L}_{\infty}$, while $\boldsymbol{J}^{-1}\boldsymbol{\Psi}\boldsymbol{z}$ and $\boldsymbol{J}^{-1}\boldsymbol{\Psi}\tilde{\boldsymbol{\theta}}\in\mathcal{L}_{2}$. 

	Meanwhile, potential functions $U(t)$, $V_{\omega}(t)$ and the BLF function $V_{B}(t)$ are bounded for $\forall t\in\left[t_{0},+\infty\right)$. Owing to the barrier condition Lemma \cite{tee2009barrier}, the constraint-satisfying set $\mathcal{S}_{p}$, $\mathcal{S}_{\omega}$ and $\mathcal{S}_{e}$ remain invariant with $\boldsymbol{B}_{i}(t_{0})\in\mathcal{S}_{p}$, $\boldsymbol{\omega}(t_{0})\in\mathcal{S}_{\omega}$ and $x_{e}(t_{0})\in\mathcal{S}_{e}$, and thereby indicating that constraints (\ref{pointcons})(\ref{angcons}) and (\ref{pfe}) are respected during the whole control process. 
	
	Consequently, system will converge to an implicit manifold $\boldsymbol{\Psi}\tilde{\boldsymbol{\theta}}=\boldsymbol{0}$ as $t\to+\infty$, and this further provides the conclusion that $x_{e}(t)$, $\varepsilon_{q}(t)$ and $\boldsymbol{\omega}_{s}$ will be asymptotically stabled, with all constraints satisfied. This completes the proof of Theorem \ref{T1}.
\end{proof}
\begin{remark}
	It should be noticed that the value of $S_{\omega}$, $J_{\min}$ and $\chi(t)$ are not incorporated into the implementation of the controller, which are introduced only for stability analysis.
\end{remark}

\section{FURTHER DISCUSSION ON THE EFFECT OF SPPF-based PPC}\label{Discussion}
Based on the Lyapunov energy perspective, this section serves as an additional statement of the theoretic effect of the proposed SPPF-based PPC strategy.
\begin{figure}[hbt!]
	\centering 
	\includegraphics[width=0.3\textwidth]{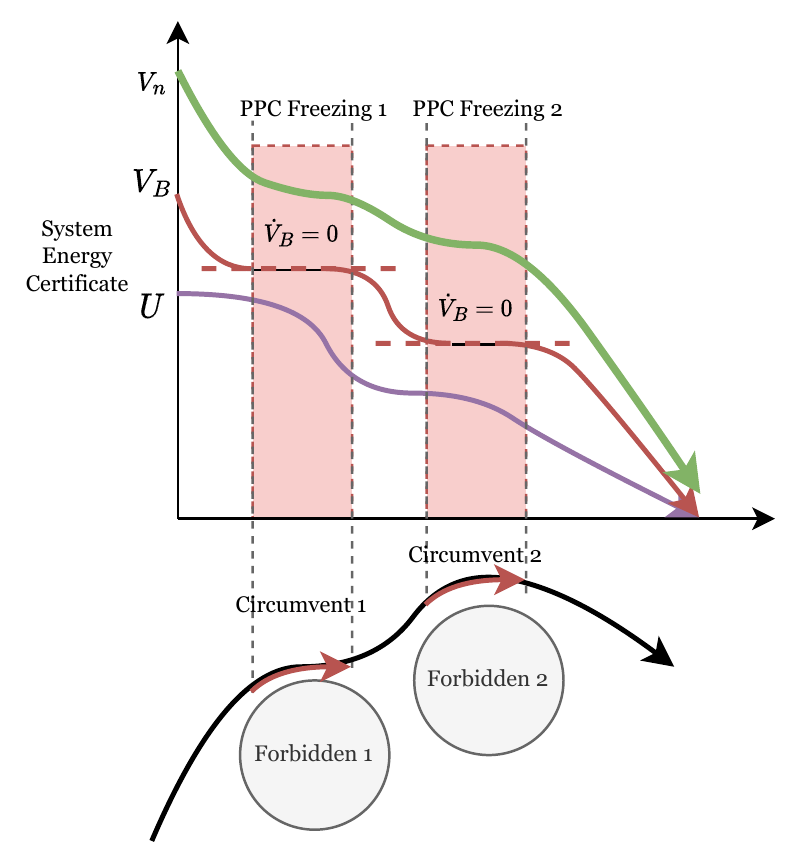}
	\caption{Sketch-map of the system switching process. The above graph shows the time-evolution of the system energy $V_{B}$, $U_{a}$ and $V_{n} = V_{B}+U_{a}$; the below graph shows the time-evolution of the boresight pointing trajectory. As we claimed about the freezing mechanism, when the boresight vector bypasses the forbidden zone (not the only circumstance), $\dot{V}_{B} = 0$ and hence excludes the performance consideration temporarily}       
	\label{freezingmechanism}   
\end{figure}	
Considering a combined Lyapunov energy function related to the "pointing error" as $V_{n}(t) = U_{a}(t) + V_{B}(t)$, i.e. the sum of the potential function for pointing error control and the PPC's Barrier Lyapunov function. It can be observed that $\dot{V}_{n} \le 0$ is necessary for the system's stability.

\textbf{1.} For $\Omega_{Q}(t) = 0$, we have $\dot{\rho}_{q}(t) = -k_{\rho}\left(\rho_{q}(t)-\rho_{\infty}\right)$. Taking the time-derivative of $V_{B}$ and $V_{n}$ yields:
\begin{equation}\label{dVcase1}
	\begin{aligned}
		\dot{V}_{B} &= \frac{R_{\rho}}{\rho_{q}}\left[\left(\boldsymbol{r}^{\times}_{b}\boldsymbol{B}_{b}\right)^{\text{T}}\boldsymbol{\omega}_{s}-\frac{-k_{\rho}(\rho_{q}-\rho_{\infty})}{\rho_{q}}\varepsilon_{q}\right]\\
		\dot{V}_{n} &= \left(k_{a}+\frac{R_{\rho}}{\rho_{q}}\right)\left(\boldsymbol{r}^{\times}_{b}\boldsymbol{B}_{b}\right)^{\text{T}}\boldsymbol{\omega}_{s}-\frac{R_{\rho}\dot{\rho}_{q}}{\rho_{q}}\varepsilon_{q}
	\end{aligned}
\end{equation}
It can be observed that the PPC scheme adds an additional high-gain  term $\frac{R_{\rho}(t)}{\rho_{q}(t)}$ on the APF's attraction gradient. By simultaneously designing an appropriate controller to cancel the right term $-\frac{R_{\rho}\dot{\rho}_{q}}{\rho_{q}}\varepsilon_{q}$ and compensate the overall gradient $(k_{a}+\frac{R_{\rho}}{\rho_{q}})\boldsymbol{r}^{\times}_{b}\boldsymbol{B}_{b}$, the PPC scheme will serve as an additional "accelerator" to ensure the satisfaction of PFE constraint, thus facilitating the convergence of $x_{e}$ and helps the system against with external disturbances. 
%

\textbf{2.}
For $\Omega_{Q}(t) = 1$, considering the time-derivative of $V_{B}$, we have:
\begin{equation}\label{switch}
	\begin{aligned}
		\dot{V}_{B} &= \frac{R_{\rho}}{\rho_{q}}\left[\left(\boldsymbol{r}^{\times}_{b}\boldsymbol{B}_{b}\right)^{\text{T}}\boldsymbol{\omega}_{s}-\frac{\dot{\rho}_{q}}{\rho_{q}}x_{e}\right]\\
	\end{aligned}
\end{equation}
 Correspondingly, it can be obtained that $\dot{V}_{B} = 0$, indicating that the PPC system has no impact on the system total energy. Accordingly, the intrinsic conflict is avoided by the "freezing" mechanism as the time-derivative of $V_{n}$ is only governed by the nominal APF controller at current moment, i.e., $\dot{V}_{n} = \dot{U}_{a}$. 

\textbf{3.} 
Considering the condition that $\Omega_{Q}(t)\in\left(0,1\right)$ holds, it can be further obtained that:
\begin{equation}
	\begin{aligned}\label{VB3}
		\dot{V}_{B}
		&=\left(1-\Omega_{Q}(t)\right)\left[\frac{R_{\rho}}{\rho_{q}}\dot{x}_{e}-\frac{R_{\rho}}{\rho_{q}}\left[-k_{\rho}(\rho_{q}-\rho_{\infty})\right]\right]
	\end{aligned}
\end{equation}
Notably, $\dot{V}_{B}$ can be rearranged as follows:
\begin{equation}
	\dot{V}_{B} = \left(1-\Omega_{Q}(t)\right)dV_{B1}
\end{equation}
where $dV_{B1}$ is defined as $dV_{B1}\triangleq\frac{R_{\rho}}{\rho_{q}}\dot{x}_{e}-\frac{R_{\rho}}{\rho_{q}}\left[-k_{\rho}(\rho_{q}-\rho_{\infty})\right]$, denoting the expression of $\dot{V}_{B}$ when $\Omega_{Q}(t) = 0$ holds, as depicted in equation (\ref{dVcase1}).
Correspondingly, it can be inferred that $\dot{V}_{B}$ changes with a coefficient $1-\Omega_{Q}(t)$ as $\Omega_{Q}(t)$ switches from $0$ to $1$, and $\dot{V}_{B} = 0$ holds when $\Omega_{Q}(t)$ increases to $1$ finally. Note that the switching of $\Omega_{Q}(t)$ is smooth and differentiable, hence such a switching of $\dot{V}_{B}$ is also smooth and differentiable.
A sketch-map of this process is illustrated in Figure \ref{freezingmechanism} for a further explanation.

\section{NUMERICAL SIMULATION AND ANALYSIS}\label{Simulation}
In this section, we present a series of simulation results to validate the efficacy of our proposed scheme. Firstly, we conduct two sets of simulation results to demonstrate the fundamental effect of the proposed adaptive APF-PPC control scheme, considering scenarios involving $2$ or $3$ forbidden zones. Subsequently, we perform a comparative simulation to validate the robustness of the proposed framework and its superiority in control accuracy improvement compared with existing APF-only approaches. Moreover, we present a set of Monte Carlo simulations, showcasing the capability of the proposed scheme to address control scenarios with randomized conditions.

\subsection{Boresight Alignment Control Scenario Establishing}\label{scenario}
The spacecraft is assumed to be rigid, of which the inertial matrix $\boldsymbol{J} = \text{diag}(\left[2,2.9,2.3\right])$. The unit is $\text{kg}\cdot\text{m}^{2}$.

The external environmental disturbance $\boldsymbol{d}\in\mathbb{R}^{3}$ is modeled as follows:
\begin{equation}
	\boldsymbol{d} = 
	\begin{bmatrix}
		1e-3 \cdot \left(4\sin\left(3\omega_{p}t\right) + 3\cos\left(10\omega_{p}t\right) -4\right)\\ 	
		1e-3\left(-1.5\sin\left(2\omega_{p}t\right) + 3\cos\left(5\omega_{p}t\right) +4\right)\\ 	
		1e-3\left(3\sin\left(10\omega_{p}t\right) - 8\cos\left(4\omega_{p}t\right) +4\right)\\ 	
	\end{bmatrix}
\end{equation}
where $\omega_{p} = 0.01\text{rad/s}$ refers to the period of the external disturbance, and the unit of $\boldsymbol{d}$ is $\text{N}\cdot\text{m}$. It can be observed that $\|\boldsymbol{d}\|$ is much more larger than a normal one in actual spacecraft environment, hence is enough for the assessment of the proposed scheme's robustness. 
Meanwhile, for practical consideration, the maximum output torque is limited to $0.05$$\text{N}\cdot\text{m}$ in this paper.

Specific settings of constraints of the simulation are stated as follows:

1. \textbf{Attitude Angular Velocity Constraint}

Each $i$-th component of $|\omega_{si}(t)|$ should not exceed $\textbf{3}^{\circ}/\text{s}$ during the whole control process. Accordingly, $M_{\omega}$ is given as $M_{\omega} = 0.0524$.

2. \textbf{PFE Constraint for Guaranteed Accuracy}

Suppose that the terminal pointing accuracy should be greater than $\textbf{0.05}^{\circ}$, such that:
\begin{equation}
	\Theta(t) < 0.05^{\circ}, \text{for}\quad t\in\left[t_{\text{set}},+\infty\right) 
\end{equation}
where $t_{\text{set}}>t_{0}$ stands for a practical settling time. Since we have $\cos\Theta = 1-x_{e}$, hence this condition is equivalent to:
	$x_{e}(t) < 3.8\times10^{-7}, \text{for}\quad t\in\left[t_{\text{set}},+\infty\right) $.
Accordingly, we design the performance envelope as $k_{\rho} = 0.05$, $\rho_{\infty} = 1\times 10^{-4}$, and the initial condition is given as $\rho_{q}(t_{0})=4$;

For the following simulation, main parameters of the APF control are given as: $k_{a}=2.5$, $k_{r}=0.5$, $P^{N}_{0} = \cos(\Theta^{N}_{f}+10^{\circ})$, $P^{N}_{1}=\cos\Theta^{N}_{f}$, $k_{\omega} = 0.03$, $K_{\omega} = 0.05$. Main parameters of the BLF is $k_{B} = 0.4$. Notably, this suggest that the acting domain is $10$ degrees larger than the minimum permitted angle. Main parameters of the adaptive strategy are $C_{\beta} = 0.05$,  $K_{f}=30$. The coefficient of the RSM function is provided as $K_{s}=100$. For the switching indicator, we set $S^{N}_{f0} = \cos(\Theta^{N}_{f}+10^{\circ})$, $S^{N}_{f1} = \cos(\Theta^{N}_{f}+5^{\circ})$, $S^{i}_{\omega0} = 0.8M^{2}_{\omega}$, $S^{i}_{\omega1}=0.9M^{2}_{\omega}$, $S_{\varepsilon0}=0.9$, $S_{\varepsilon1}=0.95$, $S_{\rho0}=\rho_{\infty}+0.001$, $S_{\rho1}=\rho_{\infty}+0.002$.

In the following analysis, the initial condition of the spacecraft is specified as $\boldsymbol{A}_{bi} = \boldsymbol{I}_{3}$, $\boldsymbol{\omega}_{s} = \boldsymbol{0}_{3}$. We further randomly set a desired pointing direction $\boldsymbol{r}_{i}$ expressed in $\mathfrak{R}_{i}$ as:
\begin{equation}
	\boldsymbol{r}_{i} = \left[-0.8617,0.4975,-0.0995\right]^{\text{T}}
\end{equation}

\subsection{Normal Case Simulation: Two Pointing-forbidden Zones}\label{TwoForbidden}

Firstly, we present a simulation campaign, considering the circumstance that there exist $2$ forbidden directions. These forbidden zones are specified as follows:
\begin{equation}
	\begin{aligned}
		\boldsymbol{f}^{1}_{i} &= \left[0.5715,0.8165,0.0816\right]^{\text{T}}\\
		\boldsymbol{f}^{2}_{i} &= \left[-0.3369,0.8422,-0.4211\right]^{\text{T}}
	\end{aligned}
\end{equation}
\begin{figure}[hbt!]
	\centering 
	\includegraphics[width=0.4\textwidth]{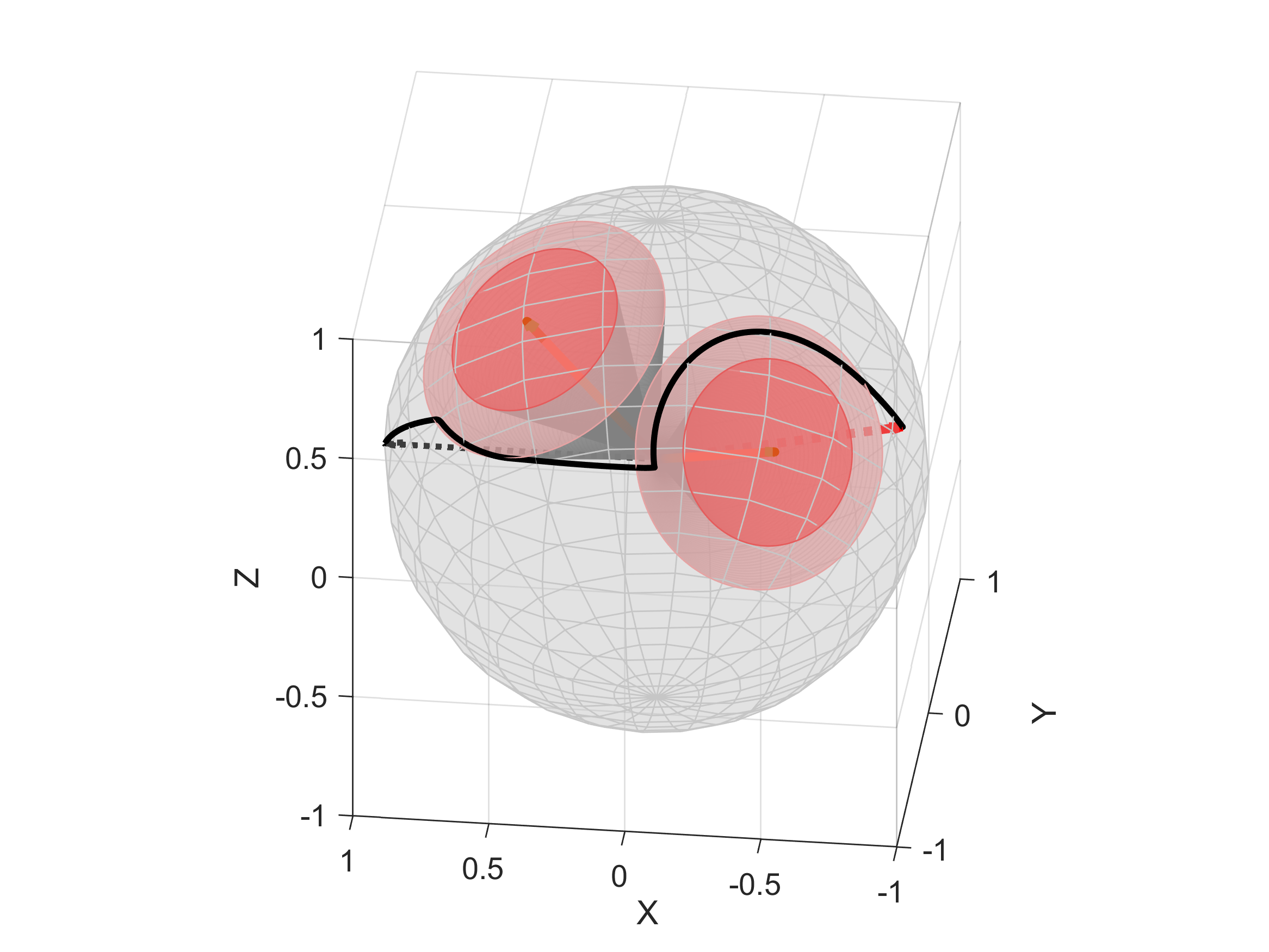}
	\caption{Pointing Trajectory of $\boldsymbol{B}_{i}$ expressed in the inertia frame $\mathfrak{R}_{i}$ (Two Forbidden Cones)}   
	\label{S1}
			\centering 
	\includegraphics[width=0.5\textwidth]{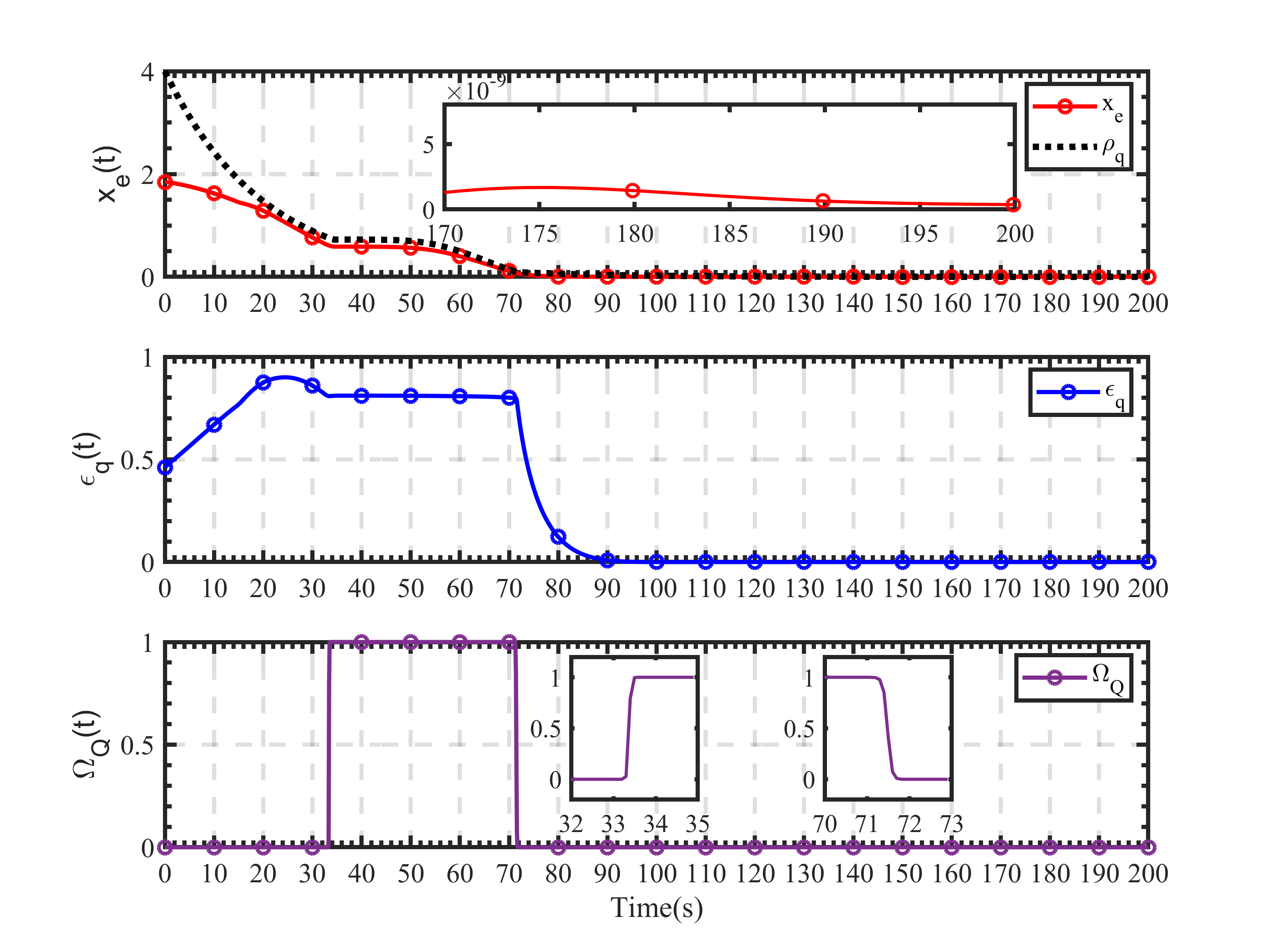} 
	\caption{Time evolution of the pointing error variable $x_{e}(t)$, transformed error variable $\varepsilon_{q}(t)$ and the switching indicator $\Omega_{Q}(t)$ (Two Forbidden Cones)}    
	\label{xe1}
\end{figure}

Permitted minimum angles are given as $\Theta^{1}_{f}=\Theta^{2}_{f} = 20^{\circ}$.
The pointing trajectory of the boresight vector $\boldsymbol{B}_{i}$ (resolved in $\mathfrak{R}_{i}$) is illustrated in a unit sphere $\mathbb{S}^{2}$, as shown in Figure \ref{S1}. The black vector stands for the initial position of the boresight vector, while the desired direction is illustrated in red. The strict pointing-forbidden regions are illustrated in dark red, while acting domains of repulsion fields are illustrated in light red, corresponding to each $N$-th forbidden zone.

The time-evolution of the pointing error variable $x_{e}(t)$, the PPC transformed variable $\varepsilon_{q}(t)$ and the time-evolution of the overall switching indicator $\Omega_{Q}(t)$ are illustrated in Figure \ref{xe1}. In the first subfigure of Figure \ref{xe1}, the red solid line stands for the time-evolution of $x_{e}(t)$, while the black-dotted line stands for the time-evolution of the performance envelope $\rho_{q}(t)$. The time-evolution of the attitude angular velocity $\boldsymbol{\omega}_{s}$ and the control input $\boldsymbol{u}$ is presented in Figure \ref{W1} and \ref{u1}, respectively.

From Figure \ref{xe1}, it can be observed that $\varepsilon_{q}(t)<1$ holds for the whole control process, indicating that the PFE constraint (\ref{pfe}) is satisfied. It can be also discovered that without the design of SPPF, utilizing a typical exponential-converged performance function will not able to ensure the satisfaction of PFE constraint. Meanwhile, the pointing error variable $x_{e}(t)$ converges to be smaller than $x_{e}(t) < 2\times 10^{-9}$, i.e., $\Theta<0.0036^{\circ}$, indicating that the provided performance criteria $x_{e}\le 3.8\times10^{-7}$, is achieved. 

\begin{figure}[hbt!]
		\centering 
	\includegraphics[width=0.5\textwidth]{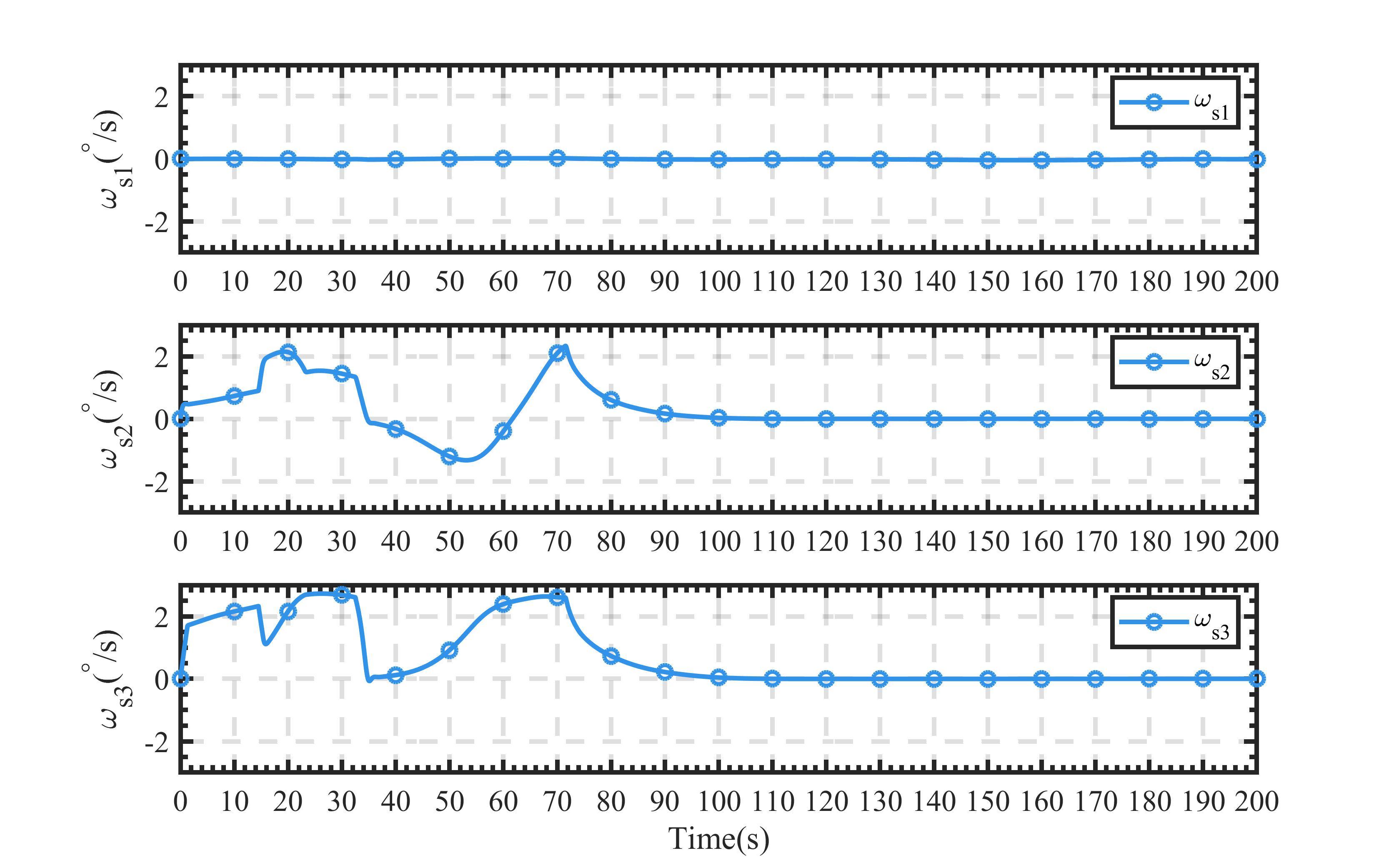}
	\caption{Time evolution of the attitude angular velocity $\boldsymbol{\omega}_{s}(t)$ (Two Forbidden Cones)}       
	\label{W1} 
		\centering 
	\includegraphics[width=0.5\textwidth]{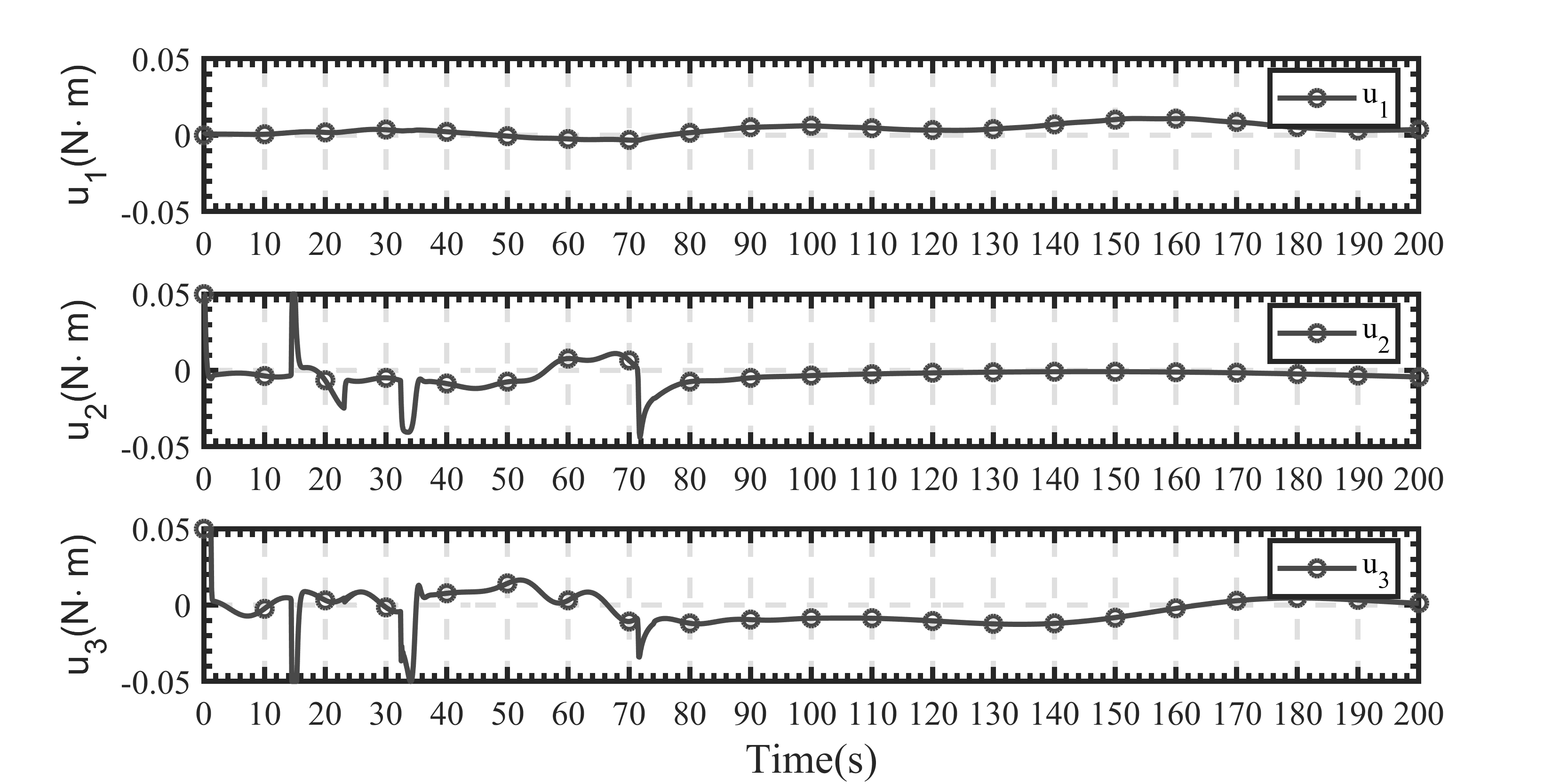}
	\caption{Time evolution of the control input $\boldsymbol{u}(t)$ (Two Forbidden Cones)}       
	\label{u1}  
\end{figure}

The maximum magnitude of the angular velocity is lower than $2.5^{\circ}/\text{s}$ for each component, showing that the attitude angular velocity limitation (\ref{angcons}) is strictly respected.
Meanwhile, from Figure \ref{S1}, it can be observed that all pointing-forbidden zones are not invaded during the attitude maneuvering, and the pointing-forbidden constraint (\ref{pointcons}) is satisfied for each forbidden region.
%

Notably, from the time-evolution of the PPC transformed variable $\varepsilon_{q}(t)$ and the switching indicator $\Omega_{Q}$ that given in Figure \ref{xe1}, it can be observed that when $\Omega_{Q}$ switches to $1$ at around $t=34s$, the value of $\varepsilon_{q}(t)$ remains unchanged. After $\Omega_{Q}$ switches to $\Omega_{Q} = 0$ again, the performance envelope (the black dotted line) converges exponentially again, and the value of $\varepsilon_{q}(t)$ quickly converges with the guidance of the PPC controller. This stands for the validation of the theoretical analysis of the freezing mechanism, as we discussed in Subsection \ref{APFPPC}-\ref{SPPFMECHANISM}.

\subsection{Normal Case Simulation: Three Pointing-forbidden Zones}\label{3obs}
Secondly, we consider the circumstance that there exist $3$ obstacle regions, expressed as:
\begin{equation}
	\begin{aligned}
		\boldsymbol{f}^{1}_{i} &= \left[0.6529,0.7255,0.2176\right]^{\text{T}}\\
		\boldsymbol{f}^{2}_{i} &= \left[-0.4402,0.8805,0.1761\right]^{\text{T}}\\
		\boldsymbol{f}^{3}_{i} &= \left[0.0741,0.7412,-0.6671\right]^{\text{T}}
	\end{aligned}
\end{equation}
\begin{figure}[hbt!]
	\centering 
	\includegraphics[width=0.4\textwidth]{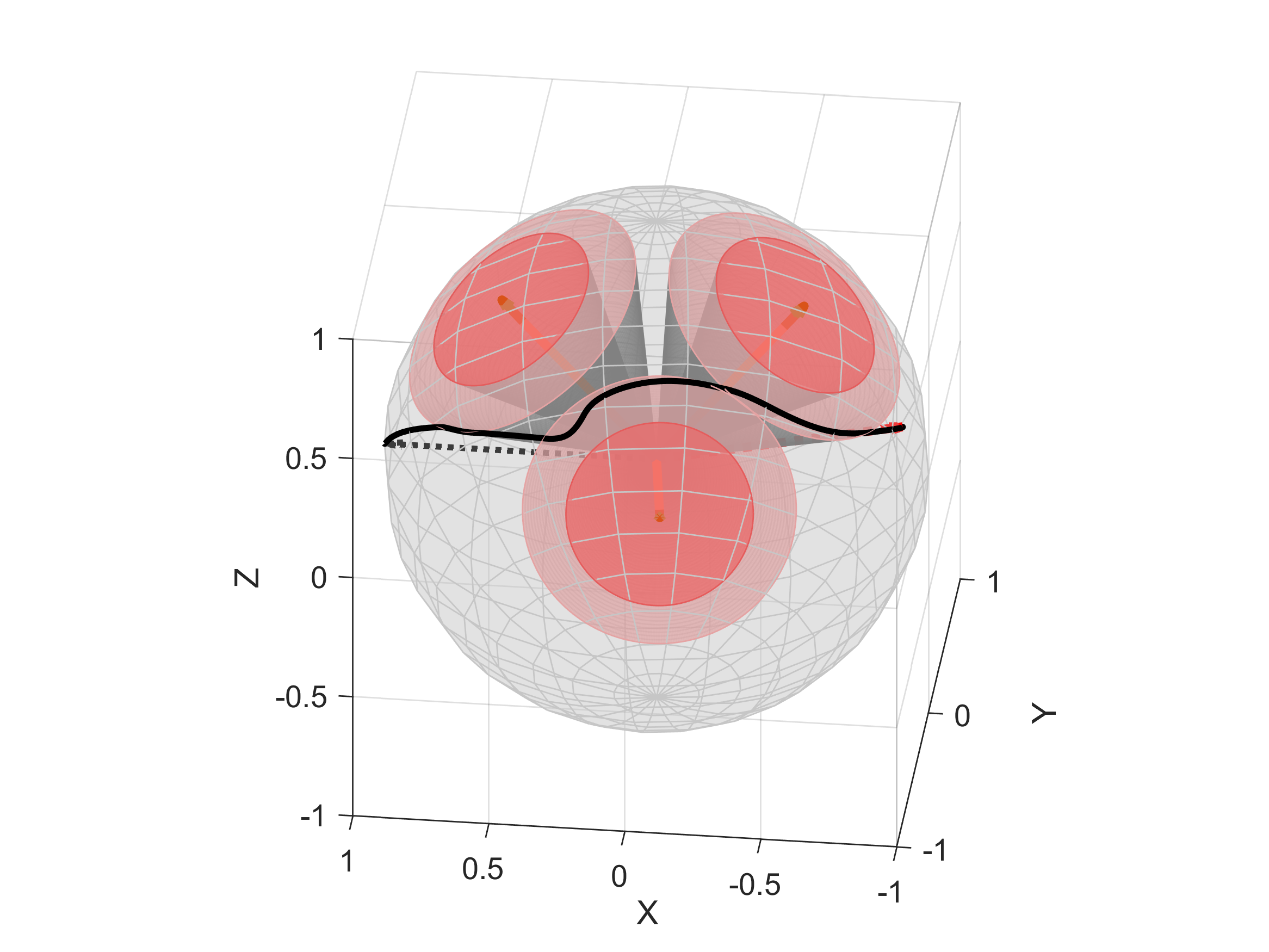}
	\caption{Pointing Trajectory of $\boldsymbol{B}_{i}$ expressed in the inertia frame $\mathfrak{R}_{i}$ (Three Forbidden Cones)}       
	\label{S2} 
\end{figure}

The permitted minimum angles are given as $\Theta^{1}_{f}=\Theta^{2}_{f}=\Theta^{3}_{f}=20^{\circ}$.
 Figure \ref{S2} shows the trajectory of the boresight vector $\boldsymbol{B}_{i}$ during the control process.
The time-evolution of the pointing error variable $x_{e}(t)$, the PPC transformed variable $\varepsilon_{q}(t)$ and the time-evolution of the switching indicator $\Omega_{Q}(t)$ are illustrated in Figure \ref{xe2}. Similarly, the black-dotted line in the first subfigure of Figure \ref{xe2} stands for the performance function $\rho_{q}(t)$. The time-evolution of the angular velocity $\boldsymbol{\omega}_{s}$ is presented in Figure \ref{W2}, while the control input $\boldsymbol{u}$ is illustrated in Figure \ref{u2}.

Also, From Figure \ref{W2}, it can be observed that the maximum magnitude of $\boldsymbol{\omega}_{s}$ is restricted to $\max(|\omega_{si}(t)|)<2.5^{\circ}/\text{s}$, while $\varepsilon_{q}(t)<1$ holds during the whole control process, showing that the PFE constraint (\ref{pfe}) is satisfied. The boresight control achieves an accuracy of $x_{e} < 3\times 10^{-10}$, i.e., $\Theta<0.0012^{\circ}$, indicating that the desired pointing performance criteria is achieved. Meanwhile, Figure \ref{S2} shows that $\boldsymbol{B}_{i}$ avoids each forbidden zone and keeps the satisfaction of pointing-forbidden constraint (\ref{pointcons}).
\begin{figure}[hbt!]
		\centering
	\includegraphics[width=0.5\textwidth]{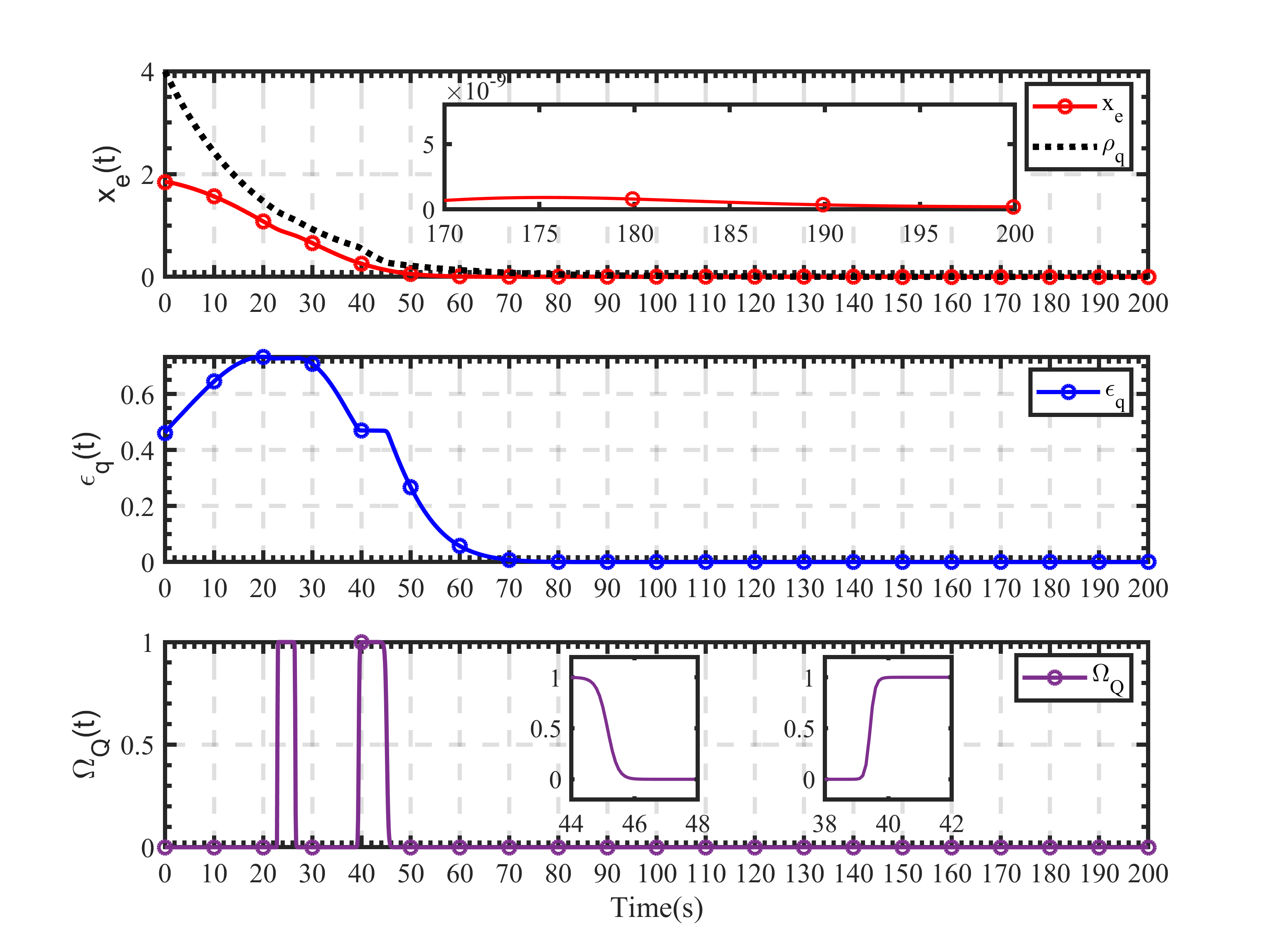}
	\caption{Time evolution of the pointing error variable $x_{e}(t)$, transformed error variable $\varepsilon_{q}(t)$ and the switching indicator $\Omega_{Q}(t)$ (Three Forbidden Cones)}       
	\label{xe2} 
					\centering 
	\includegraphics[width=0.5\textwidth]{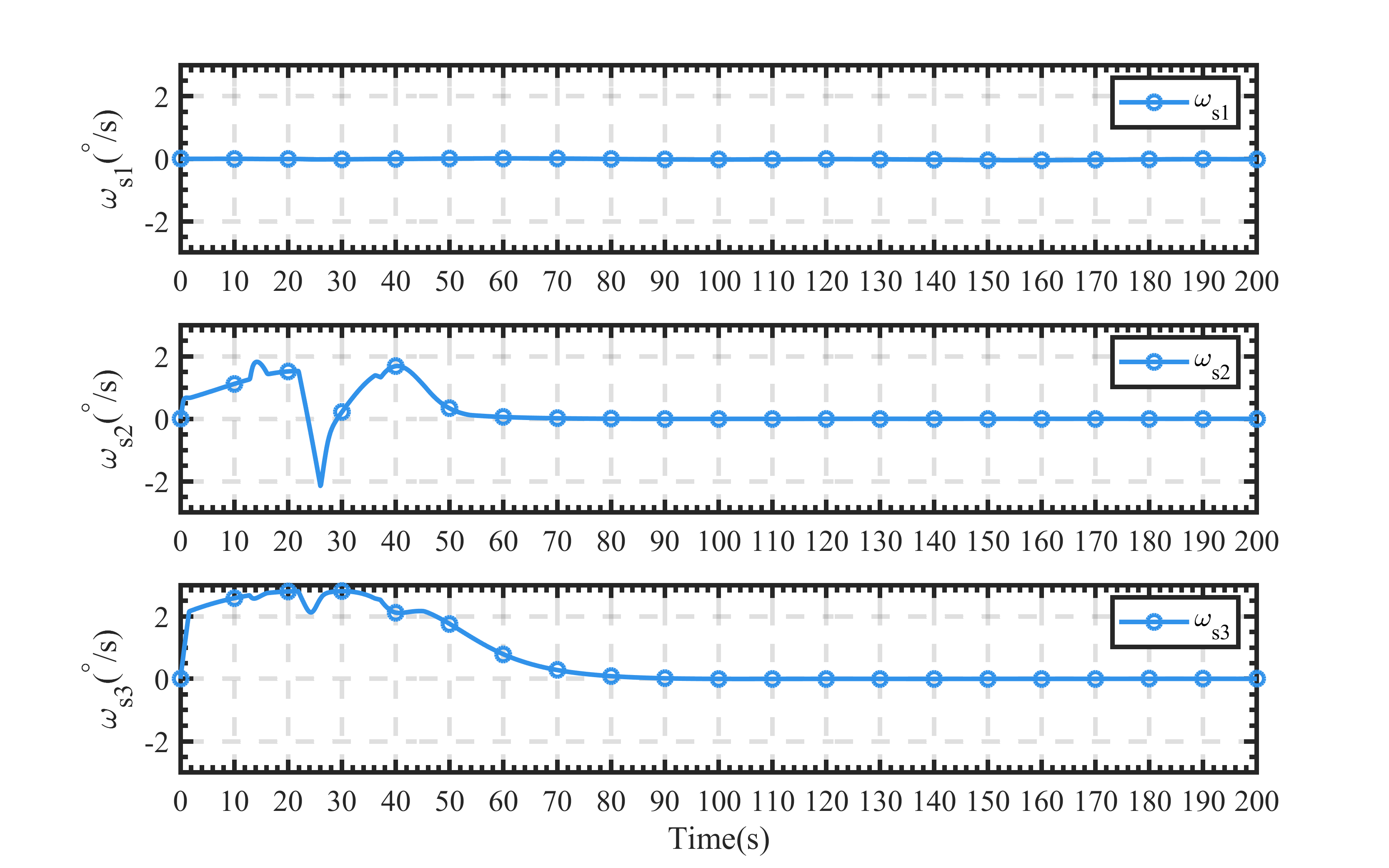}
	\caption{Time evolution of the attitude angular velocity $\boldsymbol{\omega}_{s}(t)$ (Three Forbidden Cones)}       
	\label{W2} 
			\centering 
	\includegraphics[width=0.5\textwidth]{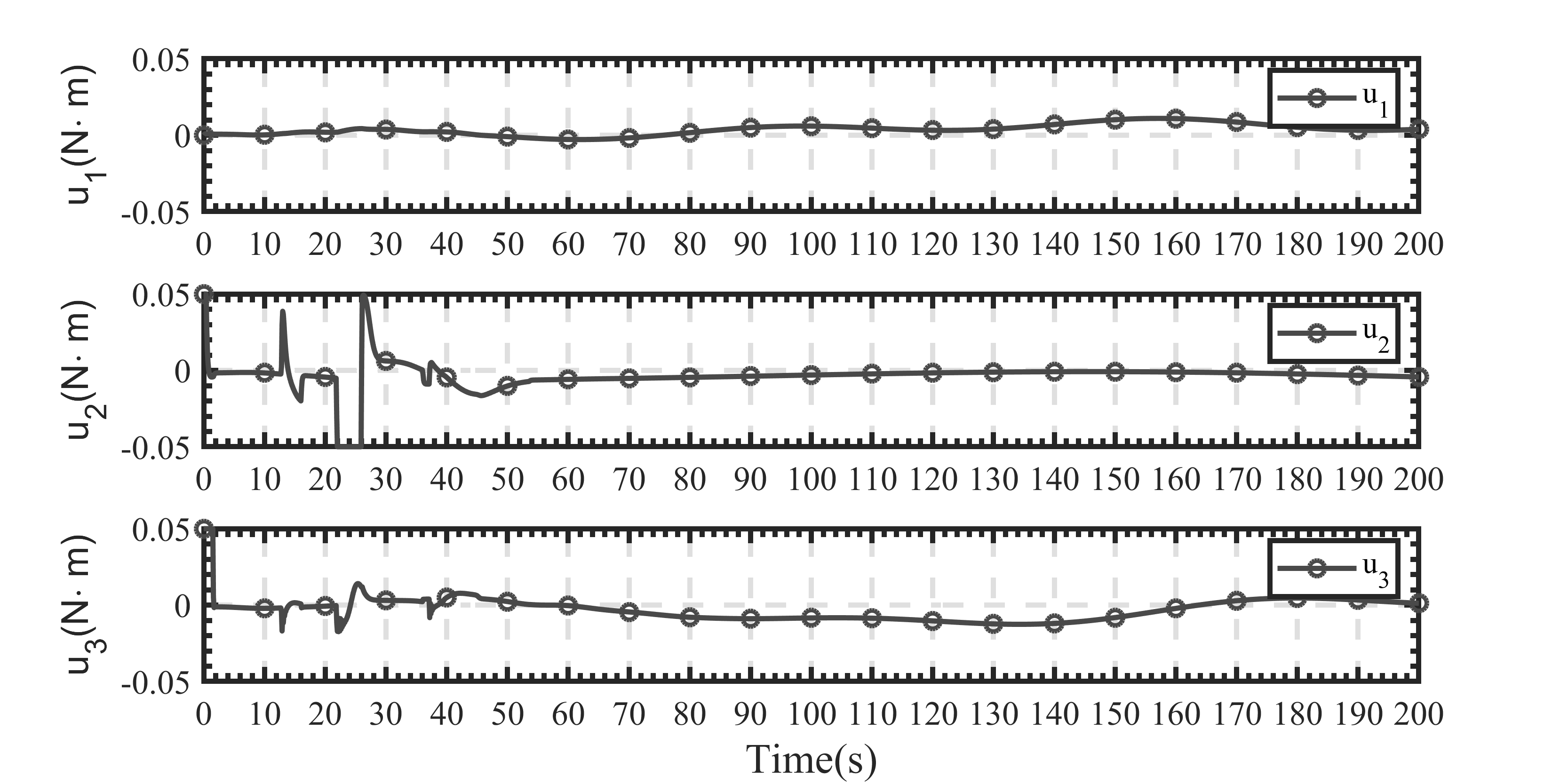}
	\caption{Time evolution of the control input $\boldsymbol{u}(t)$ (Three Forbidden Cones)}       
	\label{u2} 
\end{figure}
On the other hand, the freezing mechanism also works during the control process. Note that there exist two time intervals that the value of $\varepsilon_{q}$ remains unchanged anymore, i.e., $\dot{\varepsilon}_{q} = 0$ holds, as shown in the Figure \ref{xe2}. This indicates the working of the proposed "PPC Freezing" mechanism, corresponding to the circumstance that $\Omega_{Q}(t)=1$.

\subsection{Comparison Simulation 1: Effectiveness of the Proposed APF-PPC Composite Control}
In this subsection, we emphasize validating the proposed APF-PPC composite control scheme's efficacy in enhancing pointing accuracy, in order to validate our analysis in Section \ref{Discussion}, which supports the superiority of employing the PPC control into the APF control. A comparative analysis is presented with conventional APF-only controllers.

 We first consider a non-PPC version of the proposed controller by directly removing the PPC part of the proposed control law (\ref{controllaw}), and regards it as a benchmark controller, denoted as "\textbf{NOPPC1}".

 Subsequently, to validate that the improvement in pointing accuracy is specifically originated from the proposed SPPF-based PPC control scheme, rather than an outcome of the potential function design, we further introduce another benchmark APF-based controller that presented in \cite{dongare2021attitude}, and denote it as "\textbf{NOPPC2}". 
  Notably, since the original controller in \cite{dongare2021attitude} does not introduce an adaptive strategy, we combine it with a same I\&I adaptive structure as the proposed one to make the comparison convincing.

  \begin{figure}[hbt!]
	\centering 
	\includegraphics[width=0.4\textwidth]{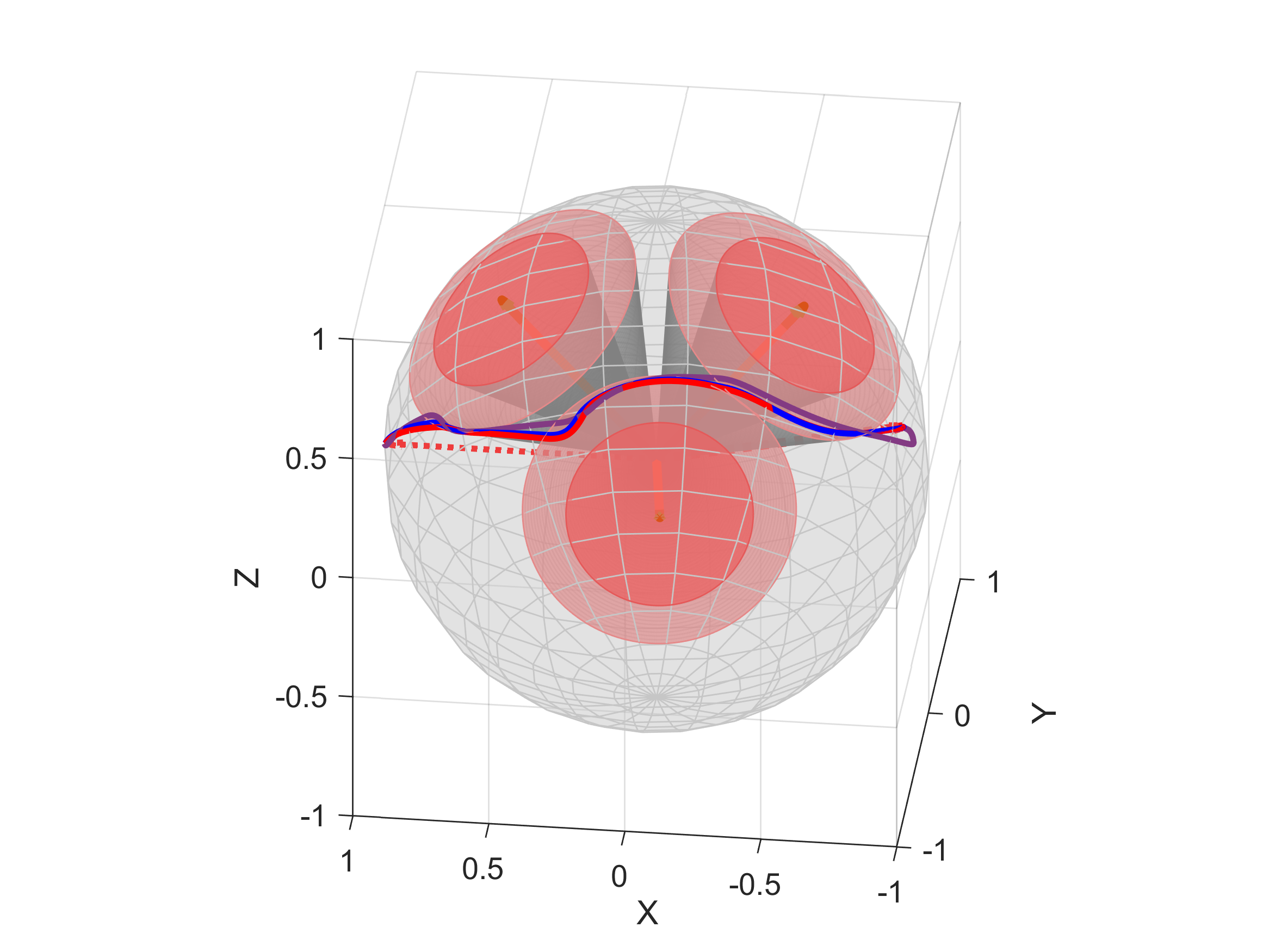}
	\caption{Pointing Trajectories of the boresight vector $\boldsymbol{B}_{i}$, corresponding to the proposed controller (red), NOPPC1 benchmark controller (blue) and the NOPPC2 benchmark controller (purple) (Comparison Simulation)}       
	\label{SCompare}
	\centering 
	\includegraphics[width=0.5\textwidth]{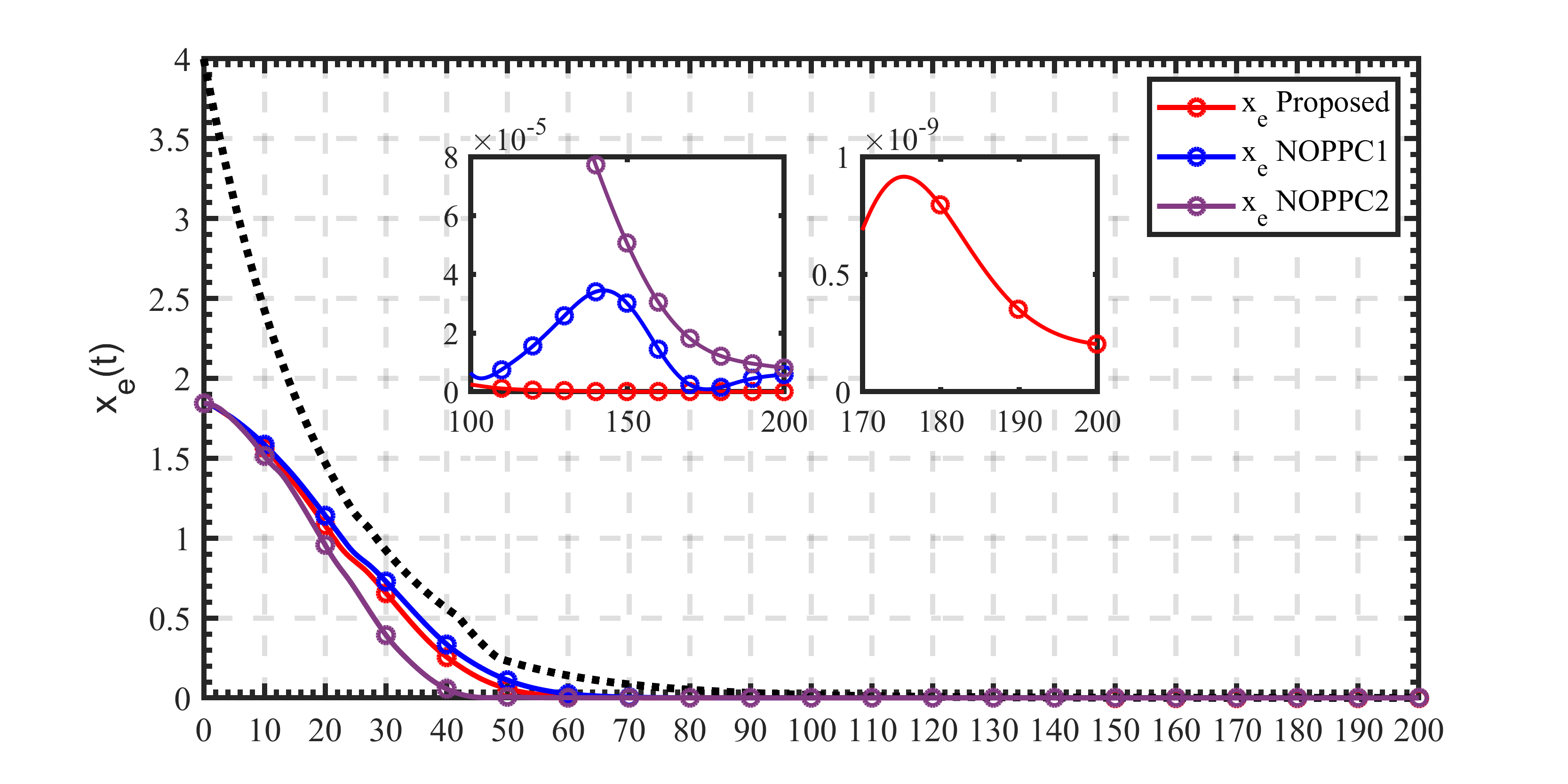}
	\caption{Time evolution of the pointing error variable $x_{e}(t)$, corresponding to the proposed controller (red), NOPPC1 benchmark controller (blue) and the NOPPC2 benchmark controller (purple) (Comparison Simulation)}       
	\label{xeCompare} 
\end{figure}

Figure \ref{SCompare} illustrates the pointing trajectory of the boresight vector $\boldsymbol{B}_{i}$ on the $\mathbb{S}_{2}$ unit sphere, corresponding to the proposed controller (red solid line), the benchmark controller NOPPC1 (blue solid line) and the benchmark controller NOPPC2 (purple solid line). Figure \ref{xeCompare} presents the comparison simulation results of the time-evolution of the pointing error variable $x_{e}(t)$, while the comparison simulation result of the angular velocity $\boldsymbol{\omega}_{s}$
is shown in Figure \ref{WCompare}.
%
\begin{figure}[hbt!]
	\centering 
	\includegraphics[width=0.5\textwidth]{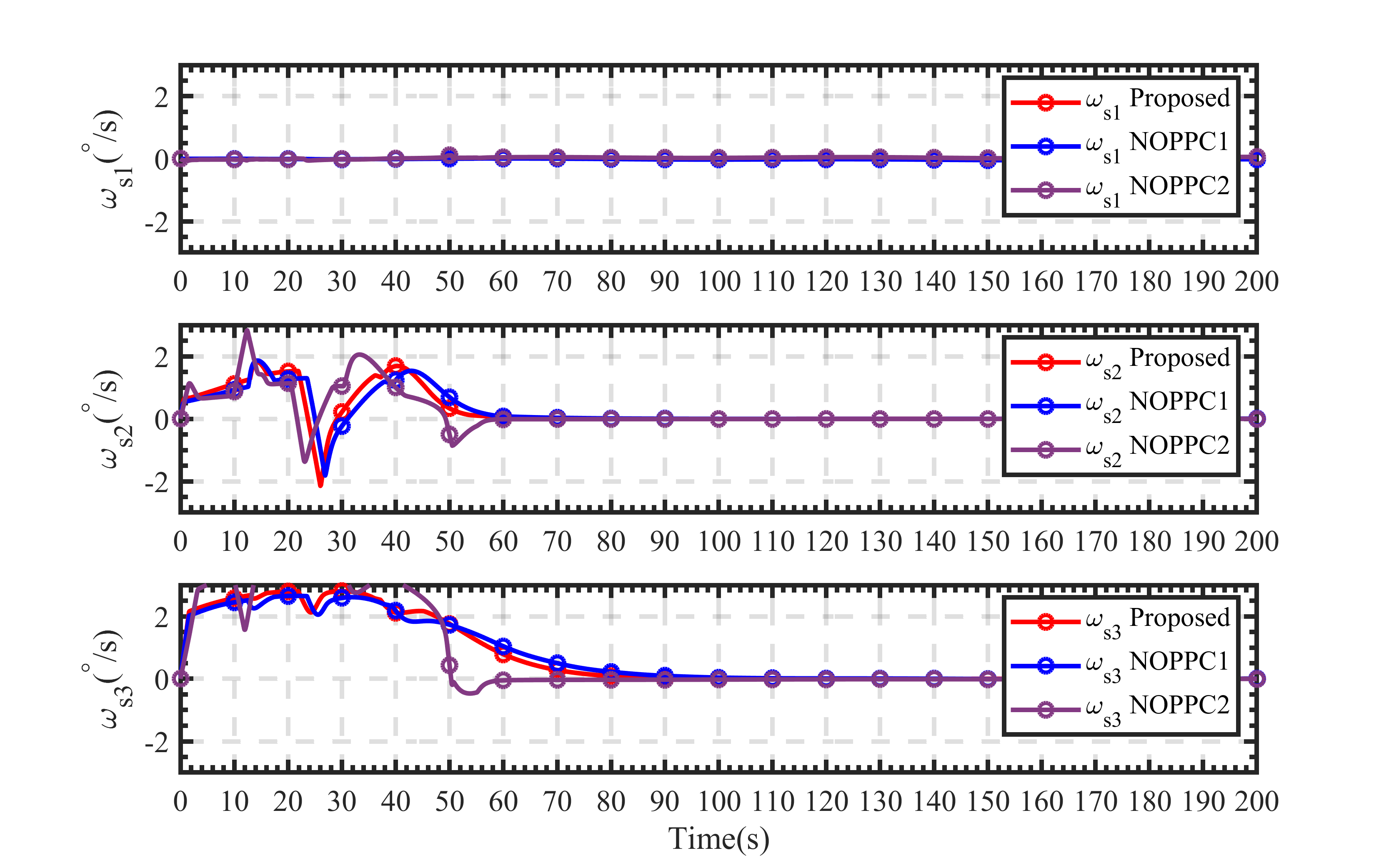}
	\caption{Time evolution of each component of the attitude angular velocity $\omega_{si}(t)(i=1,2,3)$, corresponding to the proposed controller (red), NOPPC1 benchmark controller (blue) and the NOPPC2 benchmark controller (purple) (Comparison Simulation)}       
	\label{WCompare}
\end{figure}

From Figure \ref{SCompare}, as a fundamental effect of the potential field, it can be observed that all controllers successfully circumvent the provided $3$ forbidden zones. Nevertheless, from Figure \ref{xeCompare}, it can be observed that only the proposed APF-PPC composite controller achieves an accuracy of $x_{e}\le 1\times 10^{-9}$, corresponding to $\Theta\le 0.0026^{\circ}$, which is the highest control accuracy among all these three presented controllers, and is the only one that satisfies the performance criteria that $\Theta < 0.05^{\circ}$.  

 On the other hand, note that there does not exist significant diversity between the time-evolution of $\boldsymbol{\omega}_{s}$ of these controllers, as shown in Figure \ref{WCompare}. This implies that the improvement of the control accuracy is not built based on an aggressive attitude maneuvering process.

Consequently, These results indicate that the achieved high-accuracy is not a result of the choosing of the potential function, but instead, a result of the integration of APF and PPC. To be specific, as $\Omega_{Q}$ switches to $0$, the performance envelope turns to converge exponentially at the terminal stage. As we discussed in Section \ref{Discussion}, this provides an equivalent additional high-gain term  $\frac{R_{\rho}}{\rho_{q}}$ on the attraction field, thereby providing an additional compensation that help the system against with the external disturbance.

\subsection{Monte Carlo Simulation}
In this subsection, we present a group of Monte Carlo simulation results to validate the effectiveness of the proposed controller under randomly provided simulation conditions.

For all simulation cases, the initial condition of the spacecraft is fixed, specified by the attitude quaternion and the attitude angular velocity, expressed as follows:
\begin{equation}\label{initial}
	\begin{aligned}
			\boldsymbol{q}_{s}(t_{0}) = \left[0,0.6428,0,0.7660\right]^{\text{T}},\quad
			\boldsymbol{\omega}_{s}(t_{0}) = \boldsymbol{0}_{3}
	\end{aligned}
\end{equation}

We consider the circumstance that there exists $5$ forbidden zones, specified as follows:
\begin{equation}
	\begin{aligned}
		\boldsymbol{f}^{1}_{i} &= \left[0.6529,0.7255,0.2176\right]^{\text{T}}\\
		\boldsymbol{f}^{2}_{i} &= \left[-0.4402,0.8805,0.1761\right]^{\text{T}}\\
		\boldsymbol{f}^{3}_{i} &= \left[0.0741,0.7412,-0.6671\right]^{\text{T}}\\
		\boldsymbol{f}^{4}_{i} &= \left[-0.6529,-0.7255,-0.2176\right]^{\text{T}}\\
		\boldsymbol{f}^{5}_{i} &= \left[0.4402,-0.8805,-0.1761\right]^{\text{T}}
	\end{aligned}
\end{equation}

For all simulation cases, the target pointing direction $\boldsymbol{r}_{i}$ are randomly placed on the $70^{\circ}\text{N}$-latitude circular of the unit sphere. According to the given initial condition shown in equation (\ref{initial}), the initial pointing direction of the boresight vector can be calculated as $\boldsymbol{B}_{i}(t_{0}) = \left[0.1736,0,-0.9848\right]^{\text{T}}$, located on the $80^{\circ}\text{S}$-latitude circular of the unit sphere. Therefore, it can be observed that the boresight vector need to rotate at least $150^{\circ}$ to accomplish the boresight alignment control task.
\begin{figure}[hbt!]
	\centering 
	\includegraphics[width=0.51\textwidth]{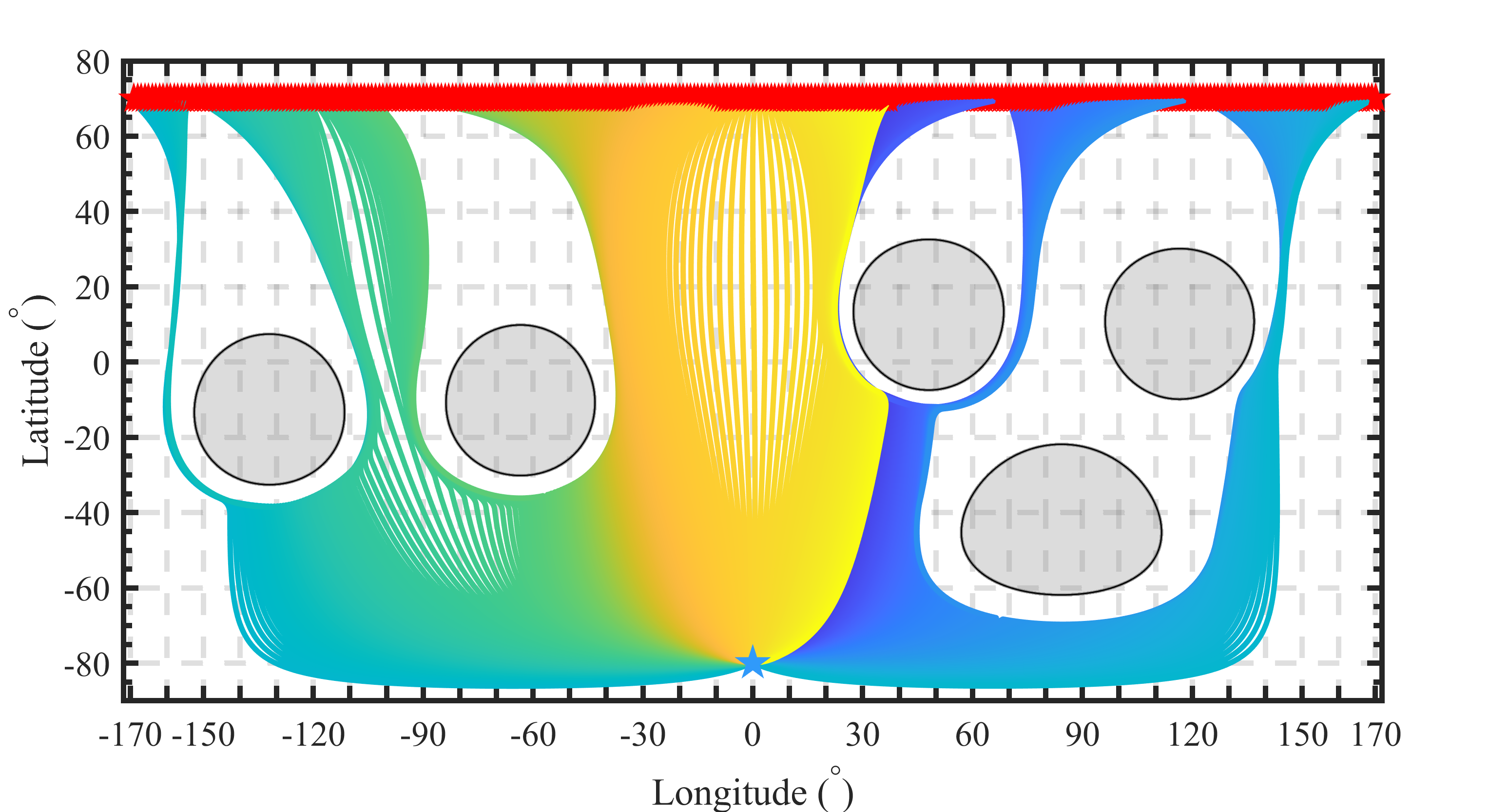}
	\caption{Mercator Projection of $\mathbb{S}_{2}$: Pointing Trajectory of the Boresight Vector $\boldsymbol{B}_{i}(t)$ of all simulation cases (Monte Carlo Simulation)}       
	\label{MonteCarlo1}  
\end{figure}

In this subsection, we use the Mercator projection \cite{pijls2001some} to map the unit sphere $\mathbb{S}_{2}$ to a two-dimensional plane, which allows us to show all simulation results in a single Figure. As depicted in Figure \ref{MonteCarlo1}, the blue star marker that placed at the bottom stands for the initial position of $\boldsymbol{B}_{i}$, while those red filled-pentagrams that placed on the top of Figure \ref{MonteCarlo1} represent randomly-chosen desired pointing directions $\boldsymbol{r}_{i}$. The gray-filled region represents the projected forbidden zone.
From Figure \ref{MonteCarlo1}, it can be observed that this control scenario requires the boresight vector $\boldsymbol{B}_{i}$ crossing the area that placed with multiple forbidden-zones.

Figure \ref{MonteCarlo1} illustrates the main result of the Monte Carlo Simulation. Subsequently, Figure \ref{MonteCarlo2} illustrates the time-evolution of the point error variable $x_{e}(t)$ and the terminal pointing accuracy of all simulation cases, while the time-evolution of the attitude angular velocity $\boldsymbol{\omega}_{s}$ of all simulation cases are illustrated in Figure \ref{MonteCarlo3}, with each component illustrated separately. Additionally, Figure \ref{MonteCarlo4} illustrates the time-evolution of the PPC transformed error variable $\varepsilon_{q}(t)$.
\begin{figure}[hbt!]
		\centering 
\includegraphics[width=0.51\textwidth]{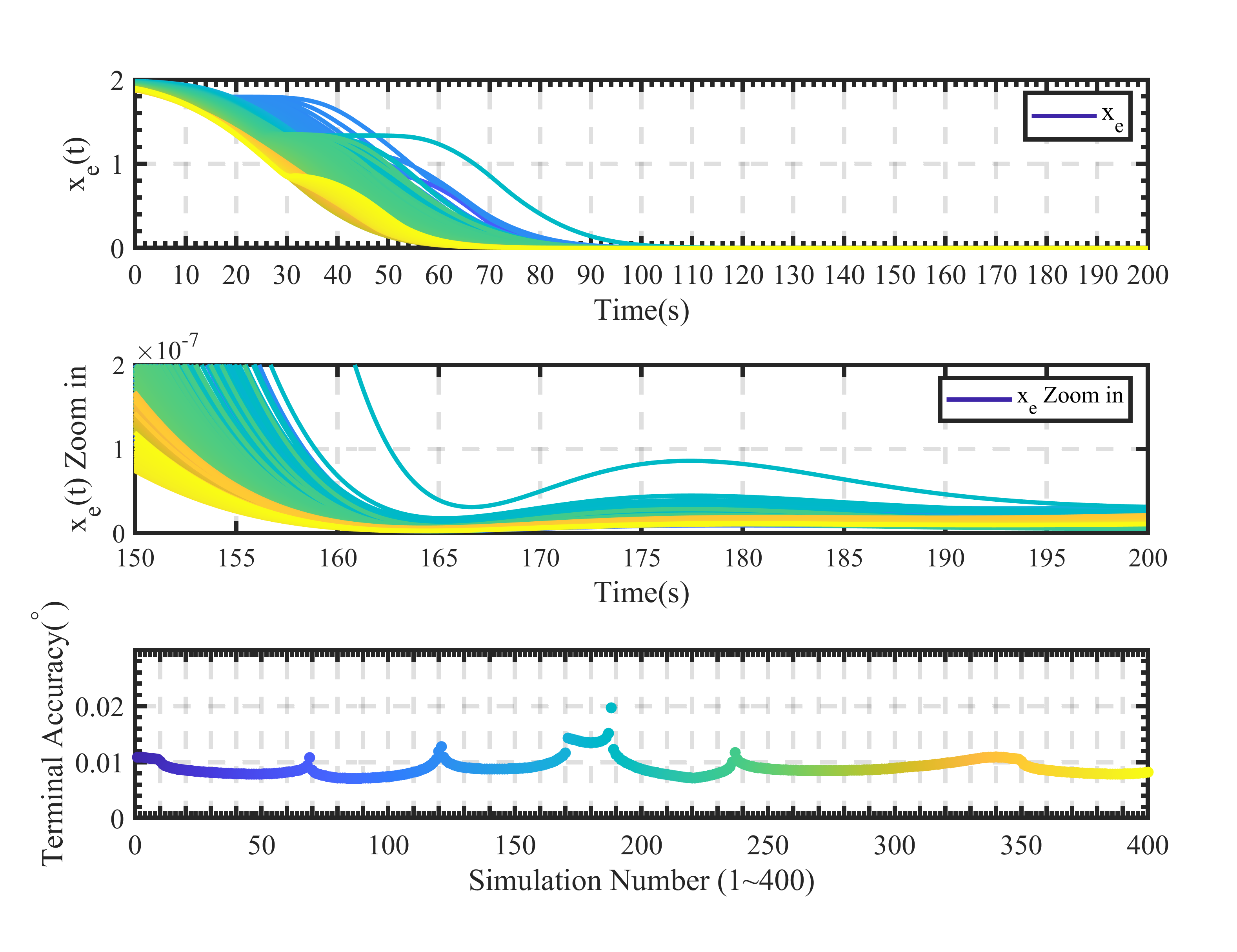}
\caption{Time-Evolution of the pointing error variable $x_{e}(t)$ and the terminal pointing accuracy of all cases (the terminal pointing accuracy is calculated as the mean value of $x_{e}(t)$ during $t\in\left[170s,200s\right]$) (Monte Carlo Simulation)}       
\label{MonteCarlo2} 
\centering 
\includegraphics[width=0.51\textwidth]{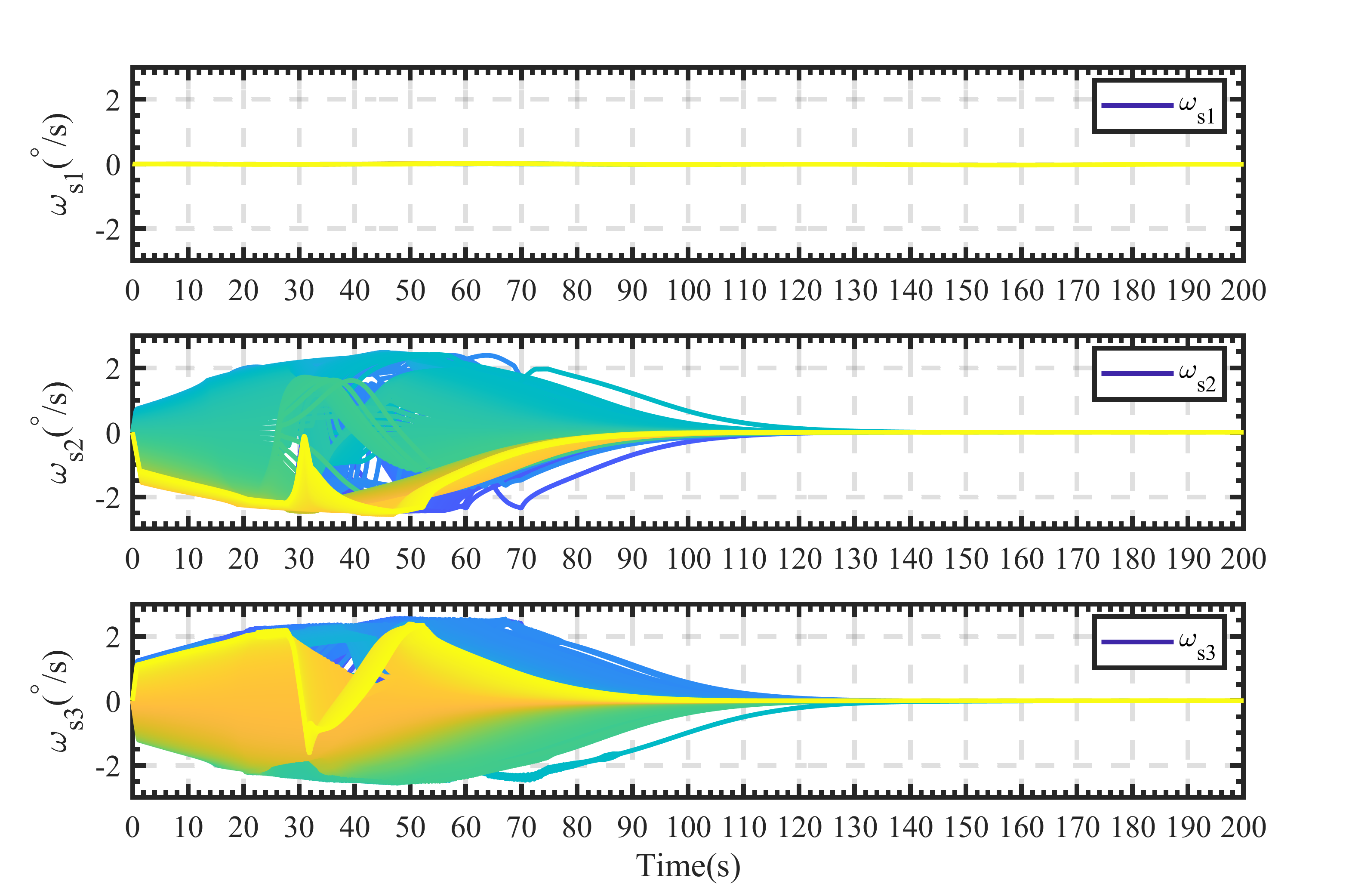}
\caption{Time-Evolution of each component of the attitude angular velocity $\omega_{si}(t)(i=1,2,3)$ of all cases (Monte Carlo Simulation)}       
\label{MonteCarlo3}
\end{figure}

From Figure \ref{MonteCarlo1}, it can be observed that all trajectories of $\boldsymbol{B}_{i}$ converge to a horizontal line, which stands for the projected $80^{\circ}\text{N}$-latitude circular. This implies that the proposed controller successfully guides the boresight vector achieving every desired position. Meanwhile, from Figure \ref{MonteCarlo2}, it can be discovered that almost all simulation cases achieve an accuracy of $0.01^{\circ}$, while the worst one is $0.02^{\circ}$. This indicates that all simulation cases satisfy the performance criteria. 
\begin{figure}[hbt!]
	\centering 
	\includegraphics[width=0.45\textwidth]{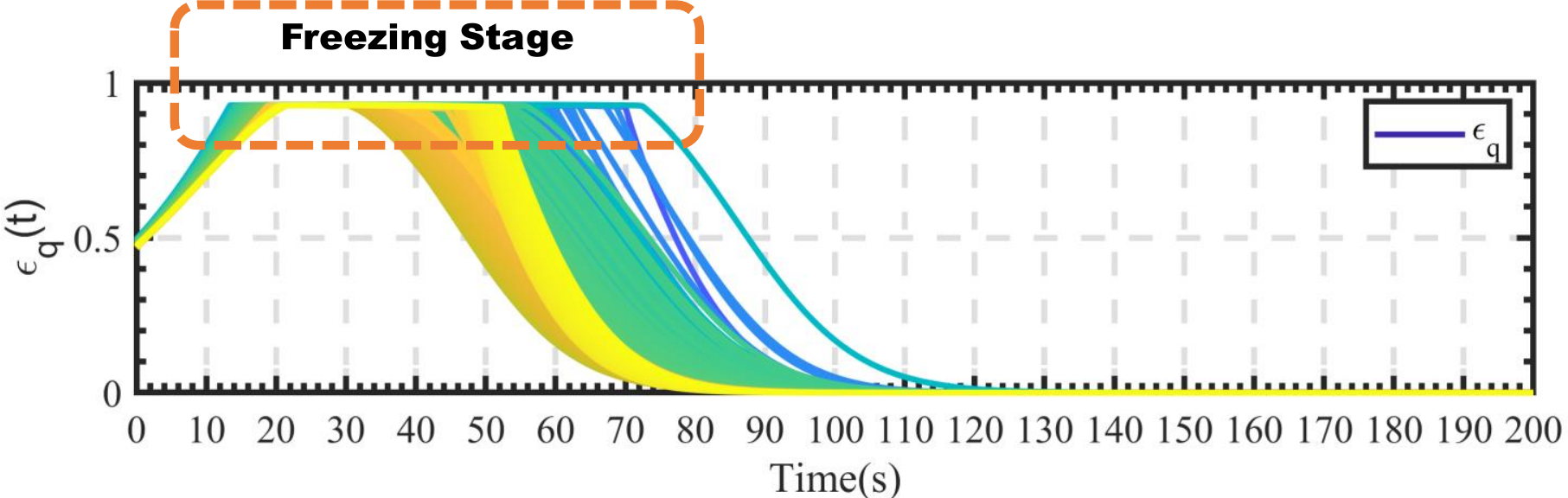}
	\caption{Time-Evolution of the PPC transformed error variable $\varepsilon_{q}(t)$ of all cases (Monte Carlo Simulation)}       
	\label{MonteCarlo4} 
\end{figure}

Next, from Figure \ref{MonteCarlo4}, note that $\varepsilon_{q}(t)<1$ holds for all simulation cases, it can be observed that a maximum value of $\varepsilon_{q}(t)$ is strictly established such that $\varepsilon_{q}(t) < 0.95$ holds for all simulation cases. This validates the effectiveness of the freezing mechanism, and further shows that the PFE constraint is strictly respected. Simultaneously, Figure \ref{MonteCarlo3} reveals that the limitation on the angular velocity $|\omega_{si}(t)|<3^{\circ}/\text{s}$ is strictly respected.

 Consequently, the proposed controller ensures the satisfaction of multiple pointing-forbidden constraints, the attitude angular velocity limitation and the PFE constraint for all simulation cases, with a guaranteed pointing accuracy. This further suggests the effectiveness of the proposed control scheme, and particularly, the effectiveness of the PPC freezing mechanism.

\section{CONCLUSION}
This paper addressed the boresight alignment control problem under parameter uncertainties while considering pointing-forbidden constraints, attitude angular velocity limitations, and pointing accuracy requirements. A composite controller that integrates the Artificial Potential Field (APF) methodology and the Prescribed Performance Control (PPC) scheme was developed to achieve the discussed control objective. Further, the proposed APF-PPC control scheme was incorporated into the Immersion-and-Invariance adaptive structure and efficiently addressed the parameter uncertainty issue.
The proposed controller framework presented a novel Switched Prescribed Performance Function (SPPF), which enables the temporary deactivation of the PPC system such that $\dot{\varepsilon}_{q} = 0$ holds. This ensures that the PPC system has no influence on the system dynamics when safety concerns are conflicted with the performance issue. Simulation results demonstrated that the proposed mechanism effectively eliminated the conflict between safety and the yearning for performance under such circumstance. Furthermore, the proposed scheme satisfied all given constraints while ensuring a guaranteed high control accuracy compared to conventional APF-only methods, even under significant external disturbances.

In essence, the concept of the proposed freezing mechanism suggests that performance concerns should be prioritized only when other higher-priority concerns are not presented, highlighting the need to prioritize safety concerns, which are paramount in mechanical systems such as spacecrafts, quadrotors and robotics. The proposed scheme presents a method of freezing the PPC scheme as needed under appropriate conditions, which may have broader applications in situations where performance requirements conflict with other constraints. Our further investigation may focus on such a direction.

\bibliographystyle{IEEEtran}
\bibliography{apfppcfreezing2}

\begin{IEEEbiography}
	[{\includegraphics[width=1in,height=1.25in,clip,keepaspectratio]{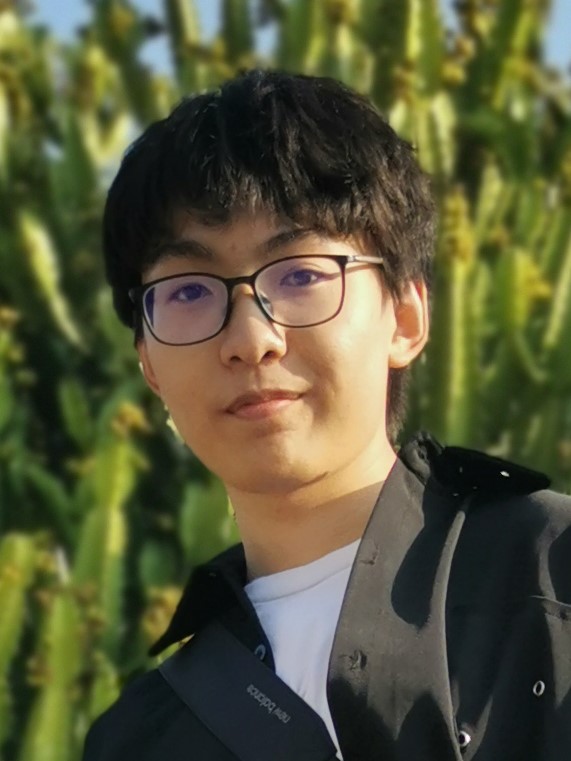}}]
	{Jiakun Lei}{\space} 
	received the B.S. degree in Automatic Control, from the University of Electronic Science and Technology of China(UESTC), Chengdu, China, in 2019. He is working toward a Ph.D. in aeronautical and astronautical science and technology at Zhejiang University, Hangzhou, China.
His research interests include constrained attitude control, nonlinear hybrid control methodology, and attitude control of spacecraft with complex structures.
\end{IEEEbiography}

\begin{IEEEbiography}
	[{\includegraphics[width=1in,height=1.25in,clip,keepaspectratio]{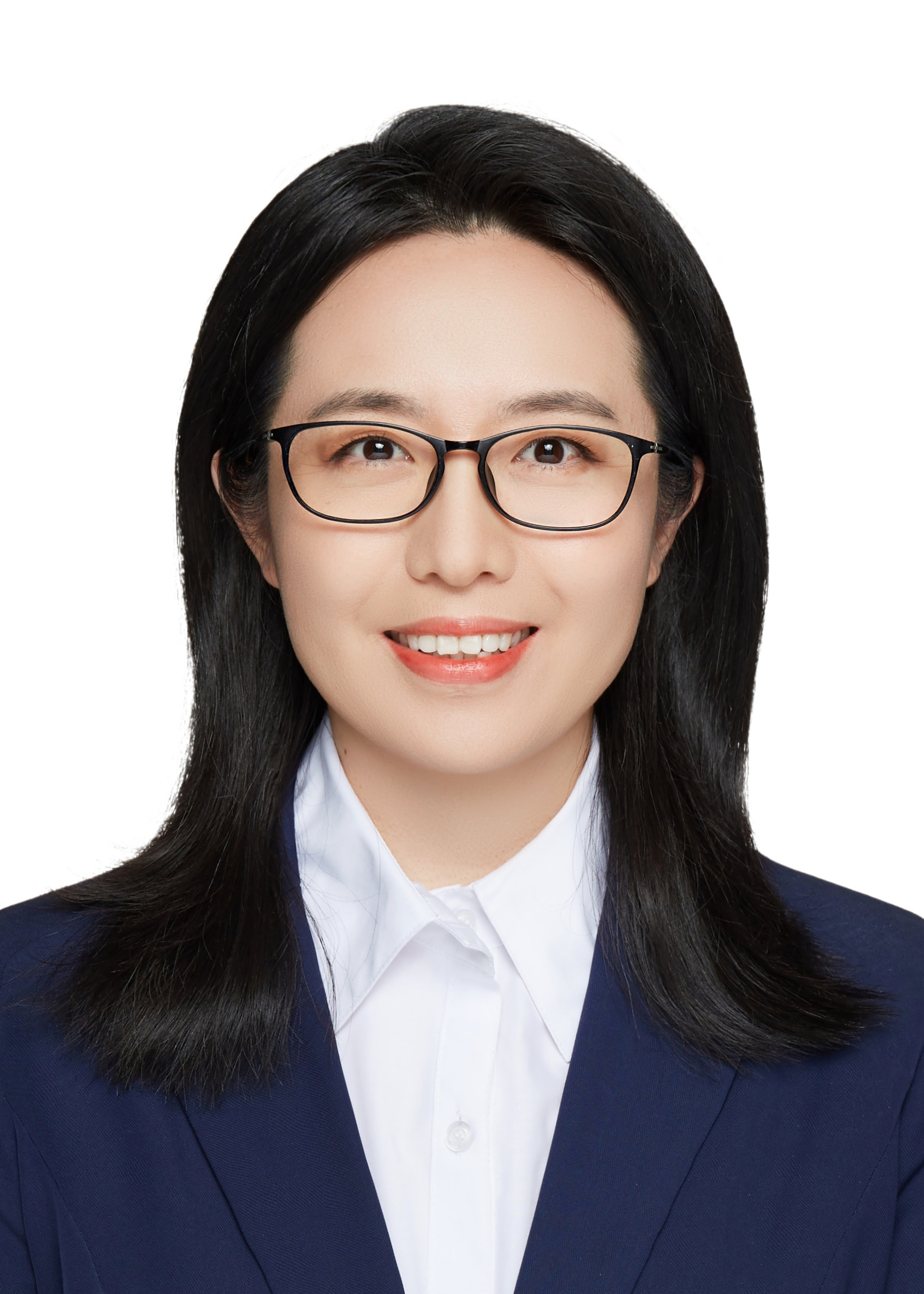}}]
	{Tao Meng}{\space}
	received the B.S. degree in Electronic science and technology, Zhejiang University, Hangzhou, China, in 2004, the M.S. degree in Electronic science and technology, Zhejiang University, Hangzhou, China, in 2006, and the Ph.D. degree in Electronic science and technology, Zhejiang University, Hangzhou, China, in 2009. She is currently a Professor at the School of Aeronautics and Astronautics. Her research interest includes attitude control, orbital control, and constellation formation control of micro-satellite.
\end{IEEEbiography}

\begin{IEEEbiography}
	[{\includegraphics[width=1in,height=1.25in,clip,keepaspectratio]{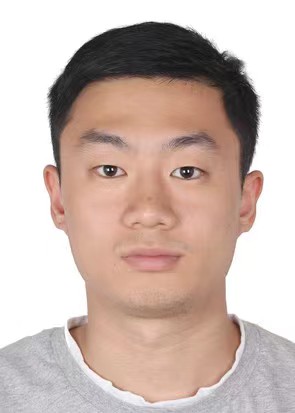}}]
	{Kun Wang}{\space}
received the B.S. degree from college of control science and engineering, Zhejiang University, Hangzhou, China, in 2019. He is currently working toward the Ph.D. degree in aeronautical and astronautical science and technology in Zhejiang University, Hangzhou, China. His research interests include spacecraft 6-DOF control, spacecraft safety critical control and spacecraft formation control. 
\end{IEEEbiography}

\begin{IEEEbiography}
	[{\includegraphics[width=1in,height=1.25in,clip,keepaspectratio]{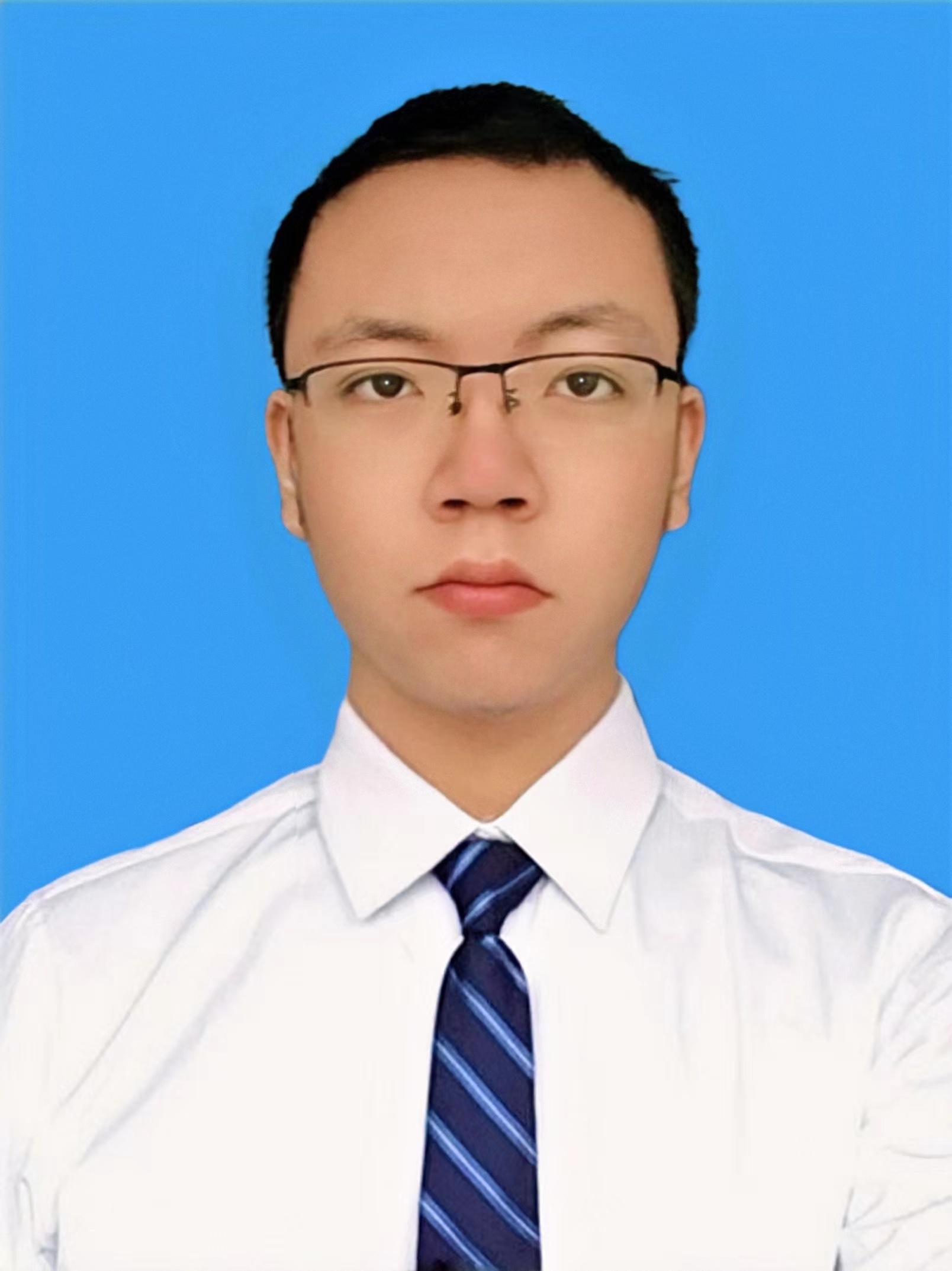}}]
	{Weijia Wang}{\space}
received the B.S. degree in Aerospace Engineering, from the University of Electronic Science and Technology of China in 2020. He is working toward a Ph.D. in aeronautical and astronautical science and technology at Zhejiang University, Hangzhou, China. His research interests include model predictive control and learning-based adaptive control for 6-DOF spacecraft formation.
\end{IEEEbiography}

\begin{IEEEbiography}
	[{\includegraphics[width=1in,height=1.25in,clip,keepaspectratio]{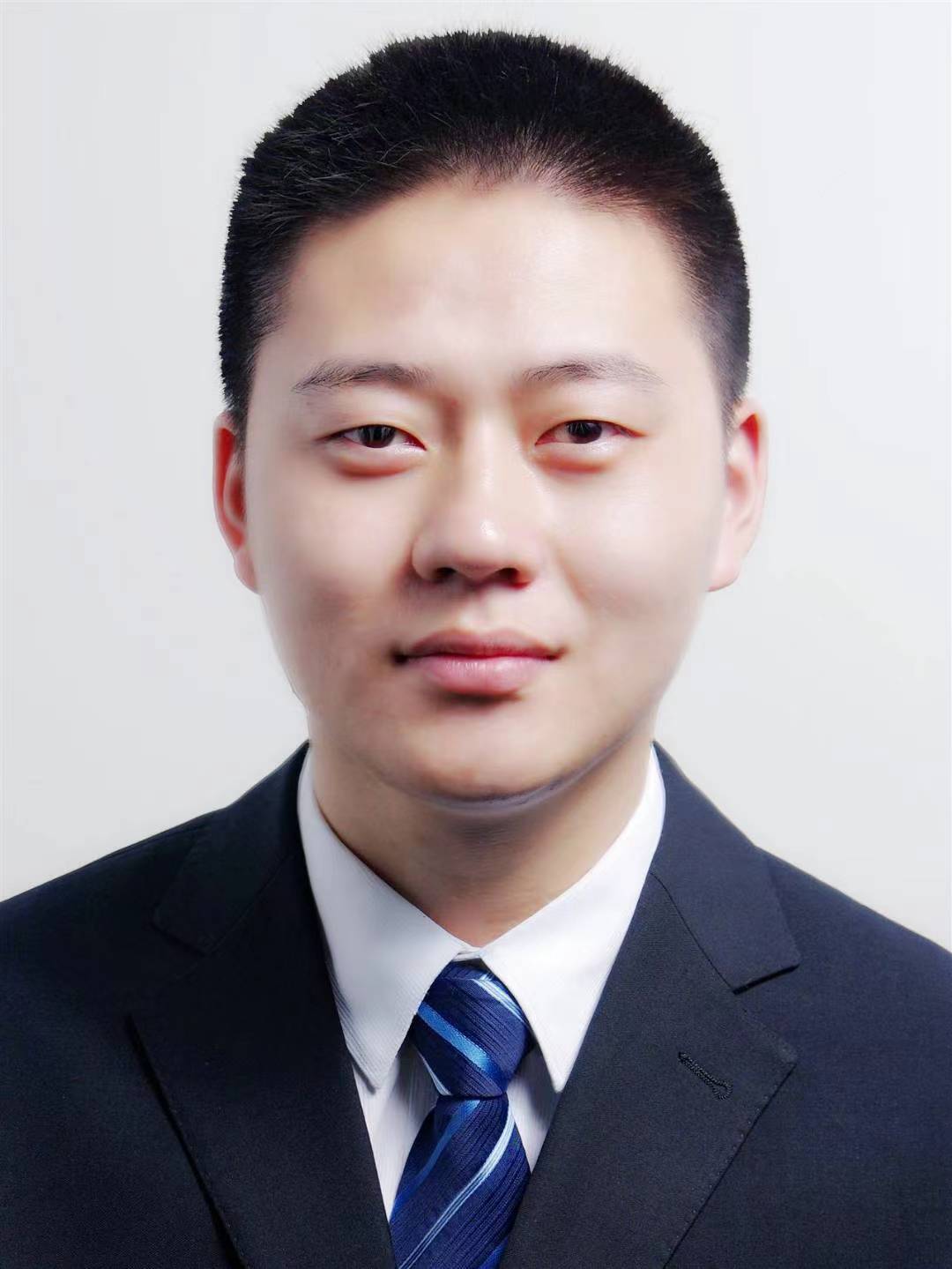}}]
	{Shujian Sun}{\space}
	received the B.S. degree from Electronic Information Engineering (Underwater Acoustic), Harbin Engineering University, Harbin, China, in 2013, and the Ph.D. degree from Electronic Science and Technology, Zhejiang University, Hangzhou, China, in 2020. He is recently the Assistant Professor of School of Aeronautics and Astronautics of Zhejiang University. His research interest include orbit control and formation flying of micro-satellite and micro-propulsion technology.
\end{IEEEbiography}

\begin{IEEEbiography}
	[{\includegraphics[width=1in,height=1.25in,clip,keepaspectratio]{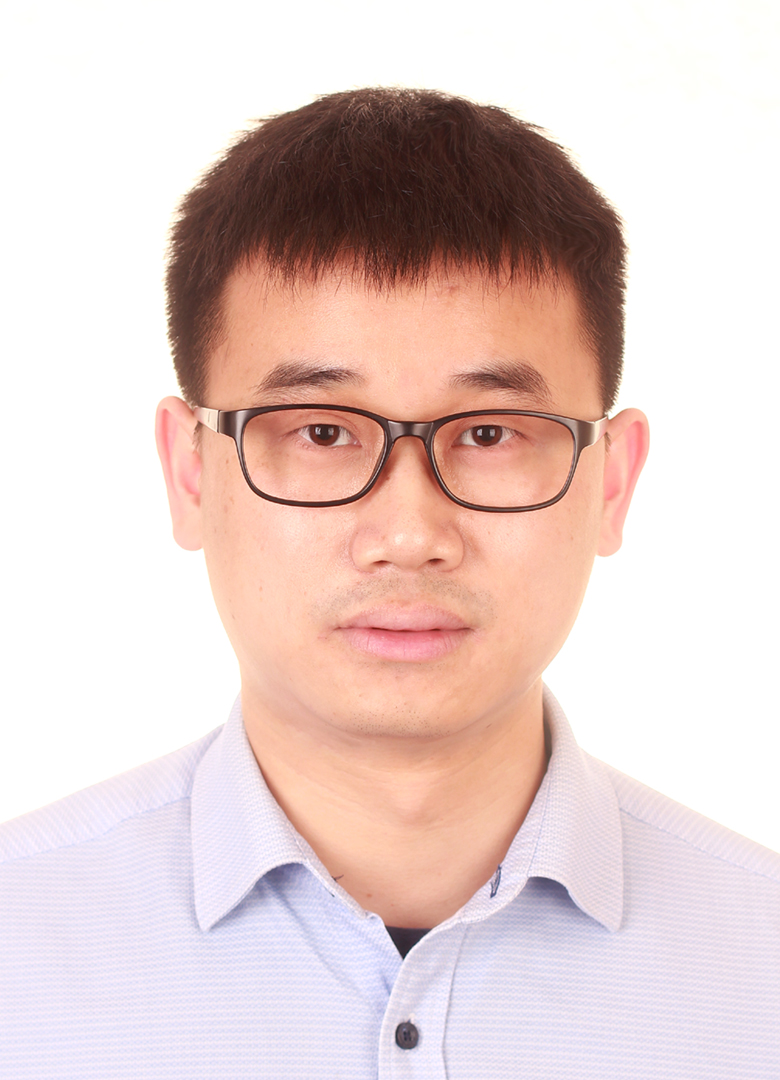}}]
	{Lei Wang}{\space}
received the B.Eng. degree in automation from Wuhan University, China, in 2011, and Ph.D. degree in Control Science and Engineering from Zhejiang University, China in 2016. 
From December 2014 to December 2015, he visited C.A.SY.-DEIS, University of Bologna as a visiting PhD student.

Lei held research positions with School of Electrical and Electronic Engineering at Nanyang Technological
University, Singapore, School of Electrical Engineering and Computing at University of Newcastle, Australia, and Australian Centre for Field Robotics, The University of Sydney, Australia.
Since November 2021 he has been a Hundred-Talent Researcher at College of Control Science and Engineering, Zhejiang University, China.
His current research interests include robust nonlinear estimation and control, distributed computation and learning, privacy analysis and protection, 
with applications to fuel-cell systems, power systems and robotics.
\end{IEEEbiography}

\end{document}